\documentclass{article}
\usepackage{amsmath,inputenc}

\usepackage{graphicx}
\usepackage{csvsimple}
\usepackage{rotating}
\usepackage[letterpaper, portrait, margin=1in]{geometry}
\usepackage{hyperref}
\hypersetup{
    colorlinks=true,
    linkcolor=blue,
    citecolor=red,
    filecolor=magenta,      
    urlcolor=cyan
    }

\usepackage[nocomma]{optidef}
\usepackage{listings}
\usepackage{comment}
\usepackage{appendix}
\usepackage[ruled,vlined]{algorithm2e}

\usepackage{floatrow}
\newfloatcommand{capbtabbox}{table}[][\FBwidth]

\usepackage{subcaption}
\usepackage{booktabs}

\SetAlFnt{\small}
\SetAlCapFnt{\small}
\SetAlCapNameFnt{\small}
\SetAlCapHSkip{0pt}
\IncMargin{-\parindent}

\usepackage{amsfonts} 
\usepackage{amssymb}
\usepackage{bm}
\usepackage[square,comma,numbers]{natbib}

\usepackage{color} 
\definecolor{mygreen}{RGB}{28,172,0} 
\definecolor{mylilas}{RGB}{170,55,241}
\DeclareFixedFont{\ttb}{T1}{txtt}{bx}{n}{12} 
\DeclareFixedFont{\ttm}{T1}{txtt}{m}{n}{12}  
\usepackage{bbm}
\usepackage{amsmath,amsthm}

\newtheorem{theorem}{Theorem}
\newtheorem{corollary}{Corollary}

\newtheorem{lemma}{Lemma}
\newtheorem{observation}{Observation}
\newtheorem{problem}{Problem}

\theoremstyle{definition}
\newtheorem{definition}{Definition}
\newtheorem{remark}{Remark}

\definecolor{mygreen}{RGB}{28,172,0} 
\definecolor{mylilas}{RGB}{170,55,241}

\usepackage[font={small,it}]{caption}

\usepackage[acronym]{glossaries}
\newacronym{otr}{OTR}{Online Traffic Routing}

\newacronym{otr-i}{OTR-I}{Online Traffic Routing with Identical Users}

\usepackage{tikz,pgfplots}
\pgfplotsset{compat=1.14}
\pgfplotsset{scaled y ticks=false}
\pgfplotsset{scaled x ticks=false}

\definecolor{mycolor1}{rgb}{0.00000,0.44700,0.74100}%
\definecolor{mycolor2}{rgb}{0.85000,0.32500,0.09800}%
\definecolor{mycolor3}{rgb}{0.00000,0.44700,0.74100}%
\definecolor{mycolor4}{rgb}{0.00000,0.44700,0.74100}%
\definecolor{mycolor5}{rgb}{0.85000,0.32500,0.09800}%
\definecolor{mycolor6}{rgb}{0.92900,0.69400,0.12500}%
\definecolor{mycolor7}{rgb}{0.49400,0.18400,0.55600}%
\definecolor{mycolor8}{rgb}{0.46600,0.67400,0.18800}%
\definecolor{mycolor9}{rgb}{0.00000,0.44700,0.74100}%
\definecolor{mycolor10}{rgb}{0.85000,0.32500,0.09800}%
\definecolor{mycolor11}{rgb}{0.00000,0.44700,0.74100}%
\definecolor{mycolor12}{rgb}{0.85000,0.32500,0.09800}%

\usepackage{color}
\definecolor{deepblue}{rgb}{0,0,0.5}
\definecolor{deepred}{rgb}{0.6,0,0}
\definecolor{deepgreen}{rgb}{0,0.5,0}

\newcommand{\dd}[1]{\mathrm{d}#1}

\usepackage{listings}

\usepackage[T1]{fontenc}
\usepackage[font=small,labelfont=bf,tableposition=top]{caption}

\DeclareCaptionLabelFormat{andtable}{#1~#2  \&  \tablename~\thetable}

\renewcommand\footnotemark{}

\newcommand{\fakeparagraph}[1]{\vspace{5pt} 

\noindent\textbf{#1}}

\def\fso{$\beta$-SO\xspace}
\def\fpso{$\beta$-PSO\xspace}
\def\objue{h^{UE}}
\def\objso{h^{SO}}

\def\objitap{h^{I_{\alpha}}}

\def\P{\mathcal{P}}

\def\x{\mathbf{x}}

\def\f{\mathbf{f}}

\def\A{\mathcal{A}}

\newif\ifarxiv   
\arxivtrue 

\ifarxiv
\else
\fi

\title{Balancing Fairness and Efficiency in Traffic Routing via Interpolated Traffic Assignment}

\author{Devansh Jalota$^1$, Kiril Solovey$^2$, Matthew Tsao$^1$, Stephen Zoepf$^3$, and Marco Pavone$^1$%
\thanks{$^1$ Stanford University, USA; {\tt \{djalota, kirilsol, mwtsao, pavone\}@stanford.edu}.}%
\thanks{$^3$ Technion - Israel Institute of Technology, Israel;
{\tt kirilsol@stanford.edu}.}
\thanks{$^2$ Lacuna AI, USA; {\tt stephen.zoepf@lacuna.ai}.}%
}

\date{}

\begin{document}

\maketitle

\begin{abstract}
    System optimum (SO) routing, wherein the total travel time of all users is minimized, is a holy grail for transportation authorities. However, SO routing may discriminate against users who incur much larger travel times than others to achieve high system efficiency, i.e., low total travel times. To address the inherent unfairness of SO routing, we study the ${\beta}$-fair SO problem whose goal is to minimize the total travel time while guaranteeing a ${\beta\geq 1}$ level of unfairness, which specifies the maximum possible ratio between the travel times of different users with shared origins and destinations. 

    To obtain feasible solutions to the ${\beta}$-fair SO problem while achieving high system efficiency, we develop a new convex program, the Interpolated Traffic Assignment Problem (I-TAP), which interpolates between a fairness-promoting and an efficiency-promoting traffic-assignment objective. We evaluate the efficacy of I-TAP through theoretical bounds on the total system travel time and level of unfairness in terms of its interpolation parameter, as well as present a numerical comparison between I-TAP and a state-of-the-art algorithm on a range of transportation networks. The numerical results indicate that our approach is faster by several orders of magnitude as compared to the benchmark algorithm, while achieving higher system efficiency for all desirable levels of unfairness. We further leverage the structure of I-TAP to develop two pricing mechanisms to collectively enforce the I-TAP solution in the presence of selfish homogeneous and heterogeneous users, respectively, that independently choose routes to minimize their own travel costs. We mention that this is the first study of pricing in the context of fair routing for general road networks (as opposed to, e.g., parallel road networks).
\end{abstract}

\section{Introduction}

Traffic congestion has soared in major urban centres across the world, leading to widespread environmental pollution and huge economic losses. In the US alone, almost 90 billion US dollars of \ifarxiv economic \fi losses are incurred every year, with commuters losing hundreds of hours due to traffic congestion~\cite{us-congestion-costs}. A contributing factor to increasing road traffic \ifarxiv levels \fi is the often sub-optimal route selection by users due to the lack of centralized \ifarxiv traffic \fi control \cite{how-bad-is-selfish,ROUGHGARDEN2004389}. In particular, \textit{selfish routing}, wherein users choose routes to minimize their travel times, results in a user equilibrium (UE) traffic pattern that is often far from the system optimum (SO) \cite{Sheffi1985,boston-large-poa}. To cope with the efficiency loss due to the selfishness of users, several methods including the control of a fraction of compliant users~\cite{Sharon2018TrafficOF} and marginal cost tolls, where users pay for the externalities they impose on others, have been used to enforce the SO solution as a UE \cite{pigou,wardrop-ue}. 


While determining SO tolls is of fundamental theoretical importance, it is of limited practical interest \cite{METRP} since SO traffic patterns are often unfair with some users incurring much larger travel times than others. This discrepancy among user travel times is referred to as \textit{unfairness}, which, more formally, is the maximum possible ratio across all origin-destination (O-D) pairs of the experienced travel time of a given user to the travel time of the fastest user between the same O-D pair. The unfairness of the SO solution can be quite high in real-world transportation networks, since users may spend nearly twice as much time as others travelling between the same O-D pair \cite{so-routing-seminal}. Moreover, a theoretical analysis established that the SO solution can even have unbounded unfairness~\cite{Roughgarden2002HowUI}.

The lack of consideration of user-specific travel times in the \ifarxiv global \fi SO problem has led to the design of methods that aim to achieve a balance between the total travel time of a traffic assignment and the level of fairness that it provides. In a seminal work, Jahn et al. \cite{so-routing-seminal} introduced the Constrained System Optimum (CSO) to reduce the unfairness of traffic flows by bounding the ratio of the normal length of a path of a given user to the normal length of the shortest path for the same O-D pair. Here, \textit{normal length} is any metric for an edge that is fixed \emph{a priori} and is independent of the traffic flow, e.g., edge length or free-flow travel time. While many approaches to solve the CSO problem have been developed~\cite{ANGELELLI20161,ANGELELLI2020,ANGELELLI2018234}, they suffer from the \ifarxiv inherent \fi limitation that the level of experienced unfairness in terms of user travel times can be much higher than the bound on the ratio of normal lengths that the CSO is guaranteed to satisfy. In addition to this drawback, the algorithmic approaches to solve the CSO problem are often computationally prohibitive and do not provide theoretical guarantees in terms of the resulting solution fairness and efficiency. Furthermore, it is unclear how to develop a pricing scheme to enforce such proposed traffic assignments in practice.

In this work, we study a problem analogous to CSO that differs in the problem's unfairness constraints. In particular, we explicitly consider the experienced unfairness in terms of user travel times as in \cite{ANGELELLI2020105016}, which, arguably, is a more accurate representation of user constraints as it accounts for costs that vary according to a traffic assignment. Our work further addresses the algorithmic concerns of existing approaches to solve fairness-constrained traffic routing problems by developing (i) a computationally-efficient approach that trades off efficiency and fairness in traffic routing, (ii) theoretical bounds to quantify the performance of our algorithm, and (iii) a pricing mechanism to enforce the resulting traffic assignment.

\fakeparagraph{Contributions.}
We study the $\beta$-fair System Optimum (\fso) problem, which involves minimizing the total travel time of users subject to unfairness constraints, where a $\beta\geq 1$ bound on unfairness specifies the maximum allowable ratio between the travel times of different users with shared origins and destinations. 

We develop a simple yet effective 
approach for \fso that involves solving a new convex program, the interpolated traffic assignment problem (I-TAP). I-TAP interpolates between the fair UE and efficient SO objectives to achieve a solution that is simultaneously fair and efficient. This allows us to approximate the \fso problem as an unconstrained traffic assignment problem, which can be solved quickly. We further present theoretical bounds on the total system travel time and unfairness level in terms of the interpolation parameter of I-TAP.


We then exploit the structure of I-TAP to develop two pricing schemes which enforce users to selfishly select the flows satisfying the $\beta$ bound on unfairness computed through our approach. \ifarxiv In particular, for \else For \fi homogeneous users with the same value of time we develop a natural marginal-cost pricing scheme. For heterogeneous users, we exploit a linear-programming method \cite{multicommodity-extension} \ifarxiv to derive prices that enforce the optimal flows computed through I-TAP\else\fi. We mention that our work is the first to study road pricing in connection with fair routing for general road networks as opposed to, e.g., parallel networks.

Finally, we evaluate the performance of our approach on real-world transportation networks for the above (Section~\ref{sec:numerical-experiments}) and other valid notions of unfairness \ifarxiv (Appendix~\ref{apdx:general-unfairness}) \else (Appendix F) \fi. The numerical results indicate significant computational savings as well as superior performance for all desirable levels of unfairness~$\beta$, as compared to the algorithm in~\cite{so-routing-seminal}. Moreover, our results demonstrate that our approach can reduce \ifarxiv the level of \fi unfairness by 50\% while increasing the total travel time by at most 2\%, which indicates that a huge gain in user fairness can be achieved for a small loss in 
efficiency, making our approach a desirable option for use in route guidance systems.

This paper is organized as follows. Section~\ref{sec:related-lit} reviews related literature. We introduce in Section~\ref{sec:model} the \fso problem and metrics to evaluate the fairness and efficiency of a \ifarxiv feasible \fi traffic assignment. We introduce the I-TAP method and discuss its properties in Section~\ref{sec:itap-main}, and develop pricing schemes \ifarxiv to enforce the computed traffic assignments as a UE \fi in Section~\ref{sec:main-pricing}. We evaluate the performance of the I-TAP method 
through numerical experiments in Section~\ref{sec:numerical-experiments} and provide directions for future work in Section~\ref{sec:future-work}. \ifarxiv Furthermore, in the appendix we also present an extension of our work to additional fairness notions beyond the one considered in the main text.

\ifarxiv
\else

In the appendix, we provide proofs that were omitted from the main text and discuss additional results on the sensitivity analysis as well as numerical implementation details of our developed approach. Furthermore, we elaborate on some additional observations on our pricing results and present an extension of our work to additional fairness notions beyond the one considered in the main text.


\fi

\section{Related Work} \label{sec:related-lit}

The trade-off between system efficiency and user fairness has been widely studied in applications including resource allocation, reducing the bias of machine-learning algorithms, and influence maximization. While different notions of fairness have been proposed, 
the level of fairness is typically controlled through the problem's objective or constraints. For instance, fairness parameters that trade-off the level of fairness in the objective can be tuned to investigate the loss in system efficiency in the context of influence maximization \cite{Rahmattalabi2020FairIM} and resource allocation~\cite{Bertsimas-Eff-Fairness} problems. On the other hand, fairness parameters that bound the degree of allowable inequality between different user groups through the problem's constraints have been proposed to reduce bias towards disadvantaged groups \cite{stoica-seeding}\ifarxiv. For example, group-based fairness notions \cite{yiling-fair-classification} have been studied to reduce the bias of machine-learning algorithms, and diversity constraints have been introduced to ensure that the benefits of social interventions are fairly distributed \cite{Tsang2019GroupFairnessII}. \else, e.g., through group-based or diversity constraints~\cite{yiling-fair-classification,Tsang2019GroupFairnessII}. \fi

In the context of traffic routing, several traffic assignment formulations have been proposed to achieve a balance between multiple performance criteria~\cite{DIAL1996,DIAL1997357}, with a particular focus on fairness considerations in traffic routing~\cite{so-routing-seminal}. Since Jahn et al. \cite{so-routing-seminal} introduced the CSO problem, there have been both theoretical studies~\cite{so-routing-user-constraints} as well as the development of heuristic approaches to solve the NP-hard CSO problem. For instance,~\cite{so-routing-seminal} proposed a Frank-Wolfe based heuristic, \ifarxiv wherein the solution to the linearized CSO problem is obtained by solving a constrained shortest-path problem, \fi \ifarxiv while another work~\cite{BAYRAM2015146} developed a second-order cone programming technique. Several subsequent approaches for CSO have considered linear relaxations of the original problem~\cite{ANGELELLI20161,ANGELELLI2020,ANGELELLI2018234}. \else while others have considered linear relaxations of the original problem~\cite{ANGELELLI20161,ANGELELLI2018234,ANGELELLI2020}. \fi Each of these approaches bounds the level of unfairness in terms of normal lengths of paths by restricting the set of eligible paths on which users can travel to those that meet a specified level of normal unfairness. However, the experienced unfairness in terms of the travel times \ifarxiv on the restricted path set \fi may be much higher than the \ifarxiv specified \fi level of normal unfairness, which is an \emph{a priori} fixed quantity. 

This inherent drawback of the CSO problem in limiting the experienced unfairness in terms of user travel times was overcome by~\cite{ANGELELLI2020105016}, which proposed two Mixed Integer Non-Linear Programming models to capture traffic-dependent notions of unfairness. Their approach to solve these models relies on a linearization heuristic for the edge travel-time functions, which are in general \emph{non-linear}. Achieving a high level of accuracy of the linear relaxations in approximating the true travel-time functions, however, requires solving a large MILP which is computationally \emph{expensive}. Unlike~\cite{ANGELELLI2020105016}, our I-TAP method is computationally \emph{inexpensive}, while directly accounting for non-linear travel-time functions.

A further limitation of the existing methods for fair traffic routing is that there are limited results in providing pricing schemes to induce selfish users to collectively form the proposed traffic patterns, e.g., those satisfying a certain bound on unfairness. For instance,~\cite{marden-parallel} provides tolling mechanisms to enforce fairness constrained flows which applies only \ifarxiv to the simplified model of a parallel network\else to parallel networks\fi. In more general networks,~\cite{lujak-ue-so} proposes an auction-based bidding mechanism for users to be assigned to precomputed paths. However, this approach cannot be applied as-is to our setting as users are unconstrained with respect to a specific path set. \ifarxiv Thus, we develop road tolling mechanisms to induce selfish users to collectively form traffic assignments guaranteeing a specified level of unfairness. \fi

\section{Model and Problem Definition} \label{sec:model}
We model the road network as a directed graph $G = (V, E)$, with the vertex and edge sets $V$ and $E$, respectively. Each edge $e \in E$ has a normal length $\eta_e$, 
\ifarxiv 
which represents a fixed quantity such as the physical length of the corresponding road segment, \fi 
and a flow-dependent travel-time function $t_e: \mathbb{R}_{\geq 0} \rightarrow \mathbb{R}_{> 0}$, which maps $x_e$, the rate of traffic on edge $e$, to the travel time $t_e(x_e)$. 
As is standard in the traffic routing literature, we
assume that the function~$t_e$, for each $e \in E$, is differentiable, convex, locally Lipschitz continuous, and monotonically increasing.

Users make trips between a set of O-D pairs, and we model users with the same origin and destination as one commodity, where $K$ is the set of all commodities. Each commodity $k \in K$ has a demand rate $d_k>0$, which represents the amount of flow to be routed on a set of directed paths $\mathcal{P}_k$ between its origin and destination. 
\ifarxiv 
Here $\mathcal{P}_k$ is the set of all possible paths between the origin and destination corresponding to commodity $k$. 
\fi 
The edge flow of each commodity $k$ is given by $\x^k = \{x_e^k\}_{e\in E}$, while the aggregate edge flow of all commodities is denoted as $\mathbf{x}:=\{x_e\}_{e\in E}$. For an edge flow $\mathbf{x}:=\{x_e\}_{e\in E}$ and a path $P \in \P = \cup_{k \in K} \mathcal{P}_k$, the amount of flow routed on the path is denoted as $\mathbf{x}_P$, where the vector of path flows $\f = \{\mathbf{x}_P: P \in \P \}$. Then, the travel time on path $P$ is $t_P(\mathbf{x}) = \sum_{e \in P} t_e(x_e)$, while $\eta_P = \sum_{e \in P} \eta_e$ is its normal length. 




We assume users are selfish and thus choose paths that minimize their total travel cost that is a linear function of tolls and travel time. For a value of time parameter $v>0$, and a vector of edge prices (or tolls) $\boldsymbol{\tau} = \{\tau_e\}_{e \in E}$, the travel cost on a given path $P$ under the traffic assignment $\x$ is given by $C_{P}(\x, \boldsymbol{\tau}) = v t_P(\mathbf{x}) + \sum_{e \in P} \tau_e$.

\subsection{Traffic Assignment}
In this work we will consider several variants of the traffic assignment problem (TAP). The goal of the SO traffic assignment problem (SO-TAP) is to route users to minimize the total system travel time. This behavior is captured in the following convex program:

\begin{definition}[Program for SO-TAP \cite{Sheffi1985}] \label{def:so-tap}
\begin{mini!}|s|[2]<b>
	{\f }{\objso(\mathbf{x}):=\sum_{e \in E} x_e t_e(x_e), \label{eq:so-obj}} 
	{}
	{}
	\addConstraint{\sum_{k \in K} \sum_{P \in \mathcal{P}_k: e \in P} \mathbf{x}_P}{= x_e, \quad \forall e \in E \label{eq:edge-constraint},}
	\addConstraint{\sum_{P \in \mathcal{P}_k} \mathbf{x}_P}{= d_k, \quad \forall k \in K \label{eq:demand-constraint},}
	\addConstraint{\mathbf{x}_P}{ \geq 0, \quad \forall P \in \P \label{eq:nonnegativity-constraints},}
\end{mini!}
with edge flow Constraints~\eqref{eq:edge-constraint}, demand Constraints~\eqref{eq:demand-constraint}, and non-negativity
Constraints~\eqref{eq:nonnegativity-constraints}.
\end{definition}
We mention that the total travel time objective is only a function of the aggregate edge flow $\x$, which is related to the path flow $\f$ through Constraint~\eqref{eq:edge-constraint}. Note for any given path flow $\f$ that both the edge flow $\x$ and the commodity-specific edge flows $\x^k$ for each commodity $k \in K$ are uniquely defined.


\ifarxiv
Closely related to SO-TAP is the UE traffic assignment problem (UE-TAP) that emerges from the selfish behavior of users, where each user strives to minimize its own travel time, without regard to the effect that it has on the overall travel time of all the users in the system. This behavior is captured through the following convex program:
\else
Closely related is the UE traffic assignment problem (UE-TAP) that emerges from the selfish behavior of users that minimize their own travel time, and is described by the following convex program:
\fi
\begin{definition}[Program for UE-TAP \cite{Sheffi1985}]
\begin{mini!}|s|[2]<b>
	{\f}{\objue(\mathbf{x}):=\sum_{e\in E}\int_{0}^{x_e} t_e(y) \dd{y}, \label{eq:ue-obj}}
	{}
	{}
	\addConstraint{~\eqref{eq:edge-constraint}}{-\eqref{eq:nonnegativity-constraints}.}
\end{mini!}
\end{definition}

While the integral objective used to define UE-TAP has not found a clear economic or behavioral interpretation within the transportation and game-theory communities~\cite{Sheffi1985}, the optimal solution of UE-TAP corresponds to an equilibrium, which can be seen through the KKT conditions of this optimization problem. That is, UE-TAP provides a polynomial time computable method to determine the user equilibrium flows. A defining property of the UE solution is that it is fair for all users since the travel time of all the flow that is routed between the same O-D pair is equal. In contrast, at the SO solution the sum of the travel time and marginal cost of travel is the same for all users travelling between the same O-D pair. Thus, marginal cost pricing is used to induce selfish users to collectively form the SO traffic pattern. While the number of constraints, which depend on the path sets $\mathcal{P}_k$, can be exponential in the size of the transportation network, both SO-TAP and UE-TAP are efficiently computable since they can be formulated without explicitly enumerating all the path level flows and constraints~\cite{Sheffi1985}.

\subsection{Fairness and Efficiency Metrics} \label{sec:fair-eff-metric}
\ifarxiv
We evaluate the quality of any given traffic assignment $\x$, which satisfies Constraints~\eqref{eq:edge-constraint}-\eqref{eq:nonnegativity-constraints}, using two metrics, namely: (i) efficiency, which is of importance to a traffic authority as well as all users collectively, and (ii) fairness, which is of direct importance to each user.
\else
We evaluate the quality of any traffic assignment $\x$ using two metrics, namely: (i) efficiency and (ii) fairness.
\fi

We evaluate the efficiency of a traffic assignment by comparing its total travel time to that of the SO edge flow $\mathbf{x}^{SO}$. Recalling that $\objso(\mathbf{x})$ denotes the total travel time of the edge flow $\mathbf{x}$, 
the \textit{inefficiency ratio} of $\mathbf{x}$ is
\ifarxiv 
\begin{equation}
    \rho(\mathbf{x}) := \frac{\objso(\mathbf{x})}{\objso(\mathbf{x}^{SO})}.
\end{equation}
\else
$\rho(\mathbf{x}) := \frac{\objso(\mathbf{x})}{\objso(\mathbf{x}^{SO})}$.
\fi 
Note that for the UE solution $\mathbf{x}^{UE}$, the inefficiency ratio is the Price of Anarchy  
\ifarxiv 
(PoA)~\cite{worst-case-equilibria}, which we denote as~$\Bar{\rho}$.
\else
(PoA)~\cite{worst-case-equilibria}.
\fi

To evaluate the fairness of a traffic assignment, we first introduce the notion of a \textit{positive path} from \cite{basu2017reconciling}.
\begin{definition} [Positive Path]
For any path flow $\f$ with corresponding commodity-specific edge flows $\x^k$, a path $P \in \mathcal{P}_k$ is positive for a commodity $k \in K$ if for all edges $e \in P$, $x_e^k$ is strictly positive. The set of all positive paths for a flow $\f$ and commodity $k$ is denoted as $\mathcal{P}_{k}^{+}(\f) = \{P:P \in \mathcal{P}_k, x_e^k>0, \text{ for all } e \in P \}$.
\end{definition}
The importance of the notion of a positive path is that the path decomposition of the commodity-specific edge flows $\x^k$ may be non-unique; however the set of positive paths is always uniquely defined for such edge flows. That is, for commodity specific edge flows $\x^k$ the set of positive paths for any two path decompositions $\f_1$ and $\f_2$ are equal, i.e., $\mathcal{P}_{k}^{+}(\f_1) = \mathcal{P}_{k}^{+}(\f_2)$.

We evaluate the fairness of a traffic flow $\f$ with an edge decomposition $\x$ 
through its corresponding unfairness $U$, which is defined as the maximum ratio across all O-D pairs of (i) the 
travel time on the slowest, i.e., highest travel time, positive path to (ii) the travel time on the fastest positive path between the same O-D pair\ifarxiv:
\begin{align*}
    U(\f) := \max_{k \in K} \max_{Q, R \in \P_k^{+}} \frac{t_Q(\x)}{t_R(\x)}
    .
\end{align*}
\else
, i.e., $U(\f) := \max_{k \in K}  \max_{Q, R \in \P_k^{+}} \frac{t_Q(\x)}{t_R(\x)}$.
\fi
That is, $U(\f)$ returns the maximum possible ratio of travel times on positive paths across all commodities with respect to the path flow $\f$. As a result, the unfairness $U$ is a number between one and infinity, and a traffic assignment has a high level of fairness if its unfairness is close to one while it has a low level of fairness if the corresponding unfairness is much larger than one. \ifarxiv A discussion on numerically computing the above defined notion of unfairness is presented in Section~\ref{sec:numerical-experiments}. In contrast, other valid notions of unfairness 
could also be considered. For instance, for a given path flow decomposition $\f$ with a corresponding edge flow $\x$, the unfairness $\Tilde{U}(\cdot)$ of the path flows can be evaluated as the maximum ratio between the travel times of any two users travelling between the same O-D pair, i.e., $\Tilde{U}(\f) = \max_{k \in K}  \max_{Q, R \in \P_k: \x_Q, \x_R>0} \frac{t_Q(\x)}{t_R(\x)} $. Note here that we only consider a ratio of travel times on paths with strictly positive flow for the path decomposition $\f$ rather than the ratio of travel times for all positive paths. We defer a detailed treatment of path-based unfairness measures to \ifarxiv Appendix~\ref{apdx:general-unfairness} \else Appendix F \fi and highlight here some key features of the positive path based unfairness measure $U$. 


The unfairness measure $U(\f)$ can be efficiently computed and has the benefit that it applies to all possible path decompositions of the commodity specific edge flows $\x^k$. As a result, in the context of a single O-D pair travel demand, the unfairness measure $U(\f)$ has the benefit that it is a property of the unique edge flow $\x$ and is relevant when users are not constrained to a specific path decomposition, as happens in practice. In contrast, path decomposition specific unfairness measures, e.g., $\Tilde{U}(\f)$, are likely to be more sensitive to the method used to compute the path decomposition. Furthermore, we note that the positive path based unfairness notion serves as an upper bound on the ratio of travel times for any two users travelling between the same O-D pair for the path flow $\f$, i.e., $\Tilde{U}(\f) \leq U(\f)$ for all $\f$. As a result, our theoretical bounds on unfairness obtained for the positive path-based unfairness notion will naturally extend to path decomposition specific unfairness measures such as $\Tilde{U}(\f)$. Thus, in the rest of this paper we focus on the positive path-based unfairness measure and, for numerical comparison, we present other path decomposition specific unfairness measures, e.g., $\Tilde{U}(\f)$, in Appendix~\ref{apdx:general-unfairness}.

\subsection{Toy Network Example} \label{sec:toy-network}

To illustrate the fairness and efficiency properties of the two optimization problems, SO-TAP and UE-TAP, studied in this work, we present a toy example of a two-edge Pigou network, as depicted in Figure~\ref{fig:toy}. In particular, consider a demand of one that needs to be routed from the origin $v_1$ to the destination $v_2$, with two edges ($e_1$ and $e_2$) connecting the origin to the destination. Observe that if the travel time functions on the two edges are given by $t_1(x_1) = 1$ and $t_2(x_2) = x_2$, then under the UE-TAP solution all users will be routed on edge two, while the SO-TAP solution that minimizes the total travel time will route $0.5$ units of flow on both edges. The level of unfairness and the total travel time of the two traffic assignments are presented in the following table, which indicates that the UE-TAP solution is fair but inefficient while the SO-TAP solution is efficient but unfair.


\begin{figure}[!h] 
    \begin{minipage}{0.4\linewidth}
		\centering
		\includegraphics[width=60mm]{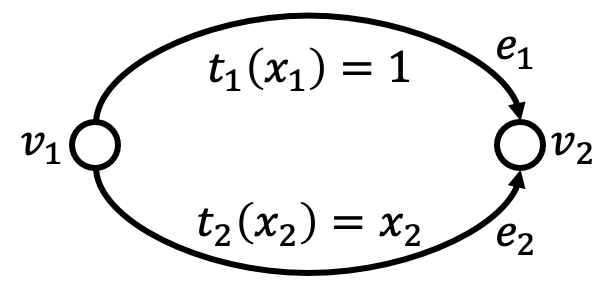}
	\end{minipage} \hspace{15pt}
	\begin{minipage}{0.5\linewidth}
		\small
		\label{table:student}
		\centering
		\begin{tabular}[b]{lcc}\hline
      & UE-TAP & SO-TAP  \\ \hline
      Total Travel Time & 1 & $3/4$ \\
      Unfairness & 1 & 2 \\ \hline
    \end{tabular}
	\end{minipage}\hfill
	\caption{A two-edge Pigou network example to illustrate the fairness and efficiency properties of the SO-TAP and UE-TAP optimization problems. For a demand of one, the UE-TAP solution routes all the flow on edge $e_2$ resulting in a fair solution but a total travel time of one. On the other hand, the SO-TAP solution routes $0.5$ units of flow on both edges, resulting in an efficient solution with the minimum total travel time, but an unfair solution with an unfairness of two.} 
	\label{fig:toy}
  \end{figure}

\subsection{$\beta$-Fair System Optimum} \label{sec:beta-so}

\ifarxiv
Our focus in this work is in solving the following problem where we impose an upper bound on the maximum allowable level of unfairness in the network. In particular, for an unfairness parameter $\beta \in [1, \infty)$, any feasible path flow $\f$ must satisfy $U(\f) \leq \beta$. To trade-off between user fairness and system efficiency, we consider the following $\beta$-fair System Optimum (\fso) problem.
\else
To trade-off between user fairness and system efficiency, we consider the following $\beta$-fair System Optimum (\fso) problem, wherein we impose an upper bound $\beta \in [1, \infty)$ on the maximum allowable unfairness in the network, i.e., $U(\mathbf{x}) \leq \beta$ for a traffic assignment $\x$.
\fi

\begin{definition}[Program for $\beta$-Fair System Optimum]
\begin{mini!}|s|[2]<b>
	{\f }{\sum_{e \in E} x_e t_e(x_e), \label{eq:efficient-Obj}}
	{\label{eq:Example3}}
	{}
	\addConstraint{~\eqref{eq:edge-constraint}}{-\eqref{eq:nonnegativity-constraints},}
	\addConstraint{U(\f)}{ \leq \beta \label{eq:unfairness-constraints}.}
\end{mini!}
\end{definition}

Note that without the unfairness Constraints~\eqref{eq:unfairness-constraints} (or when $\beta=\infty$), the above problem exactly coincides with SO-TAP. Furthermore, the \fso problem is always feasible for any $\beta \in [1, \infty)$, since a solution to UE-TAP exists and achieves an unfairness of $\beta = 1$.

We also note that the difference between the \fso and CSO problems is in the unfairness Constraints~\eqref{eq:unfairness-constraints}. While the \fso problem explicitly imposes an upper limit on the ratio of travel times on positive paths, the CSO problem imposes normal unfairness constraints for each path $P \in \mathcal{P}_k$ and any commodity $k \in K$ of the form $\eta_P \leq \phi \min_{P^* \in \mathcal{P}_k} \eta_{P^*}$ for some normal unfairness parameter $\phi\geq 1$. That is, the CSO problem minimizes the total travel time subject to flow conservation constraints over the set of paths with a normal unfairness level of at most~$\phi$. \ifarxiv The authors of \cite{so-routing-seminal} use normal unfairness, which is a fixed quantity, as a proxy to limit the maximum possible ratio of user travel times. Note that the experienced user travel times accounts for costs that vary according to a traffic assignment unlike normal unfairness. \else The authors of \cite{so-routing-seminal} use normal unfairness, which is a fixed quantity, as a proxy to limit the ratio of user travel times, which vary according to a traffic assignment. \fi

The optimal solution of the \fso problem corresponds to the highest achievable system efficiency whilst meeting unfairness constraints. However, solving \fso directly is generally intractable as the unfairness Constraints~\eqref{eq:unfairness-constraints} are non-convex if the travel time function is non-linear. Moreover, since the unfairness metric studied in this work accounts for user costs that vary according to a traffic assignment, unlike normal unfairness that is an apriori fixed quantity, the NP-hardness of the CSO problem \cite{so-routing-seminal} suggests the computational hardness of \fso \cite{basu2017reconciling,ANGELELLI2020105016}.

Finally, we mention that we consider a setting wherein the travel demand is time invariant and fractional user flows are allowed, as is standard in the traffic routing literature. Also, for notational simplicity, we consider for now a model where all users are homogeneous, i.e., they have an identical value of time $v$, and present an extension of our pricing result to the setting of heterogeneous users in Section~\ref{sec:pricing-heterogeneous}. 

\section{A Method for $\beta$-Fair System Optimum} \label{sec:itap-main}
In this section, we develop a computationally-efficient method for solving \fso with edge-based unfairness constraints, to achieve a traffic assignment with a low total travel time, whose level of unfairness is at most $\beta$. In particular, we propose a new formulation of TAP, which we term interpolated TAP (or I-TAP), wherein the objective function linearly interpolates between the objectives of UE-TAP and SO-TAP. Our main insight is that the UE solution achieves a high level of fairness, whereas the SO solution achieves a low total travel time, and we wish to get the best of both worlds---high level of fairness at a low total travel time.


In this section, we describe the I-TAP method and evaluate its efficacy for the \fso problem by addressing three key concerns regarding the solution efficiency, feasibility and computational tractability. In particular, we establish a relationship between I-TAP and \fso through theoretical bounds on the inefficiency ratio (Section~\ref{sec:eff-properties}) and its optimality for two-edge Pigou networks (Section~\ref{sec:optimality}). We also establish the feasibility of I-TAP for \fso by finding the range of values of the interpolation parameter such that the unfairness of the optimal I-TAP solution is guaranteed to be less than $\beta$ (Section~\ref{sec:eff-properties}). Finally, we present an equivalence between I-TAP and UE-TAP to show that I-TAP can be computed efficiently (Section~\ref{sec:compTractability}). These results indicate that we can approximate \fso as an unconstrained traffic assignment problem that can be solved quickly. We also mention that we perform a sensitivity analysis to establish the continuity of the optimal traffic assignment and its total travel time in the interpolation parameter of I-TAP in \ifarxiv Section~\ref{sec:solution-properties}\else Appendix B\fi.




\subsection{Interpolated Traffic Assignment} \label{sec:algos}

We provide a formal definition of interpolated TAP:

\begin{definition} [I-TAP] \label{def:interpolated-TAP}
For a convex combination parameter $\alpha \in [0, 1]$, the interpolated traffic assignment problem, denoted as I-TAP$_{\alpha}$, is given by:
\begin{mini!}|s|[2]<b>
	{\f}{\objitap(\x) := \alpha \objso(\x) + (1-\alpha) \objue(\x), \label{eq:interpolated-TAP-Obj}}
	{}
	{}
	\addConstraint{~\eqref{eq:edge-constraint}}{-\eqref{eq:nonnegativity-constraints}. \label{eq:interpolated-TAP-all-constraints}}
\end{mini!}
\end{definition}

A few comments are in order. First, it is clear that I-TAP$_0$ and I-TAP$_1$ correspond to UE-TAP and SO-TAP, respectively. Next, under the assumption that the travel time functions are strictly convex, we observe that for any $\alpha\in [0,1]$ the program I-TAP$_\alpha$ 
\ifarxiv 
has a strictly convex objective with linear constraints, and so I-TAP$_\alpha$ 
\fi 
is a convex optimization problem with a unique edge flow solution~$\x(\alpha)$.

For numerical implementation purposes, we propose a \emph{dense sampling} procedure to compute a solution for \fso with a low total travel time while guaranteeing a $\beta$ bound on unfairness.
\ifarxiv
\paragraph{Algorithm for Computation of Optimal Interpolation Parameter} 
To compute a good solution for \fso, we evaluate the optimal solution $\f(\alpha)$ of I-TAP$_\alpha$ (with corresponding edge flows $\x(\alpha)$) for $\alpha$ taken from a finite set $\A_s:=\{0,s,2s,\ldots,1\}$ for some step size $s\in (0,1)$. That is, in $O(\frac{1}{s})$ computations of I-TAP, we can return the path flow $\f(\alpha^*)$ (with unique edge decomposition $\mathbf{x}(\alpha^*)$), for some 
$\alpha^* \in \A_s$, with the lowest total travel time that is at most $\beta$-unfair, i.e., $U(\mathbf{f}(\alpha^*)) \leq \beta$, and the value $\objso(\x(\alpha^*))$ is minimized.
\else
In particular, to compute a good solution for \fso, we evaluate the optimal solution $\x(\alpha)$ of I-TAP$_\alpha$ for $\alpha$ taken from a finite set $\A_s:=\{0,s,2s,\ldots,1\}$ for some step size $s\in (0,1)$. That is, in $O(\frac{1}{s})$ computations of I-TAP, we can return the traffic assignment $\mathbf{x}(\alpha^*)$, for some 
$\alpha^* \in \A_s$, with the lowest total travel time that is at most $\beta$-unfair, i.e.,  $ U(\mathbf{x}(\alpha^*)) \leq \beta$, and the value $\objso(\x(\alpha^*))$ is minimized.
\fi

We observe experimentally (Section~\ref{sec:numerical-results}) that this method of computing the I-TAP solution \ifarxiv for a finite set of convex combination parameters \fi achieves a good solution for \fso in terms of fairness and total travel time. We note here that our approach also naturally extends to other unfairness notions wherein the user equilibrium achieves the highest possible level of fairness, while the system optimum achieves the lowest total travel times (see \ifarxiv Appendix~\ref{apdx:general-unfairness}\else Appendix F\fi). Finally, we restrict $\alpha$ to lie in the finite set $\A_s$ since the exact functional form of the optimal solution $\f(\alpha)$ (with edge flow $\x(\alpha)$), and thus the unfairness $U(\f(\alpha))$ and the total travel time $\objso(\f(\alpha))$ functions, in $\alpha$ is not directly known, though we show that $\x(\alpha)$ and $\objso(\x(\alpha))$ are continuous 
in $\alpha$ in Section~\ref{sec:solution-properties}.
\ifarxiv \else Finally, we restrict $\alpha$ to lie in the finite set $\A_s$ since the exact functional form of the optimal solution $\x(\alpha)$, and thus the unfairness $U(\x(\alpha))$ and the total travel time $\objso(\x(\alpha))$ functions, in $\alpha$ is not directly known, though we show that  $\x(\alpha)$ and $\objso(\x(\alpha))$ are continuous 
in $\alpha$ in Appendix B.

\fi 

We also test (Section~\ref{sec:numerical-experiments}) an alternative approach to I-TAP, which instead of taking a convex combination of the SO-TAP and UE-TAP objectives, interpolates between their solutions. That is, we first compute the optimal solutions of UE-TAP ($\x^{UE}$) and SO-TAP ($\x^{SO}$), and return the value $(1-\gamma) \x^{UE} + \gamma \x^{SO}$ for $\gamma \in [0,1]$. While this Interpolated Solution (I-Solution) method only requires two traffic assignment computations as compared to $O(\frac{1}{s})$ computations of the I-TAP method, it leads to poor performance in comparison with the I-TAP method (see Section~\ref{sec:numerical-experiments}) and does not induce a natural marginal-cost pricing scheme, as I-TAP does (see Section~\ref{sec:main-pricing}). Thus, we focus our attention on I-TAP for this and the next sections.

\subsection{Solution Efficiency and Fairness of I-TAP} \label{sec:eff-properties}

In this section, we study the influence of the convex combination parameter $\alpha$ of I-TAP on the efficiency and fairness of the optimal solution $\f(\alpha)$ (and edge flow $\x(\alpha)$). In particular, we characterize (i) an upper bound on the inefficiency ratio as we vary $\alpha$, and (ii) a range of values of $\alpha$ that are guaranteed to achieve a specified level of unfairness $\beta$ for any optimal solution $\f(\alpha)$.

\ifarxiv
We first evaluate the performance of I-TAP by establishing an upper bound on the inefficiency ratio as a function of $\alpha$. Theorem~\ref{thm:eff-upper-bound} shows that the upper bound of the inefficiency ratio $\rho(\x(\alpha))$ of the optimal traffic assignment of I-TAP$_{\alpha}$ is the minimum between two terms: (i) the Price of Anarchy (PoA) $\Bar{\rho}$, and (ii) a more elaborate bound that is monotonically non-increasing in $\alpha$.
\else
We first evaluate the performance of I-TAP by establishing an upper bound on the inefficiency ratio of the optimal solution of I-TAP$_{\alpha}$ as a function of $\alpha$. Theorem~\ref{thm:eff-upper-bound} shows that this bound is a minimum between (i) the PoA, which we denote as~$\Bar{\rho}$, and (ii) a more elaborate bound that is monotonically non-increasing in $\alpha$.
\fi

\begin{theorem} [I-TAP Solution Efficiency] \label{thm:eff-upper-bound}
For any $\alpha \in (0, 1)$, let $\x(\alpha)$ be the optimal solution to I-TAP$_{\alpha}$. Then, the inefficiency ratio
\[\rho(\x(\alpha))\!\leq\!\min\!\left\{ \Bar{\rho},1\!+\! \frac{1-\alpha}{\alpha}\!\cdot\! \frac{\objue(\x(1))\!-\!\objue(\x(0))}{\objso(\x(1))} \right\}.\]
\end{theorem}

\begin{proof}
To prove this result, we show that (i) $\rho(\x(\alpha)) \leq \Bar{\rho}$ and (ii) $\rho(\x(\alpha)) \leq 1+ \frac{1-\alpha}{\alpha} \frac{\objue(\x(1)) - \objue(\x(0))}{\objso(\x(1))}$.

\ifarxiv
We first prove (i) by establishing that $\objso(\x(\alpha)) \leq \objso(\x(0))$ for all $\alpha \in [0, 1]$. Since $\mathbf{x}(\alpha)$ is the optimal solution to I-TAP$_{\alpha}$
\ifarxiv
\begin{align*}
    \objitap(\x(\alpha)) & = \alpha\objso(\x(\alpha)) + (1-\alpha)\objue(\x(\alpha))\\
    & \leq \objitap(\x(0)) \\
    & =  \alpha\objso(\x(0)) + (1-\alpha)\objue(\x(0)).
\end{align*}
\else
\begin{align*}
    \objitap(\x(\alpha)) & = \alpha\objso(\x(\alpha)) + (1-\alpha)\objue(\x(\alpha))\\
    & \leq \objitap(\x(0)) =  \alpha\objso(\x(0)) + (1-\alpha)\objue(\x(0)),
\end{align*}
\fi
\else
We first prove (i) by establishing that $\objso(\x(\alpha)) \leq \objso(\x(0))$ for all $\alpha \in [0, 1]$. Since $\objitap(\x(\alpha)) \leq \objitap(\x(0))$ it follows that $\alpha\objso(\x(\alpha)) + (1-\alpha)\objue(\x(\alpha)) \leq  \alpha\objso(\x(0)) + (1-\alpha)\objue(\x(0))$.
\fi
Next from the UE-TAP objective, it is the case that 
\ifarxiv
\[\objue(\x(0))\leq \objue(\x(\alpha)).\] 
\else
$\objue(\x(0))\leq \objue(\x(\alpha)).$
\fi
From these two inequalities, it follows that $\objso(\mathbf{x}(\alpha)) \leq \objso(\mathbf{x}(0))$ for all $\alpha \in [0, 1]$. Dividing both sides of the inequality by the minimum possible system travel time $\objso(\x(1))$ establishes that the inefficiency ratio $\rho(\x(\alpha)) \leq \Bar{\rho}$ proving claim (i).

\ifarxiv
Now, to prove claim (ii), we note that $\objitap(\x(\alpha)) \leq \objitap(\x(1))$ since $\mathbf{x}(\alpha)$ is an optimal solution to I-TAP$_{\alpha}$. It thus follows that
\else
To prove claim (ii), from $\objitap(\x(\alpha)) \leq \objitap(\x(1))$ it follows that
\fi
\ifarxiv
\begin{align*}
    &\alpha \objso(\x(\alpha)) + (1-\alpha) \objue(\x(\alpha)) \leq \alpha \objso(\x(1)) + (1-\alpha) \objue(\x(1)).
\end{align*}
\else
$\alpha \objso(\x(\alpha)) + (1-\alpha) \objue(\x(\alpha)) \leq \alpha \objso(\x(1)) + (1-\alpha) \objue(\x(1))$.
\fi
Dividing this expression by $\alpha \objso(\x(1))$ and rearranging gives \ifarxiv us that \fi
\ifarxiv
\begin{align*}
    \frac{\objso(\x(\alpha))}{\objso(\x(1))} &\leq 1+ \frac{1-\alpha}{\alpha}\cdot \frac{\objue(\x(1)) - \objue(\x(\alpha))}{\objso(\x(1))} \\
    &\leq 1+ \frac{1-\alpha}{\alpha}\cdot \frac{\objue(\x(1)) - \objue(\x(0))}{\objso(\x(1))},
\end{align*}
where the second inequality follows from the fact that $\objue(\x(0)) \leq \objue(\x(\alpha))$ for any $\alpha \in [0, 1]$. This proves claim (ii).
\else
$\frac{\objso(\x(\alpha))}{\objso(\x(1))} \leq 1+ \frac{1-\alpha}{\alpha}\cdot \frac{\objue(\x(1)) - \objue(\x(0))}{\objso(\x(1))}$,
where we used that $\objue(\x(0)) \leq \objue(\x(\alpha))$ for any $\alpha$. This proves claim (ii).
\fi
\end{proof}
Theorem~\ref{thm:eff-upper-bound} establishes that, even in the worst case, the ratio between the total travel time of the edge flow $\x(\alpha)$ \ifarxiv of I-TAP$_{\alpha}$ \fi and that of the system optimal solution is at most the PoA. This result is not guaranteed to hold for other state-of-the-art CSO algorithms, \ifarxiv \else e.g., the algorithm in~\cite{so-routing-seminal} (see Section~\ref{sec:numerical-experiments}), \fi which may output traffic assignments with much larger total travel times than the UE solution. \ifarxiv For instance, in Section~\ref{sec:numerical-experiments} we show through numerical experiments that the algorithm proposed by Jahn et al. \cite{so-routing-seminal} achieves total travel times much higher than that of the user equilibrium solution for certain ranges of unfairness. \fi Further, Theorem~\ref{thm:eff-upper-bound} shows that the upper bound on the inefficiency ratio becomes closer to one as the objective $\objitap$ gets closer to $\objso$.

\ifarxiv 
The bound on the inefficiency ratio obtained from Theorem~\ref{thm:eff-upper-bound} is depicted in Figure~\ref{fig:ineff_ratio}. Note  that in order to obtain the value of the convex combination parameter $\alpha^*$ at which the two upper bounds on the inefficiency ratio are equal we need to impose the constraint $\Bar{\rho} = \frac{1-\alpha^*}{\alpha^*}\cdot \frac{\objue(\x(1)) - \objue(\x(0))}{\objso(\x(1))}$, which yields that $\alpha^* = \frac{\objue(\x(1)) - \objue(\x(0))}{\objso(\x(0)) + \objue(\x(1)) - \objue(\x(0))}$.

\begin{figure}[tbh!]
      \centering
      \includegraphics[width=80mm]{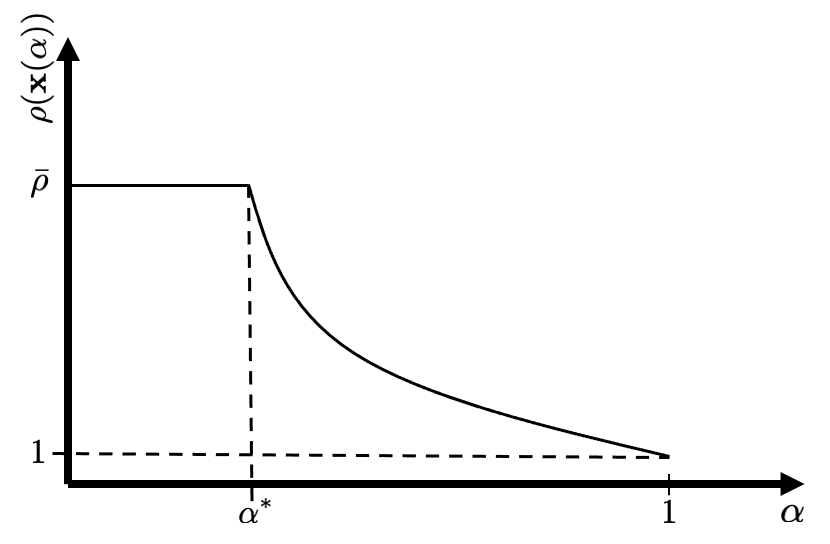}
      \vspace{-5pt}
      \caption{\small \sf Upper bound on the inefficiency ratio of the optimal solution $\x(\alpha)$ of the problem I-TAP$_{\alpha}$ for $\alpha \in [0, 1]$}
      \label{fig:ineff_ratio} 
   \end{figure}
\fi

\ifarxiv Having established a worst case performance guarantee for I-TAP in terms of the inefficiency ratio, we \else We \fi now establish a range of values of $\alpha$ at which any optimal solution $\f(\alpha)$ of I-TAP$_{\alpha}$ is guaranteed to attain a $\beta$ bound on \ifarxiv unfairness. In particular, we specialize the following result to polynomial travel time functions, e.g., the commonly used BPR function \cite{Sheffi1985}. \else unfairness for polynomial travel time functions, e.g., the commonly used BPR function \cite{Sheffi1985}. \fi 

\begin{theorem} [Feasibility of I-TAP for \fso] \label{thm:poly-ttf-result}
Suppose that the largest degree of the polynomial travel time functions $t_e(x_e)$ is $m$ for some $e \in E$. Then, the unfairness of any optimal solution $\f(\alpha)$ of I-TAP$_{\alpha}$ is upper bounded by $\beta$, i.e., $U(\f(\alpha)) \leq \beta$, for any $\alpha \leq \frac{\beta - 1}{m}$.
\end{theorem}

The proof of Theorem~\ref{thm:poly-ttf-result} leverages the fact that at an equilibrium flow, under a given vector of tolls, the travel cost on all positive paths is equal for all users in any commodity $k$, as was established in~\cite[Lemma 2]{basu2017reconciling}. We provide a proof of this claim here for completeness.

\begin{lemma} [Equality of Travel Costs on Positive Paths~\cite{basu2017reconciling}] \label{lem:equalityCosts}
Let the cost on an edge be given by $c_e(x_e, \tau_e) = t_e(x_e) + \tau_e$, where the value-of-time $v$ of users is normalized to one. Then, under any equilibrium flow induced by the edge tolls $\boldsymbol{\tau}$, the total travel cost of any two positive paths $P, Q$ are equal, i.e., $\sum_{e \in P} c_e(x_e, \alpha) = \sum_{e \in Q} c_e(x_e, \alpha)$. 
\end{lemma}

\begin{proof}
To prove this claim, we consider a network with travel time functions $c_e(x_e, \tau_e) = t_e(x_e) + \tau_e$ for all $e \in E$, which is valid since the travel time function is polynomial and so $c_e(x_e, \tau_e)$ is still differentiable, convex and monotonically increasing in $x_e$. Next, the equilibrium solution under a vector of tolls $\boldsymbol{\tau}$ is given by the UE-TAP with objective $\sum_{e \in E}\int_{0}^{x_e} c_e(y, \alpha) \dd{y}$, where $c_e(x_e, \tau_e) = t_e(x_e) + \tau_e$. By the first order necessary and sufficient conditions of this UE-TAP, for any two positive flow paths $P, Q$ between the same O-D pair $k$, it holds that $\sum_{e \in P} c_e(x_e, \alpha) = \sum_{e \in Q} c_e(x_e, \alpha) = \mu_k$ for some $\mu_k$. We now show that any positive path $R$ between the same O-D pair that is not a path with positive flow, i.e., $\x_R = 0$, also has the same travel cost as any path $P$ with positive flow, i.e., $\sum_{e \in R} c_e(x_e, \alpha) = \sum_{e \in P} c_e(x_e, \alpha) = \mu_k$.

Suppose by contradiction that there is a path $R$ for a commodity $k$ whose travel cost is not equal to $\mu_k$. From the equilibrium condition of UE-TAP it must hold that $\sum_{e \in R} c_e(x_e, \tau_e)>\mu_k$. Next, let $e_1, \ldots, e_r$ denote the edges in path $R$ in the order of traversal, and for each edge $e_i \in R$, consider a path used by commodity $k$ using that edge, i.e., a path $P_i = R_i^s-e_i-R_i^t$ where $R_i^s$ is the component of that path preceding edge $e_i$ and $R_i^t$ is the component of that path following edge $e_i$. Observe that the sum of the costs of these used paths $P_i$ is given by $\sum_{i = 1}^r \sum_{e \in P_i} c_e(x_e, \tau_e) = r \mu_k$. Next, considering the paths $P_1', \ldots, P_{r-1}'$, where $P_1' = R_{i+1}^s-R_i^t$, we obtain that
\begin{align*}
    r \mu_k = \sum_{i = 1}^r \sum_{e \in P_i} c_e(x_e, \tau_e) = \sum_{i = 1}^{r-1} \sum_{e \in P_i'} c_e(x_e, \tau_e) + \sum_{e \in R} c_e(x_e, \tau_e).
\end{align*}
Since it holds from the KKT conditions of UE-TAP that $\sum_{e \in P_i'} c_e(x_e, \tau_e) \geq \mu_k$ for each $i \in \{1, \ldots, r-1 \}$, the above equality implies that $\sum_{e \in R} c_e(x_e, \tau_e) \leq \mu_k$, a contradiction. Thus, we must have that the total cost on all positive paths is exactly equal to $\mu_k$, thereby proving our claim.
\end{proof}
We now leverage Lemma~\ref{lem:equalityCosts} to complete the proof of Theorem~\ref{thm:poly-ttf-result}.

\begin{proof}
Without loss of generality, we normalize the value of time $v$ to $1$ and for notational simplicity denote $x_e = x_e(\alpha)$ and $\f = \f(\alpha)$. We now establish this result in two steps. 
First, we show that if $c_e(x_e, \alpha) = t_e(x_e) + \alpha x_e t_e'(x_e) \in [t_e(x_e), \beta t_e(x_e)]$ for all edges $e$ then the unfairness satisfies $U(\f) \leq \beta$. Next, we show for $\alpha \leq \frac{\beta - 1}{m}$ that the cost satisfies $c_e(x_e, \alpha) \in [t_e(x_e), \beta t_e(x_e)]$, which, together with the first claim, implies the result.

To prove the first claim, we consider a network with travel time functions $c_e(x_e, \alpha)$ for all $e \in E$, which is valid since the travel time function is polynomial and so $c_e(x_e, \alpha)$ is still differentiable, convex and monotonically increasing in $x_e$. 
Next, from Lemma~\ref{lem:equalityCosts} it holds that the total cost on any two positive paths $P$ and $Q$ are also equal, i.e., $\sum_{e \in P} c_e(x_e, \alpha) = \sum_{e \in Q} c_e(x_e, \alpha)$ since at an equilibrium flow the cost of any two used paths are equal. This result implies for any commodity $k$ and any two positive paths $P, Q$ that
\begin{align*}
    \sum_{e \in P} t_e(x_e) &\leq \sum_{e \in P} c_e(x_e, \alpha) 
    = \sum_{e \in Q} c_e(x_e, \alpha) \\
    &\leq \beta \sum_{e \in Q} t_e(x_e),
\end{align*}
i.e., the ratio of the travel times on positive paths can never exceed $\beta$. Thus, $U(\f) \leq \beta$.

\ifarxiv
Next, to prove the second claim, we note that since the travel time function is a polynomial function of degree $m$, we can let $t_{e}\left(x_{e}\right)=\sum_{i=0}^{m} \gamma_{i} x_{e}^{i}$.
Then, we have:
\begin{align*}
x_{e} t_{e}^{\prime}\left(x_{e}\right) &=x_{e} \sum_{i=1}^{m} i \gamma_{i} x_{e}^{i-1} 
=\sum_{i=1}^{m} i \gamma_{i} x_{e}^{i} \\
& \leq m \sum_{i=1}^{m} \gamma_{i} x_{e}^{i}  \leq m \sum_{i=0}^{m} \gamma_{i} x_{e}^{i}  \\
&=m t_{e}\left(x_{e}\right).
\end{align*}
\else
Next, to prove the second claim, we first observe that $x_{e} t_{e}^{\prime}(x_e) \leq m t_{e}\left(x_{e}\right)$ for all $e$ since the travel time functions are non-negative polynomials of degree at most $m$. Thus, it follows that $\alpha x_e t_e^{\prime}(x_e) \leq \alpha m t_{e}\left(x_{e}\right)$. Now, if we set $\alpha m t_{e}\left(x_{e}\right) \leq (\beta-1)t_e(x_e)$, then we have for any $\alpha \leq \frac{\beta - 1}{m}$ that the cost $c_e(x_e, \alpha) \in [t_e(x_e), \beta t_e(x_e)]$ and thus the resulting flow $\x$ has an unfairness of at most $\beta$.
\fi
\ifarxiv
Note that we used the fact that $\gamma_0 \geq 0$ in the second inequality, which follows since $t_e(x_e) \geq 0$ for all $x_e \in \mathbb{R}_{\geq 0}$. From the above inequalities, we note that $\alpha x_e t_e^{\prime}(x_e) \leq \alpha m t_{e}\left(x_{e}\right)$. Now, if we set $\alpha m t_{e}\left(x_{e}\right) \leq (\beta-1)t_e(x_e)$, then we have for any $\alpha \leq \frac{\beta - 1}{m}$ that the cost $c_e(x_e, \alpha) \in [t_e(x_e), \beta t_e(x_e)]$ and thus the resulting flow $\f$ has an unfairness of at most $\beta$.
\fi
\end{proof}
\ifarxiv Theorem~\ref{thm:poly-ttf-result} establishes a relation between the convex combination parameter $\alpha$ and the level of unfairness of any optimal path flow $\f(\alpha)$ when the edge travel-time functions are polynomial.
\fi



We can further show that the bound \ifarxiv obtained \fi in Theorem~\ref{thm:poly-ttf-result} is in fact tight by demonstrating an instance such that for any $\alpha > \frac{\beta-1}{m}$ the unfairness of the solution $\f(\alpha)$ of I-TAP$_{\alpha}$ is strictly greater than $\beta$. 

\begin{lemma} [Tightness of Unfairness Bound] \label{lem:unfairness-tight}
Suppose $\f(\alpha)$ is an optimal solution to I-TAP$_{\alpha}$ for any $\alpha \in [0, 1]$. Then, there exists a two-edge parallel network with polynomial travel time functions of degree at most $m$ such that for any $\alpha > \frac{\beta-1}{m}$, the unfairness $U(\f(\alpha)) > \beta$.
\end{lemma}
\begin{proof}
Consider a demand of one in a two edge Pigou network, with two nodes---an origin and a destination node---connected by two parallel edges $e_1$ and $e_2$. For edge $e_1$, let $t_1(x_1) = 1 + \epsilon x$ for some small $\epsilon>0$, and let $t_2(x_2) = x_2^m$ for some $m \in \mathbb{N}$. Then, the first order necessary and sufficient KKT conditions of I-TAP$_{\alpha}$ for any $\alpha \in [0, 1]$ imply that
\begin{align*}
    t_2(x_2) + \alpha x_2 t_2'(x_2) = t_1(x_1) + \alpha x_1 t_1'(x_1).
\end{align*}
Substituting $t_2(x_2) = x_2^m$ and $t_1(x_1) = 1+ \epsilon x_1$ gives the following ratio between the travel times on the two links:
\begin{align*}
    \frac{t_1(x_1)}{t_2(x_2)} = \frac{1 + \epsilon x_1}{x_2^{m}} - \frac{\alpha \epsilon x_1}{x_2^m} = 1+m \alpha - \frac{\alpha \epsilon x_1}{x_2^m}.
\end{align*}
By the unfairness constraint that $\frac{t_1(x_1)}{t_2(x_2)} \leq \beta$, it follows that if $\alpha>\frac{\beta - 1}{m}$, then
\begin{align*}
    \frac{t_1(x_1)}{t_2(x_2)} = 1+m \alpha - \frac{\alpha \epsilon x_1}{x_2^m} > 1+m\frac{\beta - 1}{m} - \frac{\alpha \epsilon x_1}{x_2^m} = \beta - \frac{\alpha \epsilon x_1}{x_2^m}.
\end{align*}
Taking $\epsilon \rightarrow 0$, we have that $\frac{t_1(x_1)}{t_2(x_2)} \geq \beta$ for any $\alpha > \frac{\beta-1}{m}$. Thus, the unfairness of the optimal flow $\f(\alpha)$, which is identical to the edge flow $\x(\alpha)$ in a parallel network setting, is at least $\beta$.
\end{proof}

\ifarxiv \else For a counterexample to prove Lemma~\ref{lem:unfairness-tight}, see Appendix~\ref{apdx:pf-unfairness-tight}. \fi
Together, Theorem~\ref{thm:poly-ttf-result} and Lemma~\ref{lem:unfairness-tight} imply that a $\beta$ level of unfairness can be guaranteed using I-TAP on all traffic networks only when $\alpha \leq \frac{\beta-1}{m}$, where $m$ is the maximum degree of the polynomial corresponding to the travel time functions for each edge $e \in E$.

\subsection{Optimality of I-TAP} \label{sec:optimality}

In this section, we show that I-TAP exactly computes the minimum total travel time solution for any desired level of unfairness $\beta$ in any two edge Pigou network. That is, there is some convex combination parameter $\alpha^*$ for which the solution of I-TAP$_{\alpha^*}$ is also a solution to the \fso problem for any two edge Pigou network. 

\begin{lemma} [Optimality of I-TAP] \label{lem:itap-optimal-pigou}
Consider a two edge Pigou network where the optimal solution of the \fso problem is $\x_{\beta}^*$ for any $\beta \in [1, \infty)$. Then, there exists a convex combination parameter $\alpha^*$ such that $\x(\alpha^*) = \x_{\beta}^*$, i.e., the solution of I-TAP$_{\alpha^*}$ and the optimal solution of the \fso problem coincide.
\end{lemma}

For a proof of Lemma~\ref{lem:itap-optimal-pigou}, see \ifarxiv Appendix~\ref{apdx:pf-pigou-optimal}\else Appendix A.2\fi. We mention that Lemma~\ref{lem:itap-optimal-pigou} compares only the edge flows of I-TAP and \fso since the path and edge flows coincide for a two edge Pigou network. We also note that while the optimality for a Pigou network may appear restrictive, such networks are of both theoretical~\cite{pigou,how-bad-is-selfish} and practical significance~\cite{caltrans}. 



\subsection{Computational Tractability of I-TAP} \label{sec:compTractability}

Having established that we can solve I-TAP to obtain an approximate solution to \fso, we now establish that I-TAP can be computed efficiently due to its equivalence to a parametric UE-TAP program.


\begin{observation} [UE Equivalency of I-TAP] \label{obs:UE-equivalency}
For any $\alpha \in [0, 1]$, I-TAP$_{\alpha}$ reduces to UE-TAP with objective function $\sum_{e \in E}\int_{0}^{x_e} c_e(y, \alpha) \dd{y}$, where $c_e(y, \alpha) = t_e(y) + \alpha y t_e'(y)$.
\end{observation}
\ifarxiv 
To see Observation~\ref{obs:UE-equivalency}, note that by the fundamental theorem of calculus, $\objso(\x)$ can be written as
$$\objso(\x) = \sum_{e \in E} x_e t_e(x_e) = \sum_{e \in E} \int_{0}^{x_e} t_e(y) + y t_e'(y) \dd{y}. $$
Then, taking a convex combination of the SO-TAP and UE-TAP objectives, it is clear that 
$$\objitap(\x) = \sum_{e \in E}\int_{0}^{x_e} t_e(y) + \alpha y t_e'(y) \dd{y}.$$
\else
Observation~\ref{obs:UE-equivalency} follows 
from the fundamental theorem of calculus.
\fi 
Note that for each $\alpha \in [0, 1]$, the differentiability, monotonicity, and convexity of many typical travel time functions $t_e$, e.g., any polynomial function such as the BPR function~\cite{Sheffi1985}, imply that the corresponding properties hold for the cost functions $c_e(x_e, \alpha)$ in $x_e$. For numerical implementation, 
the equivalency of I-TAP$_\alpha$ and UE-TAP implies that I-TAP$_\alpha$ inherits the useful property that the linearization step of the Frank-Wolfe algorithm~\cite{Sheffi1985}, when applied to I-TAP$_\alpha$, corresponds to solving multiple 
unconstrained shortest path queries. The latter motivates the highly efficient approach which we employ in Section~\ref{sec:numerical-experiments} to solve I-TAP$_\alpha$.

\ifarxiv
\subsection{Sensitivity Analysis of I-TAP} \label{sec:solution-properties}

To obtain a solution that is simultaneously of low cost while keeping within a $\beta$ bound of unfairness, it is instructive to study the sensitivity of the optimal edge flow solution $\x(\alpha)$ of I-TAP$_{\alpha}$ in the convex combination parameter $\alpha$. Specifically, performing such a sensitivity analysis is important since it allows us to characterize the continuity of the SO-TAP objective $\objso(\x(\alpha))$, which we are looking to minimize in the \fso problem. In this section, we leverage the equivalency of I-TAP and UE-TAP to establish continuity properties of the optimal edge flow solution $\x(\alpha)$ of I-TAP$_{\alpha}$ and the SO-TAP objective $\objso(\x(\alpha))$.

We first establish the Lipschitz continuity of the optimal edge flow solution $\x(\alpha)$ of I-TAP$_{\alpha}$ in the convex combination parameter
$\alpha$. This result follows from the UE-TAP reformulation of I-TAP as in Observation~\ref{obs:UE-equivalency} and a direct application of the Lipschitz continuity of the optimal solution of a parametric UE-TAP \cite[Theorem 8.6a]{parametric-still}.

\begin{corollary} (Lipschitz Continuity of $\x(\alpha)$) \label{cor:lipschitz-cont-x}
For any $\alpha \in (0, 1)$, the edge flow solution $\mathbf{x}(\alpha)$ of the problem I-TAP$_{\alpha}$ is Lipschitz continuous in $\alpha$.
\end{corollary}
Corollary~\ref{cor:lipschitz-cont-x} implies a desirable property that small changes in $\alpha$ will only result in small changes in the optimal edge flows of I-TAP$_{\alpha}$. We now use Corollary~\ref{cor:lipschitz-cont-T} to establish the Lipschitz continuity of the SO-TAP objective.

\begin{corollary} (Lipschitz Continuity of $\objso$) \label{cor:lipschitz-cont-T}
For any $\alpha \in (0, 1)$, let $\x(\alpha)$ be the optimal edge flow solution to I-TAP$_{\alpha}$. Then, 
$\objso(\mathbf{x}(\alpha))$ is Lipschitz continuous in $\alpha$.
\end{corollary}

\begin{proof}
We prove the Lipschitz continuity of $\objso(\mathbf{x}(\alpha))$ in $\alpha$ through the observation that the composition of Lipschitz continuous functions is Lipschitz continuous. First, observe that $\mathbf{x}(\alpha)$ is Lipschitz continuous by Corollary~\ref{cor:lipschitz-cont-x}. Next, the SO-TAP objective $\objso(\mathbf{x}) = \sum_{e \in E} x_e t_e(x_e)$ is 
continuous in its argument $\mathbf{x}$. Furthermore, $\objso(\mathbf{x})$ is locally Lipschitz over a bounded set since the SO-TAP has a finite derivative as long as $\mathbf{x}$ is bounded. Since both $\objso(\mathbf{x})$ and $\mathbf{x}(\alpha)$ are Lipschitz continuous, it follows that $\objso(\mathbf{x}(\alpha))$ is Lipschitz continuous.
\end{proof}

While Corollary~\ref{cor:lipschitz-cont-T} establishes the continuity of $\objso(\mathbf{x}(\alpha))$ in $\alpha$, this relation is not necessarily monotone, which we show through numerical experiments in Section~\ref{sec:numerical-experiments}. We also note that, unlike the SO-TAP objective, the unfairness function $U(\f(\alpha))$ is discontinuous in $\alpha$. Its discontinuity stems from the fact that the optimal solutions $\f(\alpha_1)$ and $\f(\alpha_2)$ for convex combination parameters $\alpha_1$ and $\alpha_2$ that are arbitrarily close may not have the same set of positive paths. In particular, it may happen for some convex combination parameter $\alpha_1$ that path $P$ is a positive path, but for some $\epsilon>0$ and for any $\alpha_2$ such that $||\alpha_1 - \alpha_2 ||_2 \leq \epsilon$ we have that $P$ is not a positive path. Since the travel time on path $P$ could be very different from the travel times on other paths, the unfairness function is in general discontinuous in $\alpha$, which we validate through experiments in Section~\ref{sec:numerical-experiments}.
\fi  

\section{Pricing to Implement Flows} \label{sec:main-pricing}

\ifarxiv  In the previous section we presented a method for computing a solution to \fso that keeps within a $\beta$ bound of unfairness and strives to minimize the total travel time. \fi 
In this section, we leverage the structure of
I-TAP to develop pricing mechanisms to collectively enforce
the I-TAP solution in the presence of selfish users that independently choose routes to minimize their own travel costs. 
\ifarxiv  That is, the prices are set such that the travel cost for users in the same commodity is equivalent, ensuring the formation of an equilibrium. \fi 
We first consider the case of homogeneous users and show that I-TAP results in a natural marginal-cost pricing scheme. \ifarxiv Then, we leverage a linear programming methodology of \cite{multicommodity-extension} to set road prices to enforce the flows computed through the I-TAP method for the setting of heterogeneous users. \else Then, we characterize conditions under which tolls can be used to enforce the I-TAP flows for heterogeneous users. \fi

In this section, for the ease of exposition, we focus our discussion on inducing the optimal edge flow $\x(\alpha)$ of I-TAP$_{\alpha}$. We mention that our approach can naturally be extended to enforcing optimal path flows $\f(\alpha)$ that satisfy a given level of unfairness. In particular, we can consider a setting wherein users are recommended to use a specified path set, e.g., by traffic navigational applications, as given by $\f(\alpha)$ and the tolls set are such that no user will have an incentive to deviate from their recommended paths. Finally, we also mention by the result of Theorem~\ref{thm:poly-ttf-result} that focusing on the edge flow solution $\x(\alpha)$ is without loss of generality for certain ranges of $\alpha$ since the unfairness bound for any optimal path flow solution $\f(\alpha)$ is guaranteed to be satisfied.

\subsection{Homogeneous Pricing via Marginal Cost} \label{sec:pricing-homogeneous}

In the setting where all users have the same value of time $v$, the structure of I-TAP$_{\alpha}$ yields an interpolated variant of marginal-cost pricing to induce selfish users to collectively form the optimal edge flow $\x(\alpha)$ of I-TAP$_{\alpha}$. This result is a direct consequence of the equivalence between I-TAP and UE-TAP from Observation~\ref{obs:UE-equivalency}.

\begin{lemma} [Prices to Implement Flows] \label{lem:pricing-homogeneous}
Suppose that the edge flow $\x(\alpha)$ is a solution to I-TAP$_\alpha$ for some $\alpha\in [0,1]$. Then $\x(\alpha)$ can be enforced as a UE by setting the prices as $\tau_e = \alpha x_e(\alpha) t_e'(x_e(\alpha))$ for each $e \in E$.
\end{lemma}

\ifarxiv
\begin{proof}
Without loss of generality, normalize $v$ to 1 and for notational convenience, denote $x_e = x_e(\alpha)$. To prove this result, note from Observation~\ref{obs:UE-equivalency} that I-TAP$_{\alpha}$ is equivalent to UE-TAP with objective $\sum_{e\in E}\int_{0}^{x_e} t_e(y) + \alpha y t_e'(y) \dd{y}$. From the first-order necessary and sufficient KKT conditions~\cite{Sheffi1985} of this UE-TAP, for any two paths $P, Q \in \mathcal{P}_k$ with positive flow for a commodity $k \in K$, it must be that $\sum_{e \in P} \left( t_e(x_e) + \alpha x_e t_e'(x_e) \right) = \sum_{e \in Q} \left( t_e(x_e) + \alpha x_e t_e'(x_e) \right)$. Thus, if the prices on each edge are set as $\tau_e = \alpha x_e t_e'(x_e)$, then all users in each commodity incur the same travel cost when using any two paths $P, Q \in \mathcal{P}_k$, establishing that $\x$ is a UE.
\end{proof}
\fi

\ifarxiv  \else For a proof of Lemma~\ref{lem:pricing-homogeneous}, see Appendix A.3. \fi Note from Lemma~\ref{lem:pricing-homogeneous} that the edge prices are equal to $\alpha$ multiplied by the marginal cost of \ifarxiv users on the edges. \else users. \fi \ifarxiv For this pricing scheme, Lemma~\ref{lem:pricing-homogeneous} guarantees that selfish users will be induced to collectively form the solutions satisfying a specified bound on unfairness obtained through the I-TAP method. \fi \ifarxiv \else We also note that tradable-credit schemes~\cite{YANG2011580} can be used to implement the flows computed through I-TAP (Appendix C). \fi

\ifarxiv
\begin{remark}
We note that road tolling is not the only mechanism that can be used to induce selfish users to collectively form the flows computed through I-TAP. One of the notable mechanisms beyond road tolling to enforce the SO solution as a UE is that of tradable credit schemes, wherein users can trade initially issued credits freely in a competitive market and spend these credits to use roads with predetermined credit charges \cite{YANG2011580}. As with the marginal cost pricing scheme in Lemma~\ref{lem:pricing-homogeneous}, the equivalence between I-TAP and UE-TAP for any $\alpha \in [0, 1]$ enables us to derive a tradable credit scheme. In particular, following a similar line of reasoning to that used by Yang and Wang \cite[Proposition 5]{YANG2011580}, it can be shown that the optimal edge flow $\mathbf{x}(\alpha)$ of I-TAP$_{\alpha}$ for any $\alpha \in [0, 1]$ can be enforced as a user equilibrium through a tradable credit scheme. 
\end{remark}
\fi

\subsection{Heterogeneous Pricing via Dual Multipliers} \label{sec:pricing-heterogeneous}

\ifarxiv
The pricing mechanism in Section~\ref{sec:pricing-homogeneous} is inapplicable to the heterogeneous user setting as it would require unrealistically imposing different prices for users with different values of time for the same edges. In this section, we consider heterogeneous users and leverage a linear-programming method \cite{multicommodity-extension} to provide a necessary and sufficient condition to induce selfish users to collectively form the edge flow $\x(\alpha)$. We further establish that $\x(\alpha)$ satisfies this condition for each $\alpha \in [0, 1]$. That is, appropriate tolls can be placed on the roads of the network to induce heterogeneous selfish users to collectively form the equilibrium edge flow $\x(\alpha)$.
\else
The pricing mechanism in Section~\ref{sec:pricing-homogeneous} is inapplicable to the heterogeneous user setting as it would require unrealistically imposing different prices for users with different values of time for the same edges. In this section, we consider heterogeneous users and leverage a linear-programming method \cite{multicommodity-extension} to establish that appropriate tolls can be placed on the roads to induce heterogeneous selfish users to collectively form the equilibrium flow $\x(\alpha)$.
\fi

Before presenting the pricing scheme, we first extend the notion of a commodity to a heterogeneous user setting. In particular, each user belongs to a commodity $k \in K$ when making a trip on a set of paths $\mathcal{P}_k$ between the same O-D pair and has the value of time $v_k>0$. Then, under a vector of edge prices $\boldsymbol{\tau} = \{\tau_e\}_{e \in E}$ the travel cost that users in commodity $k$ incur on a given path $P \in \mathcal{P}_k$ under the traffic assignment $\x$ is given by $C_{P}(\x, \boldsymbol{\tau}) =  \sum_{e \in P} \left( v_k t_e(x_e) + \tau_e \right)$. Note that 
more than one commodity 
may make trips between the same O-D pair, and a user equilibrium forms when the travel cost for all users in a particular commodity is equal. We further note that we maintain the unfairness notion presented in the work \ifarxiv thus far \fi even for heterogeneous users. That is, irrespective of the value of time of two users travelling between the same O-D pair, the maximum possible ratio between their travel times can be no more than $\beta$.

We now leverage the following result 
to provide a necessary and sufficient condition that the optimal edge flow $\x(\alpha)$ of I-TAP$_{\alpha}$ \ifarxiv for any $\alpha\in [0,1]$ \fi must satisfy for it to be enforceable as a UE through road pricing. 




\ifarxiv
\begin{lemma} (Condition for Flow Enforceability
\cite[Theorem 3.1]{multicommodity-extension})
\label{lem:pricing-flow-implement}
Suppose that the flow $\x$ satisfies the feasibility Constraints~\eqref{eq:edge-constraint}-\eqref{eq:nonnegativity-constraints}. Further, consider the linear program with the variables $d_{P}^{k}$, which represents the flow of commodity $k$ routed on path $P \in \mathcal{P}_k$, where $\mathcal{P}_k$ denotes the set of all possible paths for commodity $k$:
\begin{mini!}|s|[2]                   
    {\substack{d_{P}^{k}, \forall P \in \mathcal{P}_k, \\ \forall k \in K}}                              
    {\sum_{k \in K} v_k \sum_{P \in \mathcal{P}_k}  t_P(\x) d_{P}^{k}, \label{eq:OPT-Obj-Roughgarden}}   
    {\label{eq:Eg001}}             
    {}                                
    \addConstraint{\sum_{P \in \mathcal{P}_{k}} d_{P}^{k}}{= d_{k}, \quad \forall k \in K, \label{eq:OPTcon1-Roughgarden}}    
    \addConstraint{d_{P}^{k}}{  \geq 0, \quad \forall P \in \mathcal{P}_k, k \in K, \label{eq:OPTcon2-Roughgarden}}
    \addConstraint{\sum_{k \in K} \sum_{P \in \mathcal{P}_{k}: e \in P} d_{P}^{k}}{  \leq x_e, \quad \forall e \in E, \label{eq:OPTcon3-Roughgarden}}
\end{mini!}
with demand Constraints~\eqref{eq:OPTcon1-Roughgarden}, non-negativity Constraints~\eqref{eq:OPTcon2-Roughgarden} and capacity Constraints~\eqref{eq:OPTcon3-Roughgarden}. Then $\x$ can be enforced as a user equilibrium if and only if Constraint~\eqref{eq:OPTcon3-Roughgarden} is met with equality for each edge $e \in E$ at the optimal solution of the linear Program~\eqref{eq:OPT-Obj-Roughgarden}-\eqref{eq:OPTcon3-Roughgarden}.
\end{lemma}
\else
\begin{lemma} [Condition for Flow Enforceability]
\cite[Theorem 3.1]{multicommodity-extension})
\label{lem:pricing-flow-implement}
Suppose that the non-negative edge flow $\x$ satisfies edge flow and demand constraints as in Definition~\ref{def:so-tap}. Further, consider the linear program: $\min_{d_{P}^{k} \in \Tilde{\Omega}} \sum_{k \in K} v_k \sum_{P \in \mathcal{P}_k}  t_P(\x) d_{P}^{k}$, where the non-negative variables $d_{P}^{k}$ represent the flow of commodity $k$ on path $P \in \mathcal{P}_k$, and $\mathcal{P}_k$ denotes the set of all possible paths for commodity $k$.
Here $\Tilde{\Omega}$ is the set described by non-negative flows satisfying capacity
constraints, i.e., $\sum_{k \in K} \sum_{P \in \mathcal{P}_{k}: e \in P} d_{P}^{k} \leq x_e$ for all edges $e \in E$, and demand constraints, i.e., $\sum_{P \in \mathcal{P}_{k}} d_{P}^{k} = d_{k}$ for all $k \in K$. Then $\x$ can be enforced as a UE if and only if the capacity constraints are met with equality for each edge \ifarxiv $e \in E$ \fi at the optimal solution of the linear program.
\end{lemma}
\fi
\ifarxiv
In particular, if $\x(\alpha)$ satisfies the above necessary and sufficient condition then it can be enforced as a user equilibrium 
through edge prices set based on the dual variables of the capacity constraints of the above linear program. \fi 
We now show that $\x(\alpha)$ satisfies the \ifarxiv necessary and sufficient \fi condition in Lemma~\ref{lem:pricing-flow-implement}.

\begin{lemma} \label{lem:single-od-result} [Heterogeneous User Flow Enforceability]
Suppose that the edge flow $\x(\alpha)$ is a solution for I-TAP$_\alpha$ for some $\alpha\in [0,1]$. Then for the heterogeneous user setting, $\x(\alpha)$ can be enforced as a user equilibrium.
\end{lemma}
\ifarxiv 
\begin{proof}
To prove this result, from Lemma~\ref{lem:pricing-flow-implement} it suffices to show that Constraints~\eqref{eq:OPTcon3-Roughgarden} are met with equality at the optimal solution of the linear Program~\eqref{eq:OPT-Obj-Roughgarden}-\eqref{eq:OPTcon3-Roughgarden} for the flow $\x(\alpha)$. We now suppose by contradiction that the optimal solution $\Tilde{\x}$ to the linear Program~\eqref{eq:OPT-Obj-Roughgarden}-\eqref{eq:OPTcon3-Roughgarden} for the edge flow $\x(\alpha)$ is such that for at least one edge $e^*\in E$ the Constraint~\eqref{eq:OPTcon3-Roughgarden} is met with a strict inequality, i.e., $\Tilde{x}_{e^*}<x_{e^*}(\alpha)$. Since the flow $\Tilde{\x}$ is optimal for this linear program it follows that $\Tilde{\x}_{P} \geq 0$ for all paths $P \in \mathcal{P}_k$ for all commodities $k \in K$, and that $\sum_{P \in \mathcal{P}_k} \Tilde{\x}_{P} = d_k$ for each commodity $k \in K$ by the constraints of the linear program. Note here that the demand constraints satisfy $\sum_{P \in \mathcal{P}_k} \Tilde{\x}_{P} = d_k$ and the non-negativity constraints $\Tilde{\x}_{P} \geq 0$ imply that the corresponding demand and non-negativity constraints for I-TAP$_{\alpha}$ are also satisfied for the flow $\Tilde{\x}$. By the edge decomposition of path flows it must further hold that $\sum_{k \in K} \sum_{P \in \mathcal{P}_k: e \in P} \Tilde{\x}_{P} = \Tilde{x}_e$. Thus, we observe that the flow $\Tilde{\x}$ is a feasible solution to I-TAP$_{\alpha}$. 

Next, from the capacity constraint of the linear program it follows that for each $e \in E$ that $\Tilde{x}_e \leq x_e(\alpha)$ and that for at least one edge $e^*$ that $\Tilde{x}_{e^*}<x_{e^*}(\alpha)$ by assumption. Then, by the monotonicity of the I-TAP$_{\alpha}$ objective in the edge flows, we observe that $\objitap(\Tilde{\x}) < \objitap(\x(\alpha))$, implying that $\x(\alpha)$ is not an optimal solution to I-TAP$_{\alpha}$, a contradiction. Thus, the flow $\Tilde{\x}$ cannot exist, proving our claim that any optimal solution to the Program~\eqref{eq:OPT-Obj-Roughgarden}-\eqref{eq:OPTcon3-Roughgarden} for the edge flow $\x(\alpha)$ must meet the capacity constraint with equality.
\end{proof}
\else
\begin{proof}
To prove this result, from Lemma~\ref{lem:pricing-flow-implement} it suffices to show that the capacity constraints are met with equality at the optimal solution of the linear program in the statement of the lemma for the flow $\x(\alpha)$. We now suppose by contradiction that the optimal solution $\Tilde{\x}$ to this linear program for the flow $\x(\alpha)$ is such that for at least one edge $e^*\in E$ the capacity constraint is met with a strict inequality, i.e., $\Tilde{x}_{e^*}<x_{e^*}(\alpha)$. Since the flow $\Tilde{\x}$ is optimal for this linear program it follows that $\Tilde{\x}_{P} \geq 0$ for all paths $P \in \mathcal{P}_k$ for all commodities $k \in K$, and that $\sum_{P \in \mathcal{P}_k} \Tilde{\x}_{P} = d_k$ for each commodity $k \in K$ by the constraints of the linear program. \ifarxiv Note here that the demand constraints satisfy $\sum_{P \in \mathcal{P}_k} \Tilde{\x}_{P} = d_k$ and the non-negativity constraints $\Tilde{\x}_{P} \geq 0$ imply that the corresponding demand and non-negativity constraints for I-TAP$_{\alpha}$ are also satisfied for the flow $\Tilde{\x}$. \fi 
By the edge decomposition of path flows it must further hold that $\sum_{k \in K} \sum_{P \in \mathcal{P}_k: e \in P} \Tilde{\x}_{P} = \Tilde{x}_e$. Thus, \ifarxiv we observe that the flow \fi $\Tilde{\x}$ is a feasible solution to I-TAP$_{\alpha}$. 

Next, from the capacity constraint of the linear program it follows for each $e \in E$ that $\Tilde{x}_e \leq x_e(\alpha)$ and for \ifarxiv at least one edge \else some \fi $e^*$ that $\Tilde{x}_{e^*}<x_{e^*}(\alpha)$ by assumption. Then, by the monotonicity of the I-TAP$_{\alpha}$ objective in the edge flows, we observe that $\objitap(\Tilde{\x}) < \objitap(\x(\alpha))$, implying that $\x(\alpha)$ is not an optimal solution to I-TAP$_{\alpha}$, a contradiction. \ifarxiv Thus, the flow $\Tilde{\x}$ cannot exist, proving our claim that any optimal solution to the linear program for the flow $\x(\alpha)$ must meet the capacity constraint with equality. \else Thus, $\x(\alpha)$ must meet the capacity constraints with equality. \fi
\end{proof}
\fi

Lemma~\ref{lem:single-od-result} implies that even when users are heterogeneous the edge flow $\x(\alpha)$ 
can be enforced as an equilibrium flow using tolls set through the dual variables of \ifarxiv the capacity constraints of \fi a linear program. \ifarxiv \else The proof of Lemma~\ref{lem:single-od-result} points to a more general condition for flow enforceability with heterogeneous users, which is highlighted in Appendix C.\fi
\ifarxiv
\begin{remark}
We note that in the proof of Lemma~\ref{lem:single-od-result}, all that we required was that the objective function of I-TAP$_{\alpha}$ is monotonically increasing in the flow on each edge of the network. This suggests a more general sufficient condition for enforcing flows as a user equilibrium. In particular, any flow that satisfies the flow conservation constraints and is the solution of a convex program with some convex objective $f(\x)$ that is monotonically increasing in $x_e$ for each $e \in E$ can be induced as a user equilibrium. Note that the I-TAP$_{\alpha}$ objective is a special case of such a function $f(\x)$.
\end{remark}
\fi

\section{Numerical Experiments} \label{sec:numerical-experiments}

We now evaluate the performance of our I-TAP method for \fso on several real-world transportation networks. The results of our experiments not only characterize the behavior of I-TAP but also highlight that, compared to the algorithm in \cite{so-routing-seminal}, our approach has much smaller runtimes while 
achieving lower total travel times for most levels $\beta$ of unfairness. \ifarxiv In the following, we describe the implementation details of the I-TAP method and the unfairness metric as well as the data-sets we use. We further present the corresponding results to evaluate the performance of our approach. \else We present the implementation details of the I-TAP method and the unfairness metric in Appendix D. In the following, we describe the data-sets we use and present the corresponding results to evaluate the performance of our approach. \fi


\ifarxiv
\subsection{Implementation Details and Data Sets} \label{sec:impl-details}
\else
\subsection{Data Sets} \label{sec:impl-details}
\fi
\ifarxiv
We
tested our I-TAP method using a single-thread C++ implementation of
the Conjugate Frank-Wolfe algorithm~\cite{MitradjievaLindberg13},
which we made publicly available along with our data sets and results ({\small\texttt{github.com\slash StanfordASL\slash \{\href{https://github.com/StanfordASL/frank-wolfe-traffic}{frank-wolfe-traffic}, \href{https://github.com/StanfordASL/fair-routing}{fair-routing}\}}}). Our implementation is based on a previous repository which was developed for~\cite{BuchETAL18}. While there is a rich literature on algorithm design to solve the traffic assignment problem~\cite{BARGERA2002,BARGERA20101022}, we decided to use the Frank-Wolfe algorithm which was shown recently to be superior in terms of running time~\cite{BuchETAL18}. For shortest-path computation in the all-or-nothing routine, we used the LEMON Graph Library~\cite{Lemon}. Within the same framework we implemented the solution method for CSO that was presented in~\cite{so-routing-seminal}, where we used 
\texttt{r\_c\_shortest\_paths} within
the Boost C++ Libraries~\cite{Boost}
for constrained shortest-path search.

To compute the resulting unfairness level for a given path flow $\f$ from those approaches we first recover a path-based solution by recording the paths computed for each commodity $k$ in every Frank-Wolfe iteration, and discarding paths whose relative weight in the final solution is negligible. Denote by $\P^+_k$ the resulting collection of paths for a given commodity $k$ with strictly positive flow. We also maintain for each path $P\in \P^+_k$ the volume of flow used by the path for this commodity. Then we recover for each commodity $k$ the edge-based solution by computing for each edge $e\in E$ the total flow resulting from the paths~$\P^+_k$. Then we discard edges from the graph whose flow is $0$ for the commodity $k$, which yields the DAG $G_k = (V_k,E_k)$, with vertices $V_k$ and edges $E_k$. 

Finally, to compute the unfairness level for this commodity, we compute the shortest and longest paths from origin to destination over the graph $G_k=(V_k,E_k)$, where a weight for a given edge $e\in E_k$ is set to be $t_e(x_e)$, i.e., the travel time on the edge given the flows of all the commodities combined with respect to the edge flow solution $\x$ corresponding to the path flow $\f$. To compute the shortest path over $G_k$ we simply run a Dijkstra search, whose running time is $\Theta(|E_k|+|V_k|\log |V_k|)$. Although for general graphs computing the longest path is NP-hard, for the case of a DAG, we can compute it for the same running time as Dijkstra by negating the edge weights, i.e., using the weights $-t_e(x_e)$ and then finding the shortest path~\cite{dag-book-algos}.

All results 
were obtained using a commodity laptop equipped with 2.80GHz 4-core i7-7600U CPU, and 16GB of RAM, running 64bit Ubuntu 20.04 OS. We ran the Frank-Wolfe algorithm for $100$ iterations on each data-set, both for I-TAP and the method in~\cite{so-routing-seminal}. We mention that 
this number of iterations generally allows the Frank-Wolfe algorithm to achieve a relative error of at most $10^{-5}$ when searching for UE and SO solutions over larger scenarios of the traffic assignment problem~\cite{MitradjievaLindberg13}.
\fi


Table~\ref{tab:problem-instances} shows the six instances we use for our study, which were obtained from \cite{tntp}. 
We use 
the BPR travel time function~\cite{Sheffi1985}, 
defined as
\ifarxiv 
\begin{align} \label{eq:BPR}
    t_e(x_e) = \xi_e \left(1+ a \left(\tfrac{x_e}{\kappa_e} \right)^b \right),
\end{align}
where $a,b$ are constants, $\xi_e$ is the free-flow travel time on edge $e$, and $\kappa_e$ is the capacity of edge $e$, which is the number of users beyond which the travel time on the edge rapidly increases. Since the constants $a = 0.15$ and $b = 4$ are typically chosen, we use these constants for the numerical experiments. 
\else
$t_e(x_e) = \xi_e \big(1+ 0.15 \big(\tfrac{x_e}{\kappa_e} \big)^4 \big)$,
where $\xi_e$ is the free-flow travel time on edge $e$, and $\kappa_e$ is the capacity of edge $e$, which is the number of users beyond which the travel time on the edge rapidly increases.
\fi 

\begin{table}[t] 
\centering
\caption{{\small \sf Problem instance attributes and computation time. For each instance we report the number of vertices $|V|$, edges $|E|$, and 
OD pairs $|K|$. In addition, we report the computation time of each instance for the previous method of Jahn et al.~\cite{so-routing-seminal} and our I-TAP method using 100 iterations of the Frank-Wolfe algorithm. }  }
\footnotesize
\begin{tabular}{l|ccc|cc}
\toprule

 &      \multicolumn{3}{|c|}{attributes} & \multicolumn{2}{c}{runtime (sec.)}    \\ 
Region Name  & $|V|$ & $|E|$ & $|K|$ & Jahn et al. & I-TAP \\
\midrule 
Sioux Falls (SF) & 24 & 76 & 528 & 20.0 & 0.03 \\
Anaheim (A) & 416 & 914  & 1406 & 74.0 & 0.33 \\
Massachusetts (M) & 74 & 258 & 1113  & 24.3 & 0.09 \\
Tiergarten (T)  & 361 & 766 & 644  & 18.2 & 0.20 \\
Friedrichshain (F)  & 224 & 523 & 506  & 19.8 & 0.12 \\
Prenzlauerberg (P)\  & 352 & 749 & 1406  & 74.4 & 0.32 \\
\bottomrule
\end{tabular} \label{tab:problem-instances}
\end{table}

\subsection{Results} \label{sec:numerical-results}


\fakeparagraph{Assessment of Theoretical Upper Bounds.} 
\ifarxiv
We first assess the theoretical upper bounds on the inefficiency ratio and level of unfairness that were obtained in Section~\ref{sec:eff-properties} with respect to the convex combination parameter~$\alpha$. The latter is obtained using a dense sampling method with increments of $0.01$. We present the results for the Prenzlauerberg data-set and note that the results and the following discussion extend to other problem instances in Table~\ref{tab:problem-instances} as well.
\else 
We now assess the theoretical upper bounds on the inefficiency ratio and unfairness that were obtained in Section~\ref{sec:eff-properties}. We present the results for the Prenzlauerberg data-set and note that the results and the following discussion extend to other problem instances in Table~\ref{tab:problem-instances} as well.
\fi


\ifarxiv
Figure~\ref{fig:theory_bounds} (left) depicts both (i) the change in the inefficiency ratio of the solution $\x(\alpha)$ using dense sampling
and (ii) the theoretical upper bound of the inefficiency ratio (Theorem~\ref{thm:eff-upper-bound}). As expected, the dense sampling procedure results in an inefficiency ratio that is below the theoretical upper bound for every value of $\alpha$.  

The comparison between the theoretically guaranteed level of unfairness for every value of $\alpha$ as obtained in Theorem~\ref{thm:poly-ttf-result} and the actual unfairness level of the I-TAP method is depicted in Figure~\ref{fig:theory_bounds} (right). Since the BPR travel time function we used in this work has a degree of four, we have that for any value of $\alpha$, we can guarantee a level of unfairness of $4 \alpha + 1$ by Theorem~\ref{thm:poly-ttf-result}. 
Figure~\ref{fig:theory_bounds} (right) suggests that the theoretical upper bound is even more conservative for the case of unfairness as there is an even larger gap between the actual solution and the theoretical bound. 

These findings further highlight the efficacy of the I-TAP approach for practical applications. 
\else
Figure~\ref{fig:theory_bounds} depicts both (i) the change in the inefficiency ratio (left) and unfairness (right) of the solution $\x(\alpha)$ using dense sampling, and (ii) the theoretical upper bound of the inefficiency ratio (Theorem~\ref{thm:eff-upper-bound}) and unfairness (Theorem~\ref{thm:poly-ttf-result}). Since the BPR travel time function we used in this work has a degree of four, Theorem~\ref{thm:poly-ttf-result} implies that for any $\alpha$ we can guarantee an unfairness of $4 \alpha + 1$, as depicted in Figure~\ref{fig:theory_bounds} (right). As expected, the dense sampling procedure results in both an inefficiency ratio and unfairness that is below the theoretical upper bound for every value of $\alpha$, which highlights the efficacy of the I-TAP approach for practical applications.



\fi

\newcommand{\sfwidth}{0.5\columnwidth}
\newcommand{\siwidth}{0.99\linewidth}
\ifarxiv
\begin{figure*}
    \centering
    \begin{subfigure}[t]{0.48\columnwidth}
        \centering

\begin{tikzpicture}

\definecolor{color0}{rgb}{0.12156862745098,0.466666666666667,0.705882352941177}
\definecolor{color1}{rgb}{1,0.498039215686275,0.0549019607843137}

\begin{axis}[
width=\siwidth,
height=2.1in,
legend cell align={left},
legend style={fill opacity=0.8, draw opacity=1, text opacity=1, draw=white!80!black},
tick align=outside,
tick pos=left,
x grid style={white!69.0196078431373!black},
xlabel={Convex Combination Parameter \(\displaystyle \alpha\)},
xmin=-0.05, xmax=1.05,
xticklabels={0.0, 0.0, 0.2, 0.4, 0.6, 0.8, 1.0},
xtick style={color=black},
y grid style={white!69.0196078431373!black},
ylabel style={align=center},
ylabel={Inefficiency Ratio \\ $\rho(\mathbf{x}(\alpha))$},
ymin=0.995, ymax=1.05,
yticklabels={0.00, 0.00, 1.00, 1.02, 1.04},
ytick style={color=black},
style={font=\footnotesize}
]
\addplot [dashed, line width=1.7pt, color=color0]
table {%
0 1.03651678860029
0.01 1.03609910591055
0.02 1.03324009505072
0.03 1.03262191137097
0.04 1.03070469010758
0.05 1.02974824992738
0.06 1.02846964152058
0.07 1.02815256911501
0.08 1.0271488476397
0.09 1.02616145688888
0.1 1.02519135247545
0.11 1.02402686088963
0.12 1.02253820801686
0.13 1.02197753879526
0.14 1.02106097911479
0.15 1.02011538008418
0.16 1.01925951758568
0.17 1.01857902037866
0.18 1.01751183926772
0.19 1.01650405429168
0.2 1.01551141297642
0.21 1.01459799542885
0.22 1.01373128304158
0.23 1.01290304548316
0.24 1.01212760365757
0.25 1.01155432680775
0.26 1.01109739429685
0.27 1.01043627542453
0.28 1.0100535853425
0.29 1.00970523818567
0.3 1.00938790295026
0.31 1.00905742102085
0.32 1.00870836477815
0.33 1.00834031240536
0.34 1.00814606041963
0.35 1.00785953855725
0.36 1.00744243346307
0.37 1.00665476361315
0.38 1.00674059104893
0.39 1.00624079350153
0.4 1.00605671284306
0.41 1.00573026949065
0.42 1.00540102682032
0.43 1.00516554968778
0.44 1.00489950778057
0.45 1.00459840672705
0.46 1.00437759739715
0.47 1.00400041239584
0.48 1.00389769493426
0.49 1.00367829773927
0.5 1.00346939992974
0.51 1.00328433343489
0.52 1.00309205670638
0.53 1.00288039754718
0.54 1.00275519248242
0.55 1.00252030332708
0.56 1.00242025608268
0.57 1.00223409913522
0.58 1.00213584581836
0.59 1.00199737351715
0.6 1.00170385232456
0.61 1.00176107270509
0.62 1.00167352892064
0.63 1.00154169952747
0.64 1.00143690974845
0.65 1.00134640343252
0.66 1.00125057169335
0.67 1.00113991898295
0.68 1.00108426545717
0.69 1.00100518035166
0.7 1.00094283123062
0.71 1.0008038293171
0.72 1.000814934642
0.73 1.0007357139548
0.74 1.00072900055113
0.75 1.00062683866968
0.76 1.00058125377237
0.77 1.00054202948627
0.78 1.00059355959156
0.79 1.0004363768183
0.8 1.00039700683399
0.81 1.00036457131249
0.82 1.00033337407348
0.83 1.00028743723977
0.84 1.00024293465843
0.85 1.00021761890909
0.86 1.00018517052544
0.87 1.00015872572251
0.88 1.0001245120497
0.89 1.00011404126433
0.9 1.00009203078653
0.91 1.00006288667275
0.92 1.0000523346388
0.93 1.00004196260368
0.94 1.00004259544057
0.95 1.00001302923219
0.96 1.00000592511515
0.97 1.00000545430511
0.98 0.999992809467093
0.99 1.00000772372461
1 1
};
\addlegendentry{Sampling}
\addplot [dashed, line width=1.7pt, color=color1]
table {%
0 1.03651678860029
0.01 1.03651678860029
0.02 1.03651678860029
0.03 1.03651678860029
0.04 1.03651678860029
0.05 1.03651678860029
0.06 1.03651678860029
0.07 1.03651678860029
0.08 1.03651678860029
0.09 1.03651678860029
0.1 1.03651678860029
0.11 1.03651678860029
0.12 1.03651678860029
0.13 1.03651678860029
0.14 1.03651678860029
0.15 1.03651678860029
0.16 1.03651678860029
0.17 1.03651678860029
0.18 1.03651678860029
0.19 1.03651678860029
0.2 1.03651678860029
0.21 1.03651678860029
0.22 1.03651678860029
0.23 1.03651678860029
0.24 1.03651678860029
0.25 1.03651678860029
0.26 1.03651678860029
0.27 1.03651678860029
0.28 1.03651678860029
0.29 1.03651678860029
0.3 1.03651678860029
0.31 1.03651678860029
0.32 1.03562115772471
0.33 1.03403376210418
0.34 1.03253974269662
0.35 1.03113109582664
0.36 1.0298007071161
0.37 1.02854223130883
0.38 1.02734999107037
0.39 1.02621889135695
0.4 1.02514434662921
0.41 1.02412221871745
0.42 1.0231487635634
0.43 1.02222058539325
0.44 1.02133459713993
0.45 1.02048798614232
0.46 1.01967818431851
0.47 1.01890284214678
0.48 1.01815980589887
0.49 1.01744709766108
0.5 1.01676289775281
0.51 1.01610552921348
0.52 1.01547344407951
0.53 1.01486521121475
0.54 1.01427950549313
0.55 1.01371509816139
0.56 1.01317084823435
0.57 1.01264569479598
0.58 1.01213865009686
0.59 1.01164879335364
0.6 1.01117526516854
0.61 1.0107172624977
0.62 1.01027403410656
0.63 1.009844876458
0.64 1.00942912998595
0.65 1.00902617571305
0.66 1.00863543217569
0.67 1.00825635262452
0.68 1.00788842247191
0.69 1.00753115696141
0.7 1.00718409903692
0.71 1.00684681739199
0.72 1.00651890468165
0.73 1.00619997588117
0.74 1.00588966677801
0.75 1.00558763258427
0.76 1.00529354665878
0.77 1.00500709932876
0.78 1.00472799680207
0.79 1.00445596016214
0.8 1.0041907244382
0.81 1.00393203774449
0.82 1.00367966048232
0.83 1.00343336459997
0.84 1.0031929329053
0.85 1.00295815842697
0.86 1.00272884382022
0.87 1.00250480081364
0.88 1.00228584969356
0.89 1.00207181882338
0.9 1.00186254419476
0.91 1.00165786900852
0.92 1.00145764328285
0.93 1.00126172348677
0.94 1.00106997219699
0.95 1.00088225777646
0.96 1.00069845407303
0.97 1.00051844013668
0.98 1.00034209995414
0.99 1.00016932219952
1 1
};
\addlegendentry{Theory Bound}
\end{axis}

\end{tikzpicture}
    \end{subfigure}
    \begin{subfigure}[t]{0.48\columnwidth}
        \centering
        \vspace{2.5pt}

\begin{tikzpicture}

\definecolor{color0}{rgb}{0.12156862745098,0.466666666666667,0.705882352941177}
\definecolor{color1}{rgb}{1,0.498039215686275,0.0549019607843137}

\begin{axis}[
width=\siwidth,
height=2.1in,
legend cell align={left},
legend style={
  fill opacity=0.8,
  draw opacity=1,
  text opacity=1,
  at={(0.03,0.97)},
  anchor=north west,
  draw=white!80!black
},
tick align=outside,
tick pos=left,
x grid style={white!69.0196078431373!black},
xlabel={Convex Combination Parameter \(\displaystyle \alpha\)},
xmin=-0.05, xmax=1.05,
xticklabels={0.0, 0.0, 0.2, 0.4, 0.6, 0.8, 1.0},
xtick style={color=black},
y grid style={white!69.0196078431373!black},
ylabel style={align=center},
ylabel={Unfairness \\ $U(\x(\alpha))$},
ymin=0.8, ymax=5.2,
ytick style={color=black},
style={font=\footnotesize}
]
\addplot [dashed, line width=1.7pt, color=color0]
table {%
0 1.03672245291724
0.01 1.04445746744004
0.02 1.0226344847168
0.03 1.02550179776974
0.04 1.04287080855659
0.05 1.05076227038471
0.06 1.05786024074513
0.07 1.05388727071992
0.08 1.06119700607259
0.09 1.18628935642275
0.1 1.1967912039846
0.11 1.21395587870278
0.12 1.23796790885949
0.13 1.24537947260722
0.14 1.26033373518824
0.15 1.2767270448256
0.16 1.29039888398771
0.17 1.30500702749838
0.18 1.31737457556133
0.19 1.33057673564651
0.2 1.34368011810848
0.21 1.35663634143437
0.22 1.3681592320483
0.23 1.38070454904429
0.24 1.39230905047826
0.25 1.40438803836272
0.26 1.41533613168528
0.27 1.425975934206
0.28 1.43663419499817
0.29 1.44682964565156
0.3 1.45762548231361
0.31 1.46741657602615
0.32 1.47694577100812
0.33 1.49076795080293
0.34 1.4956413097021
0.35 1.50461904869886
0.36 1.51386749588823
0.37 1.53805306958992
0.38 1.53111300536628
0.39 1.53922884253374
0.4 1.54757725236952
0.41 1.55435289610883
0.42 1.56605753078536
0.43 1.57157095627726
0.44 1.57916776911221
0.45 1.58572771075022
0.46 1.59375831823923
0.47 1.6065131196664
0.48 1.6081559088087
0.49 1.61516803979347
0.5 1.62210347703546
0.51 1.62870805659832
0.52 1.63543666623415
0.53 1.64245871468799
0.54 1.64757384713466
0.55 1.6585006939462
0.56 1.66111949639228
0.57 1.67076106669155
0.58 1.67191069770475
0.59 1.67745164444213
0.6 1.69040935023902
0.61 1.68872682281297
0.62 1.69522925412295
0.63 1.69990499073641
0.64 1.70503389001784
0.65 1.71044075716925
0.66 1.71547390222687
0.67 1.72127914579171
0.68 1.72532145913196
0.69 1.73053553677932
0.7 1.73459477204077
0.71 1.76698958420615
0.72 1.73931791357628
0.73 1.74943615581081
0.74 1.75189420260262
0.75 1.75815161279011
0.76 1.76270878449534
0.77 1.76711023325426
0.78 1.76210366503677
0.79 1.77589082131155
0.8 1.77959337551325
0.81 1.7840882476894
0.82 1.78647076475064
0.83 1.79396614573188
0.84 1.79564229911913
0.85 1.80052617131768
0.86 1.80340139370984
0.87 1.80511885674786
0.88 1.81108400933928
0.89 1.8137377847152
0.9 1.81775971141073
0.91 1.82209282795821
0.92 1.82877417629396
0.93 1.82835601887225
0.94 1.83723506461362
0.95 1.83525901128958
0.96 1.8391186241206
0.97 1.84214251019417
0.98 1.84544253125013
0.99 1.84849554216512
1 1.85219020369235
};
\addplot [dashed, line width=1.7pt, color=color1]
table {%
0 1
0.01 1.04
0.02 1.08
0.03 1.12
0.04 1.16
0.05 1.2
0.06 1.24
0.07 1.28
0.08 1.32
0.09 1.36
0.1 1.4
0.11 1.44
0.12 1.48
0.13 1.52
0.14 1.56
0.15 1.6
0.16 1.64
0.17 1.68
0.18 1.72
0.19 1.76
0.2 1.8
0.21 1.84
0.22 1.88
0.23 1.92
0.24 1.96
0.25 2
0.26 2.04
0.27 2.08
0.28 2.12
0.29 2.16
0.3 2.2
0.31 2.24
0.32 2.28
0.33 2.32
0.34 2.36
0.35 2.4
0.36 2.44
0.37 2.48
0.38 2.52
0.39 2.56
0.4 2.6
0.41 2.64
0.42 2.68
0.43 2.72
0.44 2.76
0.45 2.8
0.46 2.84
0.47 2.88
0.48 2.92
0.49 2.96
0.5 3
0.51 3.04
0.52 3.08
0.53 3.12
0.54 3.16
0.55 3.2
0.56 3.24
0.57 3.28
0.58 3.32
0.59 3.36
0.6 3.4
0.61 3.44
0.62 3.48
0.63 3.52
0.64 3.56
0.65 3.6
0.66 3.64
0.67 3.68
0.68 3.72
0.69 3.76
0.7 3.8
0.71 3.84
0.72 3.88
0.73 3.92
0.74 3.96
0.75 4
0.76 4.04
0.77 4.08
0.78 4.12
0.79 4.16
0.8 4.2
0.81 4.24
0.82 4.28
0.83 4.32
0.84 4.36
0.85 4.4
0.86 4.44
0.87 4.48
0.88 4.52
0.89 4.56
0.9 4.6
0.91 4.64
0.92 4.68
0.93 4.72
0.94 4.76
0.95 4.8
0.96 4.84
0.97 4.88
0.98 4.92
0.99 4.96
1 5
};
\end{axis}

\end{tikzpicture}
    \end{subfigure}
    \vspace{-15pt}
    \caption{{\small \sf Comparison between the inefficiency ratio of the solution $\x(\alpha)$ of I-TAP on the Prenzlauerberg data-set and the theoretical bound on the inefficiency ratio obtained in Theorem~\ref{thm:eff-upper-bound} (left). The comparison between the theoretical upper bound of the level of unfairness for any convex combination parameter $\alpha$ and the unfairness level of the solution $\x(\alpha)$ of I-TAP on the Prenzlauerberg data-set is shown on the right. The values of the convex combination parameter were chosen at increments of $0.01$.}
    } 
    \label{fig:theory_bounds}
\end{figure*}
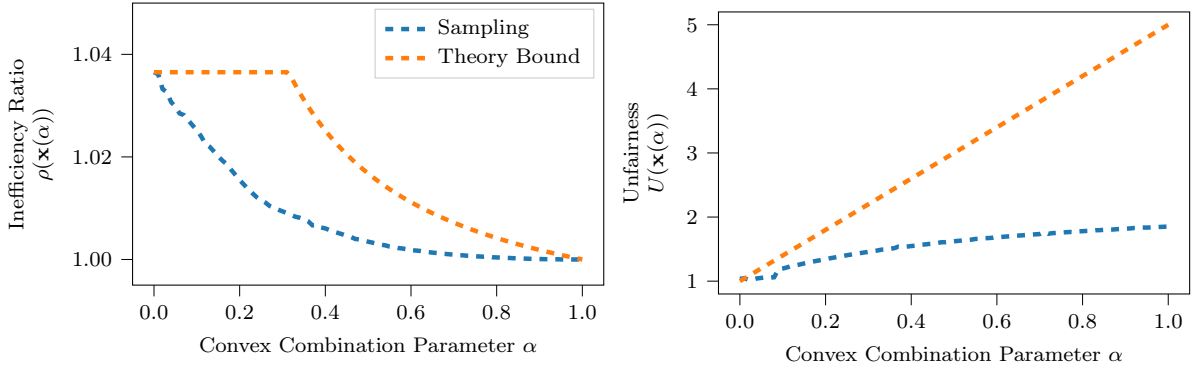
\else
\begin{figure}
    \centering
    \begin{subfigure}[t]{0.495\columnwidth}
        \centering

\begin{tikzpicture}

\definecolor{color0}{rgb}{0.12156862745098,0.466666666666667,0.705882352941177}
\definecolor{color1}{rgb}{1,0.498039215686275,0.0549019607843137}

\begin{axis}[
width=\siwidth,
height=1.4in,
legend cell align={left},
legend style={fill opacity=0.8, draw opacity=1, text opacity=1, draw=white!80!black},
tick align=outside,
tick pos=left,
x grid style={white!69.0196078431373!black},
xlabel={Convex Combination Parameter \(\displaystyle \alpha\)},
xmin=-0.05, xmax=1.05,
xticklabels={0.0, 0.0, 0.5, 1.0},
xtick style={color=black},
y grid style={white!69.0196078431373!black},
ylabel style={align=center},
ylabel={Inefficiency Ratio \\ $\rho(\mathbf{x}(\alpha))$},
ymin=0.995, ymax=1.08,
yticklabels={0.00, 0.00, 1.00, 1.02, 1.04, 1.06, 1.08},
ytick style={color=black},
style={font=\footnotesize}
]
\addplot [dashed, line width=1.7pt, color=color0]
table {%
0 1.03651678860029
0.01 1.03609910591055
0.02 1.03324009505072
0.03 1.03262191137097
0.04 1.03070469010758
0.05 1.02974824992738
0.06 1.02846964152058
0.07 1.02815256911501
0.08 1.0271488476397
0.09 1.02616145688888
0.1 1.02519135247545
0.11 1.02402686088963
0.12 1.02253820801686
0.13 1.02197753879526
0.14 1.02106097911479
0.15 1.02011538008418
0.16 1.01925951758568
0.17 1.01857902037866
0.18 1.01751183926772
0.19 1.01650405429168
0.2 1.01551141297642
0.21 1.01459799542885
0.22 1.01373128304158
0.23 1.01290304548316
0.24 1.01212760365757
0.25 1.01155432680775
0.26 1.01109739429685
0.27 1.01043627542453
0.28 1.0100535853425
0.29 1.00970523818567
0.3 1.00938790295026
0.31 1.00905742102085
0.32 1.00870836477815
0.33 1.00834031240536
0.34 1.00814606041963
0.35 1.00785953855725
0.36 1.00744243346307
0.37 1.00665476361315
0.38 1.00674059104893
0.39 1.00624079350153
0.4 1.00605671284306
0.41 1.00573026949065
0.42 1.00540102682032
0.43 1.00516554968778
0.44 1.00489950778057
0.45 1.00459840672705
0.46 1.00437759739715
0.47 1.00400041239584
0.48 1.00389769493426
0.49 1.00367829773927
0.5 1.00346939992974
0.51 1.00328433343489
0.52 1.00309205670638
0.53 1.00288039754718
0.54 1.00275519248242
0.55 1.00252030332708
0.56 1.00242025608268
0.57 1.00223409913522
0.58 1.00213584581836
0.59 1.00199737351715
0.6 1.00170385232456
0.61 1.00176107270509
0.62 1.00167352892064
0.63 1.00154169952747
0.64 1.00143690974845
0.65 1.00134640343252
0.66 1.00125057169335
0.67 1.00113991898295
0.68 1.00108426545717
0.69 1.00100518035166
0.7 1.00094283123062
0.71 1.0008038293171
0.72 1.000814934642
0.73 1.0007357139548
0.74 1.00072900055113
0.75 1.00062683866968
0.76 1.00058125377237
0.77 1.00054202948627
0.78 1.00059355959156
0.79 1.0004363768183
0.8 1.00039700683399
0.81 1.00036457131249
0.82 1.00033337407348
0.83 1.00028743723977
0.84 1.00024293465843
0.85 1.00021761890909
0.86 1.00018517052544
0.87 1.00015872572251
0.88 1.0001245120497
0.89 1.00011404126433
0.9 1.00009203078653
0.91 1.00006288667275
0.92 1.0000523346388
0.93 1.00004196260368
0.94 1.00004259544057
0.95 1.00001302923219
0.96 1.00000592511515
0.97 1.00000545430511
0.98 0.999992809467093
0.99 1.00000772372461
1 1
};
\addlegendentry{Sampling}
\addplot [dashed, line width=1.7pt, color=color1]
table {%
0 1.03651678860029
0.01 1.03651678860029
0.02 1.03651678860029
0.03 1.03651678860029
0.04 1.03651678860029
0.05 1.03651678860029
0.06 1.03651678860029
0.07 1.03651678860029
0.08 1.03651678860029
0.09 1.03651678860029
0.1 1.03651678860029
0.11 1.03651678860029
0.12 1.03651678860029
0.13 1.03651678860029
0.14 1.03651678860029
0.15 1.03651678860029
0.16 1.03651678860029
0.17 1.03651678860029
0.18 1.03651678860029
0.19 1.03651678860029
0.2 1.03651678860029
0.21 1.03651678860029
0.22 1.03651678860029
0.23 1.03651678860029
0.24 1.03651678860029
0.25 1.03651678860029
0.26 1.03651678860029
0.27 1.03651678860029
0.28 1.03651678860029
0.29 1.03651678860029
0.3 1.03651678860029
0.31 1.03651678860029
0.32 1.03562115772471
0.33 1.03403376210418
0.34 1.03253974269662
0.35 1.03113109582664
0.36 1.0298007071161
0.37 1.02854223130883
0.38 1.02734999107037
0.39 1.02621889135695
0.4 1.02514434662921
0.41 1.02412221871745
0.42 1.0231487635634
0.43 1.02222058539325
0.44 1.02133459713993
0.45 1.02048798614232
0.46 1.01967818431851
0.47 1.01890284214678
0.48 1.01815980589887
0.49 1.01744709766108
0.5 1.01676289775281
0.51 1.01610552921348
0.52 1.01547344407951
0.53 1.01486521121475
0.54 1.01427950549313
0.55 1.01371509816139
0.56 1.01317084823435
0.57 1.01264569479598
0.58 1.01213865009686
0.59 1.01164879335364
0.6 1.01117526516854
0.61 1.0107172624977
0.62 1.01027403410656
0.63 1.009844876458
0.64 1.00942912998595
0.65 1.00902617571305
0.66 1.00863543217569
0.67 1.00825635262452
0.68 1.00788842247191
0.69 1.00753115696141
0.7 1.00718409903692
0.71 1.00684681739199
0.72 1.00651890468165
0.73 1.00619997588117
0.74 1.00588966677801
0.75 1.00558763258427
0.76 1.00529354665878
0.77 1.00500709932876
0.78 1.00472799680207
0.79 1.00445596016214
0.8 1.0041907244382
0.81 1.00393203774449
0.82 1.00367966048232
0.83 1.00343336459997
0.84 1.0031929329053
0.85 1.00295815842697
0.86 1.00272884382022
0.87 1.00250480081364
0.88 1.00228584969356
0.89 1.00207181882338
0.9 1.00186254419476
0.91 1.00165786900852
0.92 1.00145764328285
0.93 1.00126172348677
0.94 1.00106997219699
0.95 1.00088225777646
0.96 1.00069845407303
0.97 1.00051844013668
0.98 1.00034209995414
0.99 1.00016932219952
1 1
};
\addlegendentry{Theory Bound}
\end{axis}

\end{tikzpicture}
    \end{subfigure}
    \begin{subfigure}[t]{0.495\columnwidth}
        \centering
        \vspace{2.5pt}

\begin{tikzpicture}

\definecolor{color0}{rgb}{0.12156862745098,0.466666666666667,0.705882352941177}
\definecolor{color1}{rgb}{1,0.498039215686275,0.0549019607843137}

\begin{axis}[
width=\siwidth,
height=1.4in,
legend cell align={left},
legend style={
  fill opacity=0.8,
  draw opacity=1,
  text opacity=1,
  at={(0.03,0.97)},
  anchor=north west,
  draw=white!80!black
},
tick align=outside,
tick pos=left,
x grid style={white!69.0196078431373!black},
xlabel={Convex Combination Parameter \(\displaystyle \alpha\)},
xmin=-0.05, xmax=1.05,
xticklabels={0.0, 0.0, 0.5, 1.0},
xtick style={color=black},
y grid style={white!69.0196078431373!black},
ylabel style={align=center},
ylabel={Unfairness \\ $U(\f(\alpha))$},
ymin=0.8, ymax=5.2,
ytick style={color=black},
style={font=\footnotesize}
]
\addplot [dashed, line width=1.7pt, color=color0]
table {%
0 1.03672245291724
0.01 1.04445746744004
0.02 1.0226344847168
0.03 1.02550179776974
0.04 1.04287080855659
0.05 1.05076227038471
0.06 1.05786024074513
0.07 1.05388727071992
0.08 1.06119700607259
0.09 1.18628935642275
0.1 1.1967912039846
0.11 1.21395587870278
0.12 1.23796790885949
0.13 1.24537947260722
0.14 1.26033373518824
0.15 1.2767270448256
0.16 1.29039888398771
0.17 1.30500702749838
0.18 1.31737457556133
0.19 1.33057673564651
0.2 1.34368011810848
0.21 1.35663634143437
0.22 1.3681592320483
0.23 1.38070454904429
0.24 1.39230905047826
0.25 1.40438803836272
0.26 1.41533613168528
0.27 1.425975934206
0.28 1.43663419499817
0.29 1.44682964565156
0.3 1.45762548231361
0.31 1.46741657602615
0.32 1.47694577100812
0.33 1.49076795080293
0.34 1.4956413097021
0.35 1.50461904869886
0.36 1.51386749588823
0.37 1.53805306958992
0.38 1.53111300536628
0.39 1.53922884253374
0.4 1.54757725236952
0.41 1.55435289610883
0.42 1.56605753078536
0.43 1.57157095627726
0.44 1.57916776911221
0.45 1.58572771075022
0.46 1.59375831823923
0.47 1.6065131196664
0.48 1.6081559088087
0.49 1.61516803979347
0.5 1.62210347703546
0.51 1.62870805659832
0.52 1.63543666623415
0.53 1.64245871468799
0.54 1.64757384713466
0.55 1.6585006939462
0.56 1.66111949639228
0.57 1.67076106669155
0.58 1.67191069770475
0.59 1.67745164444213
0.6 1.69040935023902
0.61 1.68872682281297
0.62 1.69522925412295
0.63 1.69990499073641
0.64 1.70503389001784
0.65 1.71044075716925
0.66 1.71547390222687
0.67 1.72127914579171
0.68 1.72532145913196
0.69 1.73053553677932
0.7 1.73459477204077
0.71 1.76698958420615
0.72 1.73931791357628
0.73 1.74943615581081
0.74 1.75189420260262
0.75 1.75815161279011
0.76 1.76270878449534
0.77 1.76711023325426
0.78 1.76210366503677
0.79 1.77589082131155
0.8 1.77959337551325
0.81 1.7840882476894
0.82 1.78647076475064
0.83 1.79396614573188
0.84 1.79564229911913
0.85 1.80052617131768
0.86 1.80340139370984
0.87 1.80511885674786
0.88 1.81108400933928
0.89 1.8137377847152
0.9 1.81775971141073
0.91 1.82209282795821
0.92 1.82877417629396
0.93 1.82835601887225
0.94 1.83723506461362
0.95 1.83525901128958
0.96 1.8391186241206
0.97 1.84214251019417
0.98 1.84544253125013
0.99 1.84849554216512
1 1.85219020369235
};
\addplot [dashed, line width=1.7pt, color=color1]
table {%
0 1
0.01 1.04
0.02 1.08
0.03 1.12
0.04 1.16
0.05 1.2
0.06 1.24
0.07 1.28
0.08 1.32
0.09 1.36
0.1 1.4
0.11 1.44
0.12 1.48
0.13 1.52
0.14 1.56
0.15 1.6
0.16 1.64
0.17 1.68
0.18 1.72
0.19 1.76
0.2 1.8
0.21 1.84
0.22 1.88
0.23 1.92
0.24 1.96
0.25 2
0.26 2.04
0.27 2.08
0.28 2.12
0.29 2.16
0.3 2.2
0.31 2.24
0.32 2.28
0.33 2.32
0.34 2.36
0.35 2.4
0.36 2.44
0.37 2.48
0.38 2.52
0.39 2.56
0.4 2.6
0.41 2.64
0.42 2.68
0.43 2.72
0.44 2.76
0.45 2.8
0.46 2.84
0.47 2.88
0.48 2.92
0.49 2.96
0.5 3
0.51 3.04
0.52 3.08
0.53 3.12
0.54 3.16
0.55 3.2
0.56 3.24
0.57 3.28
0.58 3.32
0.59 3.36
0.6 3.4
0.61 3.44
0.62 3.48
0.63 3.52
0.64 3.56
0.65 3.6
0.66 3.64
0.67 3.68
0.68 3.72
0.69 3.76
0.7 3.8
0.71 3.84
0.72 3.88
0.73 3.92
0.74 3.96
0.75 4
0.76 4.04
0.77 4.08
0.78 4.12
0.79 4.16
0.8 4.2
0.81 4.24
0.82 4.28
0.83 4.32
0.84 4.36
0.85 4.4
0.86 4.44
0.87 4.48
0.88 4.52
0.89 4.56
0.9 4.6
0.91 4.64
0.92 4.68
0.93 4.72
0.94 4.76
0.95 4.8
0.96 4.84
0.97 4.88
0.98 4.92
0.99 4.96
1 5
};
\end{axis}

\end{tikzpicture}
    \end{subfigure}
    \vspace{-15pt}
    \caption{{\small \sf Comparison between the inefficiency ratio (left) and unfairness (right) of the solution $\x(\alpha)$ of I-TAP sampling on the Prenzlauerberg data-set and the theoretical bounds obtained in Theorems~\ref{thm:eff-upper-bound} and~\ref{thm:poly-ttf-result}. 
    The convex combination parameters were chosen at increments of $0.01$.}
    } \vspace{-10pt} 
    \label{fig:theory_bounds}
\end{figure}
\fi

\fakeparagraph{Behavior of I-TAP.} For each of the transportation networks in Table~\ref{tab:problem-instances}, we now study the relationship between the convex combination parameter $\alpha$ and the (i) total travel time, and (ii) unfairness. \ifarxiv To review these relationships, we consider $\alpha$ to lie in the set $\A_s$ with $s = 0.01$ increments. \fi


Figure~\ref{fig:tt_v_alpha_and_beta_v_alpha_finer_discretization} (left) shows the relationship between the inefficiency ratio and 
$\alpha$. Note that when $\alpha = 1$, the inefficiency ratio is one, since the interpolated objective is the SO-TAP objective, and when $\alpha = 0$, the inefficiency ratio is the Price of Anarchy (PoA), since the interpolated objective is the UE-TAP objective. 
\ifarxiv  As shown on the left in Figure~\ref{fig:tt_v_alpha_and_beta_v_alpha_finer_discretization}, the inefficiency ratio is always between one and the PoA for each of the transportation networks, which corroborates the bound on the total travel time for any convex combination parameter, as obtained in Theorem~\ref{thm:eff-upper-bound}. Furthermore, the inefficiency ratio varies continuously in 
$\alpha$, which further aligns with the continuity of the total travel time function $\objso$ in $\alpha$, as is characterized in Section~\ref{sec:solution-properties}. The jumps in the inefficiency ratio that can be observed for certain values of $\alpha$ 
for Sioux Falls aligns with the continuity result since the relative magnitude of the jumps is small. In particular, the change in the total travel time is less than 2\% for a 1\% change in the value of $\alpha$.
\else 
Furthermore, the inefficiency ratio varies continuously in $\alpha$, which aligns with the continuity of the travel time function $\objso$ in $\alpha$ (see Appendix B) and can be observed since the relative magnitude of the jumps in the inefficiency ratio for small changes in $\alpha$ is small.
\fi    

The relationship between unfairness and 
$\alpha$ is depicted on the right in Figure~\ref{fig:tt_v_alpha_and_beta_v_alpha_finer_discretization}. For readability of this figure, we marked outliers as points where large changes in the unfairness occur for small changes in $\alpha$. 
For an explanation of the jumps in the unfairness at certain values of $\alpha$, see the discussion in \ifarxiv Section~\ref{sec:solution-properties}\else Appendix B\fi.

Finally, for each transportation network the general trend of a decrease in the inefficiency ratio and an increase in the unfairness with an increase in $\alpha$ suggests that decreasing the total travel time comes at the cost of an increase in the unfairness and vice versa.



\begin{figure*}
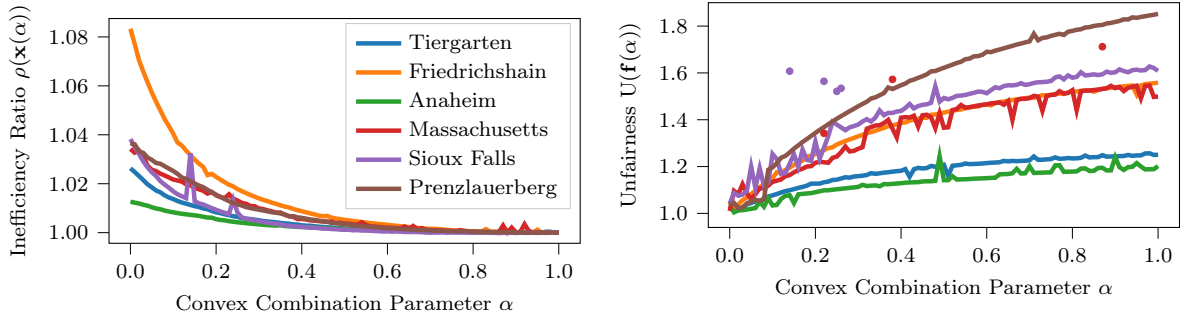

    \centering
    \begin{subfigure}[t]{0.48\columnwidth}
        \centering
            \include{Fig/cost_v_alpha_0.01}
    \end{subfigure}
    \begin{subfigure}[t]{0.48\columnwidth}
        \centering
        \vspace{2.5pt}
        \include{Fig/alpha_v_beta_outliers}
    \end{subfigure}
    \vspace{-15pt}
    \caption{{\small \sf Variation in the inefficiency ratio (left) and the level of unfairness (right) of the optimal solution $\mathbf{x}(\alpha)$ of I-TAP$_{\alpha}$ with the parameter $\alpha \in [0, 1]$ for six different transportation networks from Table~\ref{tab:problem-instances}. The values of the convex combination parameter were chosen at increments of $0.01$.}
    } 
    \label{fig:tt_v_alpha_and_beta_v_alpha_finer_discretization}
\end{figure*}




\fakeparagraph{Solution Quality Comparison.} We now explore the efficiency-fairness tradeoff through a comparison of the Pareto frontier of the I-TAP method to the approach in \cite{so-routing-seminal}, which is a benchmark solution for fair traffic routing, and the I-Solution method described in Section~\ref{sec:algos}. 
To this end, we depict the Pareto frontier of the (i) I-TAP method for $0.01$ and $0.05$
increments of the 
parameter \ifarxiv$\alpha$, \else $\alpha$ (see Appendix E for a comparison between these Pareto frontiers), \fi (ii) I-Solution method for $0.01$ increments of the convex combination parameter $\gamma$, and (iii) Jahn et al.'s approach \cite{so-routing-seminal} for $0.05$ increments of the normal unfairness parameter $\phi$ lying between one and two. 




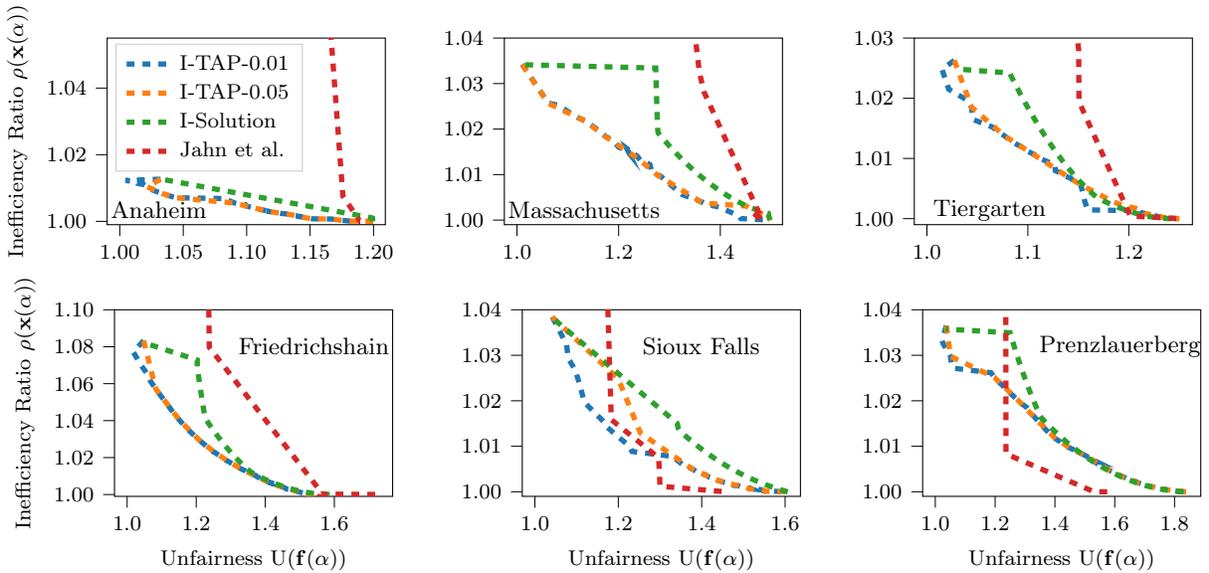
\begin{figure*}
    \centering
    \begin{subfigure}[btp]{0.323\columnwidth}
        \centering
\begin{tikzpicture}

\definecolor{color0}{rgb}{0.12156862745098,0.466666666666667,0.705882352941177}
\definecolor{color1}{rgb}{1,0.498039215686275,0.0549019607843137}
\definecolor{color2}{rgb}{0.172549019607843,0.627450980392157,0.172549019607843}
\definecolor{color3}{rgb}{0.83921568627451,0.152941176470588,0.156862745098039}

\begin{axis}[
width = \siwidth,
height = 1.6in,
legend cell align={left},
legend style={
  fill opacity=0.8,
  draw opacity=1,
  text opacity=1,
  at={(0.03,0.97)},
  anchor=north west,
  draw=white!80!black
},
tick align=outside,
tick pos=left,
x grid style={white!69.0196078431373!black},
xmin=0.990035819317997, xmax=1.20928673089475,
xticklabels={0.00, 0.00, 1.00, 1.05, 1.10, 1.15, 1.20},
xtick style={color=black},
y grid style={white!69.0196078431373!black},
ylabel={Inefficiency Ratio \(\displaystyle \rho(\mathbf{x}(\alpha))\)},
ymin=0.999, ymax=1.055,
yticklabels={0.00, 0.00, 1.00, 1.02, 1.04},
ytick style={color=black},
style={font=\footnotesize}
]
\addplot [dashed, line width=2pt, color=color0]
table {%
1.03216119669047 1.01265315471268
1.0044897566147 1.01225288492082
1.01260734852479 1.01191904507463
1.01379583603029 1.01142381375606
1.01881809614894 1.0110124492958
1.02149820218692 1.01039067771015
1.02276312171845 1.00987786255013
1.02840240378469 1.00902442598697
1.03393768664059 1.00830465071282
1.03607425022471 1.00799248081522
1.04413805909871 1.00708905081794
1.0811721459356 1.00684693861472
1.08340075892904 1.00650449003206
1.08486532206039 1.00617137760173
1.08668665674585 1.00613020702589
1.09042465771605 1.00550156998356
1.09443358247832 1.00519999074778
1.09856213216164 1.00503238003864
1.09868873788211 1.00468393056898
1.1010729555084 1.00446461325952
1.10310868961225 1.00425031887833
1.10488335657065 1.00407368072452
1.10702671395972 1.00387315146655
1.10909450596478 1.0037038361224
1.11088022820249 1.00355003916991
1.11312043762464 1.00337889475555
1.11832150560789 1.00298301190617
1.11957239718466 1.00288441426009
1.11576299843076 1.00305009196332
1.12450214819407 1.00258902194077
1.12612854940984 1.00248056257713
1.12744798632572 1.00236745095086
1.12893258511788 1.002300547913
1.1311955842223 1.0022465208458
1.13216155957712 1.00215474985936
1.13302907757777 1.00208193191694
1.13345357755932 1.00200762116026
1.13627781829183 1.00191806209722
1.13711055209856 1.00185019411029
1.13724681942364 1.00176410377697
1.13932546097928 1.00165964368335
1.14171019107959 1.00148624054797
1.14171292312897 1.00140838904069
1.14442385650706 1.0012506913988
1.14522803298108 1.00117394594353
1.14574102750633 1.00111580041303
1.14686538133259 1.00106149080132
1.1482093785137 1.00099741876134
1.1488659400729 1.00094440473113
1.14965327876036 1.00088451135893
1.15062395280812 1.00083530277679
1.15142091365338 1.00078839081307
1.15226540279908 1.00075312935743
1.16678580372163 1.00071182649804
1.16762903190981 1.00066199777364
1.16828372103866 1.00062028381514
1.1704556182483 1.00049481587456
1.17124394966693 1.00046400788947
1.17270926611697 1.00039228193738
1.17336047256507 1.00036159346821
1.17417311790774 1.00033008696083
1.17479082160598 1.00030876127438
1.17534906643653 1.00028218163202
1.17746072261263 1.00021628392479
1.1781671985861 1.00019634652979
1.17863973074529 1.00018160479142
1.17937033822046 1.00015854319168
1.18097984313844 1.00011095481322
1.18171879975772 1.0000966828893
1.18336391765158 1.00005878547196
1.18387715543773 1.00005023892565
1.18546981957506 1.00002745289667
1.18780607586213 1.00000430615126
1.1885084951836 1.00000205872337
1.18875075894308 1
};
\addlegendentry{I-TAP-0.01}
\addplot [dashed, line width=2pt, color=color1]
table {%
1.03216119669047 1.0126505869617
1.02149820218692 1.01038811569606
1.03393768664059 1.00830209398821
1.04413805909871 1.00708649717568
1.09042465771605 1.00549902036662
1.10310868961225 1.00424777243415
1.11312043762464 1.00337635052101
1.11576299843076 1.00304754856251
1.12893258511788 1.00229800641279
1.13627781829183 1.00191552156686
1.14171019107959 1.00148370111257
1.14574102750633 1.00111326191694
1.15062395280812 1.00083276499196
1.16828372103866 1.00061774657552
1.17534906643653 1.00027964524971
1.17863973074529 1.00017906866414
1.18171879975772 1.00009414697735
1.2019843261943 1
};
\addlegendentry{I-TAP-0.05}
\addplot [dashed, line width=2pt, color=color2]
table {%
1.19932078036853 1
1.19920746641122 1.00098670913635
1.19874859382737 1.00105051714274
1.03495007928098 1.01250767697042
};
\addlegendentry{I-Solution}
\addplot [dashed, line width=2pt, color=color3]
table {%
1.00000176984421 1.90599615279674
1.17585465341427 1.00759134147248
1.18766847879678 1.00095753440826
1.18914090901968 1.00001776131243
1.18948116322059 0.999998787384622
1.18917921401213 1.00000246852559
1.1892481556778 1
};
\addlegendentry{Jahn et al.}
\end{axis}
\node[text width=1.5cm] at (0.8,0.2) {\small Anaheim};
\end{tikzpicture}
    \end{subfigure}
    \begin{subfigure}[btp]{0.323\columnwidth}
        \centering
        \vspace{10pt}
\begin{tikzpicture}

\definecolor{color0}{rgb}{0.12156862745098,0.466666666666667,0.705882352941177}
\definecolor{color1}{rgb}{1,0.498039215686275,0.0549019607843137}
\definecolor{color2}{rgb}{0.172549019607843,0.627450980392157,0.172549019607843}
\definecolor{color3}{rgb}{0.83921568627451,0.152941176470588,0.156862745098039}

\begin{axis}[
width = \siwidth,
height = 1.6in,
legend cell align={left},
legend style={fill opacity=0.8, draw opacity=1, text opacity=1, draw=white!80!black},
tick align=outside,
tick pos=left,
x grid style={white!69.0196078431373!black},
xmin=0.97512811431037, xmax=1.52230959948223,
xtick style={color=black},
xticklabels={0.00, 0.00, 1.0, 1.2, 1.4},
y grid style={white!69.0196078431373!black},
ymin=0.999, ymax=1.04,
restrict y to domain=0.999:1.06,
yticklabels={0.00, 0.00, 1.00, 1.01, 1.02, 1.03, 1.04},
ytick style={color=black},
style={font=\footnotesize}
]
\addplot [dashed, line width=2pt, color=color0]
table {%
1.0102572793435 1.03423045286804
1.06158675604029 1.02571243757666
1.08185840707965 1.02522150240824
1.10655186163974 1.0234053265193
1.11898706998626 1.02260234421599
1.13490304002042 1.02175922017397
1.13965646004481 1.02136926088818
1.14535785953177 1.02092627646149
1.15900691522328 1.01985855439678
1.16381814767336 1.01939585750602
1.1753132627482 1.018307262107
1.18474864835685 1.01777432193977
1.19094137876347 1.01717943306762
1.19992179533637 1.0160216536372
1.21614204695219 1.01527929304595
1.22063352222043 1.0146632044624
1.21614053255694 1.01573468782437
1.23713989762536 1.01398380473779
1.24247064651539 1.01316806472034
1.26155499212072 1.01220940170584
1.26446452107007 1.01093862529464
1.26731561266407 1.01085180694545
1.27600417506479 1.00995504949882
1.28249493982 1.00967637699464
1.32414057605043 1.00563350166797
1.36425652279932 1.00397657779287
1.3658437817358 1.00378626098321
1.36777972164666 1.00366953398749
1.39943436400079 1.00278477876433
1.4377449472494 1.00085309683641
1.44026706225733 1.00033734297481
1.49129560747643 1.00002869037902
1.4974377137926 1
};
\addplot [dashed, line width=2pt, color=color1]
table {%
1.0102572793435 1.03417885641055
1.06158675604029 1.02566126607225
1.13490304002042 1.02170824589062
1.1753132627482 1.01825646003749
1.21614204695219 1.01522864203801
1.24247064651539 1.01311751903893
1.27600417506479 1.00990466411071
1.3658437817358 1.0037361833482
1.45129209814348 1.00309755512462
1.46571727582233 1.00203404206859
1.47984899570926 1.00150456916744
1.49207964647016 1.0010669492154
1.49907694295267 1
};
\addplot [dashed, line width=2pt, color=color2]
table {%
1.49710132639213 1
1.49438075370009 1.00148608467195
1.49022205637303 1.0016170988604
1.48605814178785 1.00175365434957
1.48188906062303 1.00189579082527
1.47771486577138 1.00204354840028
1.47353561231376 1.00219696761446
1.46935135749226 1.00235608943485
1.4651621606831 1.00252095525578
1.46096808336904 1.00269160689898
1.45676918911145 1.00286808661364
1.45276108913496 1.00305043707656
1.44916686191664 1.00323870139225
1.44556634556334 1.00343292309298
1.44195957712449 1.00363314613896
1.43834659555669 1.00383941491837
1.43472744170535 1.00405177424753
1.43110215828607 1.00427026937093
1.42747078986569 1.00449494596139
1.42383338284293 1.00472585012015
1.4201899854288 1.00496302837694
1.41654064762668 1.00520652769014
1.41288542121205 1.00545639544683
1.40922435971202 1.00571267946292
1.4055575183844 1.00597542798324
1.40188495419671 1.00624468968164
1.39820672580467 1.00652051366113
1.39452289353065 1.00680294945392
1.39083351934166 1.00709204702158
1.38713866682721 1.00738785675509
1.38343840117685 1.00769042947501
1.37973278915754 1.00799981643152
1.37602189909069 1.00831606930453
1.37230580082903 1.00863924020382
1.36858456573328 1.00896938166913
1.36485826664855 1.00930654667022
1.36112697788056 1.00965078860704
1.3573907751717 1.01000216130978
1.35364973567682 1.01036071903899
1.34990393793892 1.01072651648568
1.3461534618646 1.01109960877145
1.34239838869942 1.01148005144854
1.33863880100301 1.01186790049997
1.3348747826241 1.01226321233965
1.33110641867539 1.01266604381243
1.32733379550826 1.01307645219428
1.32355700068742 1.01349449519232
1.31977612296536 1.01392023094497
1.31599125225674 1.01435371802202
1.31220247961267 1.01479501542478
1.30840989719491 1.0152441825861
1.30461359824997 1.01570127937058
1.30081367708311 1.01616636607458
1.29701022903234 1.01663950342637
1.29320335044228 1.01712075258621
1.28939313863803 1.01761017514649
1.28557969189894 1.01810783313177
1.28176310943242 1.01861378899896
1.27794349134765 1.01912810563734
1.27699041585867 1.01965084636874
1.27678852911069 1.02018207494758
1.2765886309851 1.02072185556101
1.27639066871863 1.021270252829
1.27619459023647 1.02182733180444
1.27600034415198 1.02239315797325
1.27580787976627 1.02296779725449
1.27561714706792 1.02355131600043
1.27542809673261 1.02414378099669
1.27524068012273 1.02474525946232
1.27505484928712 1.02535581904991
1.27487055696063 1.02597552784569
1.27468775656383 1.02660445436964
1.27450640220264 1.02724266757557
1.27432644866796 1.02789023685127
1.27414785143535 1.02854723201855
1.27397056666466 1.02921372333339
1.27379455119966 1.02988978148603
1.27361976256773 1.03057547760106
1.27344615897945 1.03127088323753
1.2732736993283 1.03197607038906
1.27310234319025 1.03269111148395
1.27293205082345 1.03341607938524
1.01024826703043 1.03415104739086
};
\addplot [dashed, line width=2pt, color=color3]
table {%
1 11.9228264121261
1.34440332355233 1.04903719183988
1.34890817058833 1.04184483668255
1.35876340064822 1.03251666135641
1.36779255775214 1.02819445998195
1.48164207227087 1.00018223266581
1.46066713877938 1.0019872348476
1.48324585096453 1
};
\end{axis}
\node[text width=1.5cm] at (0.8,0.2) {\small Massachusetts};
\end{tikzpicture}
    \end{subfigure}
    \begin{subfigure}[btp]{0.323\columnwidth}
        \centering
        \vspace{10pt}
\begin{tikzpicture}

\definecolor{color0}{rgb}{0.12156862745098,0.466666666666667,0.705882352941177}
\definecolor{color1}{rgb}{1,0.498039215686275,0.0549019607843137}
\definecolor{color2}{rgb}{0.172549019607843,0.627450980392157,0.172549019607843}
\definecolor{color3}{rgb}{0.83921568627451,0.152941176470588,0.156862745098039}

\begin{axis}[
width = \siwidth,
height = 1.6in,
legend cell align={left},
legend style={fill opacity=0.8, draw opacity=1, text opacity=1, draw=white!80!black},
tick align=outside,
tick pos=left,
x grid style={white!69.0196078431373!black},
xmin=0.987473997719711, xmax=1.26304604788608,
xticklabels={0.00, 0.00, 1.0, 1.1, 1.2},
xtick style={color=black},
y grid style={white!69.0196078431373!black},
ymin=0.999, ymax=1.03,
yticklabels={0.00, 0.00, 1.00, 1.01, 1.02, 1.03},
ytick style={color=black},
style={font=\footnotesize}
]
\addplot [dashed, line width=2pt, color=color0]
table {%
1.02709771362984 1.02629964809451
1.01501092579337 1.02476463815947
1.02163661668061 1.02149320788889
1.03920897531093 1.01974910394532
1.04417999183975 1.01851943548177
1.04575892223503 1.01646103437012
1.06472184800697 1.0152377929048
1.0714251538072 1.01435601081908
1.07701418239315 1.01372989751561
1.08280579453186 1.01294362273387
1.08932864167602 1.01230335610018
1.09514445590111 1.01184637913861
1.09869978534192 1.01130984303842
1.10371348536129 1.01085796949852
1.10888525286956 1.01046757068132
1.11296364942916 1.00991163212425
1.11849496246563 1.00954035689957
1.12698322827414 1.00885665604232
1.12712504246939 1.00832659298915
1.1272385195253 1.00796476103077
1.13297986718668 1.00751132107223
1.1358245167617 1.00713048998565
1.1424960512246 1.00671667777246
1.14401849403861 1.00646575147914
1.14900317487466 1.0059153831851
1.15258362964901 1.00555316266743
1.15647819077166 1.00283248031482
1.15927304127252 1.00145522669928
1.20729323916325 1.00131054262375
1.21028393071082 1.00125872990354
1.21113489881462 1.00109541998315
1.21194606747654 1.00103995637995
1.21494065785012 1.00090545852646
1.21626319112678 1.00084393735486
1.2171942626905 1.0007932085936
1.21845816905592 1.00073766823008
1.22212342894047 1.00066178583271
1.22300226427117 1.00059646879999
1.2241248929463 1.00051526056854
1.22486422178514 1.00046559594146
1.22672905729286 1.00044251076002
1.22864732885857 1.00038090498326
1.22916156831712 1.00033789528396
1.23074298329181 1.00030660245646
1.23201259908945 1.00028711389123
1.23297643164075 1.00022828163916
1.23461373237738 1.00020113918006
1.23731958715417 1.00013713530775
1.23890492305139 1.00009958653318
1.24031109726405 1.00008929574864
1.24036834877736 1.00007492934102
1.2422835903351 1.00005442225699
1.24299614449695 1.00004473523393
1.24353628795276 1.00003546167301
1.24448941880417 1.00002885212432
1.24512949778629 1.00002383046016
1.24569516008279 1.00001788281691
1.247634435561 1.00001135450937
1.2475504476628 1.00001145674198
1.24981628419672 1.00000339303689
1.25031639679881 1
};
\addplot [dashed, line width=2pt, color=color1]
table {%
1.02709771362984 1.02629964809451
1.04417999183975 1.01851943548177
1.07701418239315 1.01372989751561
1.10371348536129 1.01085796949852
1.12712504246939 1.00832659298915
1.14401849403861 1.00646575147914
1.15899151182651 1.00504270710378
1.1711754707113 1.00399473494675
1.1818743595044 1.0029951745619
1.19189493083878 1.00232583624466
1.2078978715753 1.00167470958052
1.21310276488022 1.00138311005272
1.21494065785012 1.00090545852646
1.22212342894047 1.00066178583271
1.22672905729286 1.00044251076002
1.23201259908945 1.00028711389123
1.23461373237738 1.00020113918006
1.24031109726405 1.00008929574864
1.24353628795276 1.00003546167301
1.247634435561 1.00001135450937
1.25031639679881 1
};
\addplot [dashed, line width=2pt, color=color2]
table {%
1.24204379538903 1
1.24061881682591 1.00000016832547
1.23917768246415 1.000003929121
1.23772033995062 1.00001131680001
1.23624673874236 1.00002236611434
1.23475683013178 1.00003711215546
1.23325056727155 1.00005559035559
1.23172790519915 1.00007783648889
1.23018880086113 1.00010388667268
1.22863321313691 1.00013377736856
1.22706110286234 1.00016754538359
1.22547243285271 1.00020522787149
1.22386716792553 1.00024686233382
1.22224527492276 1.00029248662111
1.22060672273268 1.00034213893408
1.2189514823113 1.00039585782478
1.21727952670333 1.00045368219779
1.21559083106267 1.0005156513114
1.21388537267238 1.00058180477875
1.21216313096425 1.00065218256902
1.21042408753775 1.00072682500864
1.20866822617857 1.0008057727824
1.20689553287652 1.00088906693467
1.20510599584296 1.00097674887058
1.20329960552762 1.00106886035715
1.20147635463494 1.00116544352452
1.19963623813966 1.00126654086707
1.19777925330201 1.00137219524463
1.19590539968216 1.00148244988367
1.19401467915409 1.00159734837841
1.19210709591884 1.00171693469207
1.19018265651712 1.00184125315798
1.18824136984126 1.00197034848081
1.18628324714652 1.00210426573771
1.18430830206169 1.00224305037949
1.18231655059904 1.00238674823179
1.1803080111636 1.00253540549629
1.17828270456167 1.00268906875182
1.1762406540087 1.0028477849556
1.17418188513643 1.00301160144438
1.17210642599926 1.00318056593562
1.17001430707996 1.00335472652865
1.16833991707083 1.00353413170589
1.16700846651689 1.00371883033397
1.16567052803902 1.00390887166494
1.16432608779389 1.00410430533742
1.16297513250695 1.00430518137782
1.16161764947493 1.00451155020145
1.16025362656821 1.00472346261374
1.15888305223318 1.0049409698114
1.15750591549452 1.00516412338361
1.15612220595736 1.00539297531316
1.15473191380946 1.00562757797766
1.15333502982322 1.0058679841507
1.1519315453577 1.00611424700301
1.15052145236055 1.00636642010366
1.1491047433698 1.00662455742124
1.14768141151567 1.00688871332499
1.14625145052227 1.00715894258603
1.14481485470917 1.00743530037848
1.143371618993 1.00771784228068
1.14192173888886 1.00800662427635
1.14046521051176 1.00830170275575
1.13900203057789 1.00860313451687
1.13753219640585 1.00891097676662
1.13605570591783 1.00922528712194
1.13457255764065 1.00954612361107
1.13308275070674 1.00987354467465
1.13158628485509 1.01020760916691
1.13008316043201 1.01054837635686
1.12857337839193 1.01089590592949
1.12705694029803 1.01125025798686
1.12553384832278 1.01161149304936
1.12400410524851 1.01197967205685
1.12246771446771 1.01235485636983
1.12092467998344 1.01273710777062
1.11937500640951 1.01312648846453
1.11781869897062 1.01352306108107
1.11625576350246 1.01392688867506
1.11468620645163 1.01433803472785
1.11311003487555 1.01475656314849
1.11152725644227 1.0151825382749
1.10993787943014 1.01561602487503
1.10834191272747 1.01605708814806
1.10673936583201 1.01650579372556
1.10513024885044 1.01696220767266
1.10351457249768 1.01742639648925
1.10189234809618 1.01789842711111
1.10026358757511 1.01837836691112
1.09862830346937 1.01886628370044
1.09698650891866 1.01936224572965
1.09533821766636 1.01986632168996
1.0936834440583 1.02037858071435
1.09202220304156 1.02089909237879
1.09035451016304 1.02142792670335
1.08868038156804 1.02196515415346
1.08699983399868 1.02251084564101
1.08531288479231 1.02306507252553
1.08361955187972 1.02362790661544
1.0819198537834 1.02419942016913
1.03606140197166 1.02477968589618
};
\addplot [dashed, line width=2pt, color=color3]
table {%
1 1.24625684755394
1.11785597262373 1.07113824546842
1.15032749803855 1.02809181860033
1.15063965446512 1.01944791904243
1.20148930554737 1.00042647951916
1.25022564856515 1.0000040274766
1.25052004560579 1
1.25052004560579 1
};
\draw (axis cs:1.5,0.9) node[
  scale=0.5,
  text=black,
  rotate=0.0
]{Tiergarten};
\end{axis}
\node[text width=1.5cm] at (1.0,0.2) {\small Tiergarten};
\end{tikzpicture}
    \end{subfigure}\par\medskip
    \vspace{-20pt}
    \hspace{0pt}
    \begin{subfigure}[btp]{0.323\columnwidth}
        \centering
\begin{tikzpicture}

\definecolor{color0}{rgb}{0.12156862745098,0.466666666666667,0.705882352941177}
\definecolor{color1}{rgb}{1,0.498039215686275,0.0549019607843137}
\definecolor{color2}{rgb}{0.172549019607843,0.627450980392157,0.172549019607843}
\definecolor{color3}{rgb}{0.83921568627451,0.152941176470588,0.156862745098039}

\begin{axis}[
width = \siwidth,
height = 1.6in,
legend cell align={left},
legend style={fill opacity=0.8, draw opacity=1, text opacity=1, draw=white!80!black},
tick align=outside,
tick pos=left,
x grid style={white!69.0196078431373!black},
xlabel={Unfairness U($\f(\alpha)$)},
xmin=0.96338951825704, xmax=1.76882011660215,
xticklabels={0.00, 0.00, 1.0, 1.2, 1.4, 1.6},
xtick style={color=black},
y grid style={white!69.0196078431373!black},
ylabel={Inefficiency Ratio \(\displaystyle \rho(\mathbf{x}(\alpha))\)},
ymin=0.999, ymax=1.1,
yticklabels={0.00, 0.00, 1.00, 1.02, 1.04, 1.06, 1.08, 1.10},
ytick style={color=black},
style={font=\footnotesize}
]
\addplot [dashed, line width=2pt, color=color0]
table {%
1.04777619741354 1.08333681641639
1.01805202659178 1.07710131262555
1.03856957576225 1.07046933966396
1.05621547295758 1.06580191568301
1.06739136023125 1.06145508118298
1.08261267704885 1.05723137287348
1.09749662468401 1.05345783131864
1.111158570345 1.0499776990804
1.1257605120788 1.04637820942525
1.13765767631152 1.04348819399355
1.15073367543052 1.04044064006127
1.17409443530752 1.03561785030523
1.18540518922701 1.03351829670561
1.1957438278156 1.03170437811372
1.20654009047451 1.02985388900329
1.21649517029881 1.02821424641006
1.22525243967528 1.02688051215448
1.24408863119237 1.02353305734928
1.25378198558326 1.02280015379761
1.26177654679862 1.02171169256114
1.27061851020754 1.0206113870703
1.27878549059673 1.0195973868961
1.28222964995237 1.01889486334355
1.29422500292896 1.01778353516114
1.30161631721178 1.01695653933158
1.30936485066497 1.01618174040173
1.3157398834171 1.01541563528172
1.32244638216431 1.01470379877342
1.32903083379994 1.01400226329925
1.33548356433071 1.01339094478428
1.34159886710222 1.01289871201323
1.34776369223884 1.01222672934129
1.35256289514385 1.01178669562625
1.35940012648589 1.01119051825643
1.3647525061162 1.01068835281413
1.3723945023135 1.01032792955752
1.37572277012039 1.0096877769556
1.38064604179561 1.00920442853723
1.38592842834554 1.00875717194621
1.38863075654231 1.00831839941908
1.39607844759622 1.00785560538963
1.40056897402795 1.00745992640823
1.40464279861951 1.00696927852201
1.40897113167956 1.00671355064748
1.4081215725295 1.0068390950926
1.42119791005106 1.00610423029478
1.42231346414045 1.00578764161838
1.42639394291261 1.0055067284872
1.43052192470923 1.00524900506612
1.43428496571041 1.00498803927669
1.43813751057205 1.00477510445006
1.44183173030671 1.00457549064643
1.44400058400344 1.00434688224529
1.44901423088906 1.00409613521423
1.45252238397262 1.00390571396136
1.45576110565408 1.0037241267993
1.45718289429052 1.00357014227175
1.46121782868752 1.00297916712246
1.47209345623178 1.00284059630784
1.47490128944075 1.00263125951364
1.47811924918983 1.0024477975019
1.48101690976591 1.00228771143831
1.48432926192216 1.00207101284246
1.48664075358742 1.00196770131731
1.48957064740802 1.00179804116805
1.49198569127231 1.00165366483195
1.49483594230124 1.00153126284474
1.49346541729261 1.00159613416909
1.49789277508697 1.00141491115098
1.50365105099353 1.00109218719726
1.50655370564991 1.00100424987227
1.5094254104892 1.00086276299035
1.51028534459979 1.00078538061107
1.51739787535417 1.00061393612214
1.51896747794997 1.00056171675806
1.52663405852932 1.00036200856258
1.51957659664831 1.0003695032153
1.53230159318027 1.00020363255038
1.53513701275818 1.00019133888759
1.53694358197204 1.00017015667451
1.53878073347556 1.00014885590579
1.54052131136049 1.00013050425699
1.54262431739639 1.00009965487693
1.54432292523624 1.0000844240366
1.54634023237672 1.00005568724939
1.55397178591997 1.00001047563676
1.55799302060645 1
};
\addplot [dashed, line width=2pt, color=color1]
table {%
1.04777619741354 1.08333681641639
1.08261267704885 1.05723137287348
1.15073367543052 1.04044064006127
1.20654009047451 1.02985388900329
1.25378198558326 1.02280015379761
1.29422500292896 1.01778353516114
1.32903083379994 1.01400226329925
1.35940012648589 1.01119051825643
1.38592842834554 1.00875717194621
1.40897113167956 1.00671355064748
1.43052192470923 1.00524900506612
1.44901423088906 1.00409613521423
1.46573510662593 1.00323698463552
1.48101690976591 1.00228771143831
1.49483594230124 1.00153126284474
1.50655370564991 1.00100424987227
1.51896747794997 1.00056171675806
1.51957659664831 1.0003695032153
1.53878073347556 1.00014885590579
1.55799302060645 1
};
\addplot [dashed, line width=2pt, color=color2]
table {%
1.58643463565443 1
1.58113972845346 1.00002686423714
1.57583153022361 1.00006552422703
1.57051015061457 1.00011607951805
1.56517570550822 1.0001786308011
1.55982831695567 1.00025327991242
1.55446811311177 1.00034012983638
1.54909522816709 1.00043928470824
1.54370980227739 1.00055084981701
1.53831198149071 1.00067493160816
1.53290191767192 1.00081163768651
1.52747976842509 1.00096107681892
1.52204569701336 1.00112335893718
1.51659987227668 1.00129859514074
1.51114246854725 1.00148689769953
1.50567366556279 1.00168838005675
1.50019364837779 1.00190315683166
1.49470260727256 1.00213134382238
1.48920073766038 1.00237305800867
1.48368823999266 1.00262841755477
1.47816531966222 1.00289754181212
1.47263218690476 1.00318055132222
1.46708905669857 1.00347756781937
1.46153614866251 1.00378871423352
1.45597368695239 1.00411411469302
1.45040190015579 1.00445389452743
1.44482102118536 1.00480818027033
1.43923128717072 1.00517709966208
1.43363293934894 1.00556078165263
1.42802622295388 1.00595935640432
1.4224113871041 1.00637295529468
1.41678868468984 1.0068017109192
1.41115837225878 1.00724575709414
1.40552070990092 1.00770522885934
1.39987596113242 1.00818026248097
1.3942243927787 1.00867099545437
1.3885662748567 1.0091775665068
1.38290188045653 1.00970011560028
1.37723148562238 1.01023878393436
1.37155536923296 1.0107937139489
1.36587381288152 1.0113650493269
1.36018710075534 1.01195293499726
1.35580179843677 1.0125575171376
1.35192063305436 1.01317894317704
1.34802379742834 1.01381736179897
1.34411146239356 1.01447292294391
1.34018380199988 1.01514577781223
1.33624099346082 1.01583607886702
1.33228321710101 1.01654397983679
1.32831065630284 1.01726963571836
1.32432349745189 1.01801320277958
1.32032192988157 1.01877483856218
1.31630614581672 1.01955470188452
1.31227634031636 1.02035295284441
1.30823271121548 1.0211697528219
1.30417545906613 1.02200526448205
1.3001047870775 1.02285965177778
1.29602090105538 1.02373307995259
1.29192400934077 1.02462571554343
1.28781432274782 1.02553772638342
1.283692054501 1.02646928160472
1.27955742017173 1.02742055164124
1.27541063761429 1.02839170823152
1.27125192690118 1.02938292442145
1.26708151025791 1.03039437456713
1.26289961199727 1.0314262343376
1.25870645845314 1.03247868071768
1.25450227791377 1.03355189201074
1.25028730055473 1.03464604784153
1.24606175837143 1.0357613291589
1.24182588511127 1.03689791823869
1.23757991620553 1.03805599868644
1.23332408870094 1.03923575544023
1.22905864119098 1.04043737477347
1.22542748338825 1.04166104429769
1.22430242059751 1.0429069529653
1.22317016163993 1.04417529107246
1.22203070102403 1.0454662502618
1.22088403365948 1.04678002352524
1.21973015485928 1.04811680520682
1.21856906034203 1.04947679100542
1.21740074623401 1.05086017797762
1.21622520907137 1.05226716454046
1.21504244580213 1.05369795047427
1.21385245378823 1.05515273692538
1.21265523080752 1.05663172640903
1.21145077505566 1.05813512281208
1.21023908514804 1.05966313139582
1.20902016012157 1.0612159587988
1.2077939994365 1.06279381303957
1.20656060297815 1.06439690351952
1.20531997105859 1.06602544102564
1.2040721044183 1.06767963773336
1.20309191088175 1.06935970720927
1.20305869424139 1.071065864414
1.203025389648 1.07279832570495
1.05326367092673 1.08186286408947
};
\addplot [dashed, line width=2pt, color=color3]
table {%
1 1.51355675945567
1.02176255267465 1.18592437282101
1.23181669748365 1.16590656115161
1.23786559835767 1.07971011751014
1.55772539482382 1.00231949993955
1.55780963667575 1.00025161441859
1.73220963485919 1
};
\end{axis}
\node[text width=1.5cm] at (2.4,2.0) {\small Friedrichshain};
\end{tikzpicture}
    \end{subfigure}\hspace{8pt}
    \begin{subfigure}[btp]{0.323\columnwidth}
        \centering
        \vspace{10pt}
\begin{tikzpicture}

\definecolor{color0}{rgb}{0.12156862745098,0.466666666666667,0.705882352941177}
\definecolor{color1}{rgb}{1,0.498039215686275,0.0549019607843137}
\definecolor{color2}{rgb}{0.172549019607843,0.627450980392157,0.172549019607843}
\definecolor{color3}{rgb}{0.83921568627451,0.152941176470588,0.156862745098039}

\begin{axis}[
width = \siwidth,
height = 1.6in,
legend cell align={left},
legend style={fill opacity=0.8, draw opacity=1, text opacity=1, draw=white!80!black},
tick align=outside,
tick pos=left,
x grid style={white!69.0196078431373!black},
xlabel={Unfairness U($\f(\alpha)$)},
xmin=0.96967947835315, xmax=1.63674378150016,
xticklabels={0.00, 0.00, 1.0, 1.2, 1.4, 1.6},
xtick style={color=black},
y grid style={white!69.0196078431373!black},
ymin=0.999, ymax=1.04,
yticklabels={0.00, 0.00, 1.00, 1.01, 1.02, 1.03, 1.04},
ytick style={color=black},
style={font=\footnotesize}
]
\addplot [dashed, line width=2pt, color=color0]
table {%
1.0415798961255 1.03842126315122
1.07576637026551 1.032385835623
1.08007789074929 1.02930208526587
1.09631464870632 1.02683345004132
1.10865263706428 1.02327474195923
1.11795290580365 1.01978487884171
1.13125068929408 1.01849543729173
1.16881357208051 1.01504965589419
1.17239742617107 1.01427346251606
1.19896661950175 1.01214513408235
1.22466422763239 1.0101790973289
1.23276469068346 1.00886189012525
1.32949096403207 1.00778616681752
1.35524879895757 1.00590006741294
1.36737208496817 1.00556887791138
1.37303459904537 1.00517062262808
1.37945605002522 1.00493660427735
1.39242057305898 1.0040357445566
1.40894807070717 1.00376197904945
1.41947797920481 1.00313676400093
1.42898305752734 1.00283657839387
1.43503204874575 1.00264481939525
1.43745778233635 1.00255411075781
1.4482867333244 1.0022000421102
1.4583722094363 1.00201006210967
1.46383799994945 1.0018166843505
1.46582095770959 1.00181070943076
1.47225033209451 1.00165644994654
1.47720310046831 1.00139529387857
1.48741095489382 1.00133933067813
1.49315318439882 1.00113356540961
1.50311876096558 1.00102698289222
1.51065704251881 1.00095245104715
1.51914889718233 1.00065026097793
1.5131317066045 1.00070788620207
1.52638666160683 1.00059863455539
1.53441255139926 1.00056497508452
1.54386108444038 1.0004404439087
1.54868286665003 1.00019581288586
1.55170182456946 1.00014197231009
1.56227253811229 1.00010832409044
1.57942622925095 1.00009872043471
1.59088103480664 1
};
\addplot [dashed, line width=2pt, color=color1]
table {%
1.0415798961255 1.03838574471717
1.19989396774252 1.02482670464615
1.25267080980598 1.01296036416161
1.30686067519953 1.00948171247414
1.37945605002522 1.00490223116138
1.43745778233635 1.00251981913327
1.46383799994945 1.0017824179491
1.49299229784087 1.00131505912102
1.50589476268636 1.00101152767801
1.5131317066045 1.00067365772629
1.54712642254202 1.00041217018664
1.55170182456946 1.00010776319098
1.57942622925095 1.000064512795
1.593892308008 1.00006036910302
1.60274242559089 1
};
\addplot [dashed, line width=2pt, color=color2]
table {%
1.60642267681166 1
1.60098365346969 1.00013812851662
1.59556253582251 1.00028084522237
1.59015928603341 1.00042815858827
1.58477386614935 1.00058007715416
1.57940623810231 1.00073660952877
1.57405636371052 1.00089776438983
1.56872420467985 1.0010635504841
1.56340972260509 1.00123397662746
1.55811287897131 1.00140905170499
1.5528336351552 1.00158878467106
1.54757195242637 1.00177318454934
1.54232779194878 1.00196226043296
1.53710111478205 1.00215602148451
1.53189188188284 1.00235447693617
1.52670005410625 1.00255763608974
1.52152559220716 1.00276550831673
1.51636845684164 1.00297810305845
1.51122860856836 1.00319542982606
1.50610600784995 1.00341749820065
1.50100061505444 1.00364431783333
1.49591239045663 1.00387589844528
1.49084129423956 1.00411224982783
1.48578728649586 1.00435338184255
1.48075032722923 1.0045993044213
1.47573037635583 1.00485002756633
1.47072739370572 1.00510556135032
1.4657413390243 1.00536591591649
1.46077217197374 1.00563110147864
1.45581985213444 1.00590112832126
1.45088433900643 1.00617600679957
1.44596559201084 1.00645574733961
1.44106357049138 1.00674036043831
1.43617823371572 1.00702985666356
1.43130954087701 1.00732424665432
1.42645745109528 1.00762354112061
1.42162192341894 1.00792775084367
1.41680291682621 1.00823688667601
1.41200039022659 1.00855095954145
1.4072143024623 1.00886998043522
1.40244461230976 1.00919396042403
1.39769127848105 1.00952291064617
1.39295425962537 1.00985684231153
1.38823351433048 1.01019576670171
1.38352900112417 1.01053969517008
1.37884067847575 1.01088863914189
1.37416850479747 1.01124261011426
1.369512438446 1.01160161965636
1.36487243772388 1.0119656794094
1.36024846088098 1.01233480108674
1.35564046611597 1.01270899647397
1.35104841157775 1.01308827742895
1.34647225536692 1.01347265588193
1.34416411157723 1.01386214383557
1.34313036330589 1.01425675336508
1.34209909877818 1.01465649661823
1.34107031587739 1.01506138581546
1.04244143020553 1.03816958768089
};
\addplot [dashed, line width=2pt, color=color3]
table {%
1.00000058304165 8.52174813028758
1.00965980530861 8.46919433999255
1.01624640160335 5.34444789301628
1.02995503944954 3.0176163816457
1.1027181808789 1.87069600386083
1.14103855448072 1.14125963093868
1.18312237580067 1.01585642119357
1.29681347394932 1.00746786362138
1.30006096789984 1.00121828594962
1.45769176085875 1
};
\end{axis}
\node[text width=2cm] at (2.6,2.0) {\small Sioux Falls};
\end{tikzpicture}
    \end{subfigure}\hspace{-5pt}
    \begin{subfigure}[btp]{0.323\columnwidth}
        \centering
        \vspace{10pt}
\begin{tikzpicture}

\definecolor{color0}{rgb}{0.12156862745098,0.466666666666667,0.705882352941177}
\definecolor{color1}{rgb}{1,0.498039215686275,0.0549019607843137}
\definecolor{color2}{rgb}{0.172549019607843,0.627450980392157,0.172549019607843}
\definecolor{color3}{rgb}{0.83921568627451,0.152941176470588,0.156862745098039}

\begin{axis}[
width = \siwidth,
height = 1.6in,
legend cell align={left},
legend style={fill opacity=0.8, draw opacity=1, text opacity=1, draw=white!80!black},
tick align=outside,
tick pos=left,
x grid style={white!69.0196078431373!black},
xlabel={Unfairness U($\f(\alpha)$)},
xmin=0.957727873437494, xmax=1.88771465781263,
xticklabels={0.0, 0.0, 1.0, 1.2, 1.4, 1.6, 1.8},
xtick style={color=black},
y grid style={white!69.0196078431373!black},
ymin=0.999, ymax=1.04,
yticklabels={0.00, 0.00, 1.00, 1.01, 1.02, 1.03, 1.04},
ytick style={color=black},
style={font=\footnotesize}
]
\addplot [dashed, line width=2pt, color=color0]
table {%
1.03672245291724 1.03652424176196
1.0226344847168 1.03324752465105
1.02550179776974 1.03262933652619
1.04287080855659 1.03071210147687
1.05076227038471 1.02975565441929
1.05388727071992 1.02815996213305
1.06119700607259 1.02715623344039
1.18628935642275 1.02616883558966
1.1967912039846 1.02519872420061
1.21395587870278 1.02403422424142
1.23796790885949 1.02254556066436
1.24537947260722 1.02198488741122
1.26033373518824 1.02106832114016
1.2767270448256 1.02012271531014
1.29039888398771 1.01926684665749
1.30500702749838 1.01858634455729
1.31737457556133 1.01751915577269
1.33057673564651 1.01651136355009
1.34368011810848 1.01551871509716
1.35663634143437 1.01460529098159
1.3681592320483 1.01373857236214
1.38070454904429 1.01291032884821
1.39230905047826 1.01213488144674
1.40438803836272 1.01156160047472
1.41533613168528 1.01110466467822
1.425975934206 1.01044354105206
1.43663419499817 1.01006084821827
1.44682964565156 1.00971249855662
1.45762548231361 1.00939516103938
1.46741657602615 1.00906467673361
1.47694577100812 1.008715617981
1.49076795080293 1.0083475629617
1.4956413097021 1.00815330957918
1.50461904869886 1.00786678565653
1.51386749588823 1.00744967756313
1.53805306958992 1.00666200204941
1.53111300536628 1.00674783010233
1.53922884253374 1.00624802896109
1.54757725236952 1.00606394697898
1.55435289610883 1.00573750127925
1.56605753078536 1.00540825624147
1.57157095627726 1.00517277741572
1.57916776911221 1.00490673359551
1.58572771075022 1.0046056303769
1.59375831823923 1.00438481945924
1.6065131196664 1.00400763174575
1.6081559088087 1.00390491354557
1.61516803979347 1.00368551477299
1.62210347703546 1.00347661546136
1.62870805659832 1.00329154763578
1.63543666623415 1.00309926952468
1.64245871468799 1.00288760884354
1.64757384713466 1.00276240287848
1.6585006939462 1.00252751203415
1.66111949639228 1.00242746407034
1.67076106669155 1.00224130578431
1.67191069770475 1.00214305176096
1.67745164444213 1.00200457846405
1.69040935023902 1.00171105516087
1.68872682281297 1.00176827595284
1.69522925412295 1.0016807315389
1.69990499073641 1.0015489011978
1.70503389001784 1.00144411066528
1.71044075716925 1.00135360369856
1.71547390222687 1.0012577712703
1.72127914579171 1.00114711776424
1.72532145913196 1.00109146383828
1.73053553677932 1.00101237816411
1.73459477204077 1.00095002859474
1.73931791357628 1.00082213108646
1.74943615581081 1.00074290982962
1.75189420260262 1.00073619637768
1.75815161279011 1.00063403376162
1.76270878449534 1.00058844853654
1.76711023325426 1.00054922396839
1.76210366503677 1.00060075444421
1.77589082131155 1.00044357054071
1.77959337551325 1.00040420027332
1.7840882476894 1.00037176451859
1.78647076475064 1.00034056705524
1.79396614573188 1.00029462989122
1.79564229911913 1.00025012698988
1.80052617131768 1.00022481105851
1.80340139370984 1.00019236244154
1.80511885674786 1.00016591744845
1.81108400933928 1.00013170352962
1.8137377847152 1.00012123266896
1.81775971141073 1.0000992220329
1.82209282795821 1.00007007770955
1.82835601887225 1.00004915349003
1.83525901128958 1.00002021991049
1.8391186241206 1.00001311574237
1.84214251019417 1.00001264492894
1.84544253125013 1
};
\addplot [dashed, line width=2pt, color=color1]
table {%
1.03672245291724 1.03651678860029
1.05076227038471 1.02974824992738
1.1967912039846 1.02519135247545
1.2767270448256 1.02011538008418
1.34368011810848 1.01551141297642
1.40438803836272 1.01155432680775
1.45762548231361 1.00938790295026
1.50461904869886 1.00785953855725
1.54757725236952 1.00605671284306
1.58572771075022 1.00459840672705
1.62210347703546 1.00346939992974
1.6585006939462 1.00252030332708
1.69040935023902 1.00170385232456
1.71044075716925 1.00134640343252
1.73459477204077 1.00094283123062
1.75815161279011 1.00062683866968
1.77959337551325 1.00039700683399
1.80052617131768 1.00021761890909
1.81775971141073 1.00009203078653
1.83525901128958 1.00001302923219
1.85219020369235 1
};
\addplot [dashed, line width=2pt, color=color2]
table {%
1.83068070592078 1
1.82448571300119 1.00002108696109
1.81826080176109 1.00004706613974
1.8120064497772 1.00007798581003
1.80572314063241 1.00011389463114
1.79941136374976 1.00015484164885
1.79307161422382 1.0002008762969
1.78670439264974 1.00025204839851
1.7803102049499 1.00030840816774
1.77388956219848 1.00037000621102
1.76744298044389 1.0004368935285
1.76097098052939 1.00050912151556
1.75447408791177 1.00058674196421
1.74795283247847 1.00066980706455
1.74140774836313 1.00075836940619
1.73483937375968 1.00085248197972
1.72824825073519 1.00095219817811
1.72163492504158 1.0010575717982
1.71499994592627 1.00116865704208
1.70834386594194 1.0012855085186
1.70166724075559 1.00140818124474
1.69497062895688 1.00153673064711
1.68825459186604 1.00167121256334
1.68151969334142 1.00181168324355
1.67476649958674 1.0019581993518
1.66799557895827 1.00211081796749
1.66120750177208 1.00226959658683
1.65440284011134 1.00243459312429
1.64758216763398 1.00260586591399
1.64074605938069 1.00278347371121
1.6338950915835 1.00296747569376
1.62702984147499 1.00315793146349
1.62015088709821 1.00335490104766
1.61325880711761 1.00355844490043
1.60635418063084 1.00376862390428
1.59943758698178 1.00398549937147
1.59250960557475 1.00420913304543
1.58557081569005 1.00443958710226
1.57862179630101 1.00467692415214
1.57166312589258 1.00492120724078
1.56469538228152 1.00517249985082
1.55771914243848 1.00543086590336
1.55073498231186 1.0056963697593
1.54374347665362 1.00596907622084
1.53674519884721 1.0062490505329
1.52974072073763 1.00653635838458
1.52273061246365 1.00683106591056
1.51571544229247 1.00713323969259
1.50869577645668 1.00744294676088
1.50167217899377 1.00776025459558
1.49464521158812 1.00808523112821
1.4876154334157 1.00841794474308
1.48058340099142 1.00875846427876
1.47354966801923 1.00910685902949
1.46651478524509 1.00946319874664
1.45947930031281 1.00982755364016
1.45244375762279 1.01019999437997
1.44540869819382 1.01058059209749
1.43837465952788 1.01096941838697
1.43134217547805 1.01136654530702
1.42431177611954 1.011772045382
1.41728398762393 1.01218599160348
1.41025933213654 1.01260845743169
1.40323832765715 1.01303951679691
1.39622148792394 1.01347924410099
1.38920932230068 1.01392771421871
1.38220233566736 1.01438500249927
1.37520102831407 1.01485118476773
1.36820589583829 1.0153263373264
1.36121742904553 1.01581053695636
1.35423611385341 1.01630386091883
1.34971391211182 1.01680638695663
1.34618242112504 1.01731819329565
1.34264809703743 1.01783935864625
1.33911103649535 1.01836996220471
1.33557133619728 1.0189100836547
1.33202909288198 1.01945980316869
1.32848440331682 1.02001920140937
1.32493736428604 1.02058835953116
1.3213880725791 1.02116735918157
1.31783662497915 1.02175628250269
1.31428311825145 1.02235521213262
1.3107276491319 1.02296423120692
1.30717031431567 1.02358342336
1.30361121044577 1.02421287272664
1.30005043410185 1.02485266394335
1.29648808178893 1.02550288214988
1.29292424992628 1.0261636129906
1.2893590348363 1.02683494261598
1.28579253273358 1.02751695768403
1.28222483971393 1.0282097453617
1.27865605174355 1.02891339332636
1.27508626464825 1.02962798976725
1.2715155741028 1.03035362338686
1.26794407562027 1.03109038340243
1.2643718645416 1.03183835954737
1.26079903602508 1.03259764207268
1.2572256850361 1.03336832174843
1.25365190633685 1.03415048986516
1.25007779447622 1.03494423823536
1.0343028828783 1.03574965919486
};
\addplot [dashed, line width=2pt, color=color3]
table {%
1 5.62867155729629
1.2280164862816 2.82784066350701
1.22834009750098 2.39480344356529
1.23084411732849 1.42967488182058
1.23488796641289 1.07833119177006
1.23525832830698 1.05908882652385
1.23618046730401 1.01461690354873
1.23636106761171 1.00834160341115
1.52883942773326 1.00006198713296
1.53285003376805 1.000004452055
1.57447257397045 0.99999859355778
1.55554948513109 1
};
\end{axis}
\node[text width=1.5cm] at (2.3,2.0) {\small Prenzlauerberg};
\end{tikzpicture}
    \end{subfigure}
    \vspace{-15pt}
    \caption{{\small \sf Pareto frontier depicting the trade-off between efficiency and fairness for the 
    (i) I-TAP method with a step size $s = 0.01$, (ii) I-TAP method with $s = 0.05$,
    (iii) I-Solution method with $s = 0.01$, and (iv) Jahn et al.'s method~\cite{so-routing-seminal} with $s=0.05$.}}  
    \label{fig:vals_updated}
\end{figure*}

Figure~\ref{fig:vals_updated} depicts the Pareto frontiers, i.e., the set of all Pareto efficient combinations of system efficiency and user fairness, 
for the six transportation networks in Table~\ref{tab:problem-instances}. In particular, observe that the Pareto frontiers of the I-TAP method 
are below that of the other two approaches for most values of unfairness. This observation indicates that the I-TAP method outperforms the other two approaches since the inefficiency ratio of the I-TAP solution is the lowest for most desired levels of unfairness. 
Only for the Sioux Falls and Prenzlauerberg data-sets, the algorithm in \cite{so-routing-seminal} achieved lower inefficiency ratios than both the I-TAP and I-Solution methods for higher values of unfairness, which, in practice, would be undesirable. Furthermore, note that, unlike the two convex-combination approaches, the solution of the algorithm in \cite{so-routing-seminal} can result in inefficiency ratios that are much greater than the PoA for low levels of unfairness. 
The I-TAP method outperforms the I-Solution method since the set of paths that users can traverse is not restricted to the union of the routes under the UE and SO solutions as is the case for the I-Solution method. In particular, there may be traffic assignments with lower total travel times that use paths not encapsulated by the restricted set of paths corresponding to the I-Solution method. Furthermore, while the PoA for each of the data-sets is quite low, some real-world transportation networks may have much higher PoA values (even as high as two) \cite{boston-large-poa}, which would make the trade-off between efficiency and fairness even more prominent.


\ifarxiv 
Finally, the Pareto frontiers of the I-TAP method for the $0.01$ and $0.05$ increments of the convex combination parameter almost overlap each other for all the transportation networks other than Sioux Falls. Since Sioux Falls has a highly discontinuous unfairness function (cf. Figure~\ref{fig:tt_v_alpha_and_beta_v_alpha_finer_discretization}), it is likely that the $0.05$ increments of $\alpha$ values may not capture all the low total travel time solutions that keep within a $\beta$ bound of unfairness that the $0.01$ increments of $\alpha$ may be able to capture. For all the other datasets, the near equivalence of the Pareto frontiers for the two $\alpha$ discretizations suggests that a good performance of the I-TAP method can be achieved with coarse discretizations of the convex combination parameter set. Thus, we only need to compute a solution to the convex program I-TAP$_{\alpha}$ for relatively few values of $\alpha$ to characterize the Pareto frontier, implying the computational efficiency of the I-TAP method. 
\else
\fi

\fakeparagraph{Runtime Comparison.} We report in Table~\ref{tab:problem-instances} the runtime of the Jahn et al.\ method~\cite{so-routing-seminal} and our I-TAP method. For each instance we report the average runtime over the parameters $\phi$ and $\alpha$ for the competitor and our method, respectively. We observe that our approach is faster by at least three orders of magnitude. This is unsurprising since our method solves \emph{unconstrained} shortest-path queries, which can be implemented in $O(|E|+|V|\log |V|)$ time, within each Frank-Wolfe iteration, whereas~\cite{so-routing-seminal} solves \emph{constrained} shortest-path queries which are \ifarxiv known to be NP-hard. \fi We do mention that a more efficient implementation of constrained shortest-path query can be achieved by directly implementing a label-correcting algorithm rather than using the \texttt{r\_c\_shortest\_paths} routine from Boost, which is overly general for our setting and hence less efficient. Nevertheless, even with this improvement it would still be much slower than the unconstrained near-linear algorithm. Notice that both approaches can be sped up via parallel computation of shortest-path queries, and our method can be made even faster through modern heuristics for shortest-path queries
\ifarxiv 
in 
transportation networks, such as contraction hierarchies~\cite{GeisbergerETAL12}, as in~\cite{BuchETAL18}.
\else
\cite{BuchETAL18}.
\fi  

\section{Conclusion and Future Work} \label{sec:future-work}

In this paper, we developed (i) a computationally efficient method for traffic routing that trades-off system efficiency and user fairness, and (ii) pricing schemes to enforce fair traffic assignments as a UE. 
\ifarxiv  We introduced the I-TAP method, which involves solving interpolated traffic assignment problems by taking a convex combination between the UE-TAP and SO-TAP objectives, to find an efficient feasible traffic assignment to the \fso problem. We then established various solution properties of I-TAP, including computational tractability and solution continuity, and developed theoretical bounds on the inefficiency ratio and unfairness level in terms of the convex combination parameter of I-TAP. To enforce the traffic assignments outputted by the I-TAP method as a UE, we developed a marginal-cost pricing scheme when users are homogeneous, and a linear programming based pricing scheme when users are heterogeneous. Finally, we presented numerical experiments to evaluate the performance of our approach on real-world transportation networks. The results indicated that our algorithm is not only very computationally efficient but also generally results in traffic assignments with lower total travel times for most levels of unfairness as compared to the algorithm in
\cite{so-routing-seminal}.
\fi

There are various directions for future research. First, it would be worthwhile to develop theoretical bounds for I-TAP to demonstrate its applicability to other notions of unfairness, some of which are studied in \ifarxiv Appendix~\ref{apdx:general-unfairness}\else Appendix F\fi, and extend its optimality beyond two-edge Pigou networks. Next, it would be interesting to investigate fairness notions that compare user travel times across different O-D pairs. Finally, given the significant computational advantages of the I-TAP method, it would be interesting to study the generalizability of our approach when accounting for costs beyond the travel times of users, such as environmental pollution and user discomfort.

\ifarxiv 
\section{Acknowledgements}
We thank Kaidi Yang and anonymous reviewers for their detailed inputs and insightful feedback on an earlier version of this work. This work was supported in part by the National Science Foundation (NSF) CAREER Award CMMI1454737, NSF Award 1830554, Toyota Research Institute, the Israeli Ministry of Science and Technology grants no.\ 3-16079 and 3-17385, and the United States-Israel Binational Science Foundation grants no.\ 2019703 and 2021643.
\fi

\bibliographystyle{ACM-Reference-Format} 
\bibliography{main}


\appendix

\section*{Appendix}

\section{Proof of Lemma~\ref{lem:itap-optimal-pigou}}
\label{apdx:pf-pigou-optimal}

To prove the claim, we first note that in a two edge Pigou network the sum of the traffic flows on the two edges must add up to the traffic demand $D$. Thus, we can reduce the problem of determining a two-dimensional vector of edge flows to a single dimensional problem of determining the flow $x_1$ on the first edge, with the flow on the other edge given by $D-x_1$ for a traffic demand $D$. With slight abuse of notation, we denote $\objso(x_1)$, $\objue(x_1)$ and $\objitap(x_1)$ to denote the system optimum, user equilibrium and I-TAP objectives, respectively.

We now complete the proof in two steps. First we show that the optimal solutions of I-TAP for any $\alpha \in [0, 1]$ and the \fso problem lie between the user equilibrium and system optimum solutions, i.e., $x_1(\alpha), x_1^{\beta} \in [x_1(0), x_1(1)]$ for all $\alpha \in [0, 1]$ and $\beta \in [1, \infty)$, where we assume without loss of generality that $x_1(0) \leq x_1(1)$. Note here that we only compare the edge flows of I-TAP and \fso since both the path and edge flows coincide for a two edge Pigou network. We then extend the result of Corollary~\ref{cor:lipschitz-cont-x} to show that the optimal solution $x_1(\alpha)$ is continuous in the closed interval $[0, 1]$. Note that both the above claims jointly imply by the intermediate value theorem that for any $\beta \in [1, \infty)$ there exists some $\alpha^* \in [0, 1]$ such that $x_1(\alpha^*) = x_1^{\beta}$. We now proceed to prove the two claims.

We begin by establishing that $x_1(\alpha) \in [x_1(0), x_1(1)]$ for any $\alpha \in [0, 1]$. In particular, we will show that for any point not in the interval $[x_1(0), x_1(1)]$ that we can find another feasible point with a lower I-TAP objective value. To see this, we proceed by contradiction. Fix some $\alpha \in [0, 1]$ and suppose that $x_1(\alpha) < x_1(0)$. Then, we observe by the optimality of $x_1(0)$ for the user equilibrium traffic assignment problem that $\objue(x_1(\alpha))>\objue(x_1(0))$ and by the strict convexity of $\objso$ that $\objso(x_1(\alpha))> \objso(x_1(0))$. Together, both the strict inequalities imply that $\objitap(x_1(\alpha))>\objitap(x_1(0))$, a contradiction. Through an almost identical argument, we can show that $x_1(\alpha)>x_1(1)$ is also not possible. Thus, we have shown that $x_1(\alpha) \in [x_1(0), x_1(1)]$ for any $\alpha \in [0, 1]$.

Next, we show that $x_1^{\beta} \in [x_1(0), x_1(1)]$ for any $\beta \in [1, \infty)$. In particular, we will show that for any point not in the interval $[x_1(0), x_1(1)]$ that we can find another feasible point with a lower system optimum objective value. To see this, we again proceed by contradiction. First suppose that $x_1^{\beta} < x_1(0)$. In this case, note that $\objso(x_1^{\beta}) > \objso(x_1(0))$ by the strict convexity of $\objso$. Since $x_1(0)$ is a feasible solution to the \fso problem for any $\beta \in [1, \infty)$ this implies that $x_1^{\beta} < x_1(0)$ cannot be an optimal solution to the \fso problem. This implies that $x_1^{\beta} \geq x_1(0)$. Next, suppose by contradiction that $x_1^{\beta}>x_1(1)$. If $x_1^{\beta} = D$, then $\objso(x_1^{\beta}) = D t_1(D) > D t_1(x_1(0)) = x_1(0) t_1(x_1(0)) + (D-x_1(0)) t_2(x_2(0)) = \objso(x_1(0))$ since $x_1(0) < x_1^{\beta} = D$. Thus, $x_1(0)$ achieves a lower total travel time and is feasible for any $\beta \in [1, \infty)$, and so $x_1^{\beta} = D$ cannot be a solution to the \fso problem. Next, suppose that $x_1(1)<x_1^{\beta} < D$. In this case, we have that $\frac{t_1(x_1^{\beta})}{t_2(D - x_1^{\beta})} > \max \{ \frac{t_1(x_1(1))}{t_2(D - x_1(1))}, 1 \}$, and thus the ratio of travel times is strictly greater than than that under the system optimal traffic assignment. Thus, we have that for any $\beta \in [1, \infty)$ it must be the case that $x_1^{\beta} \leq x_1(1)$, thereby proving our claim that $x_1^{\beta} \in [x_1(0), x_1(1)]$.

We now prove the second claim that $x_1(\alpha)$ is continuous in the closed interval $[0, 1]$. First observe by Corollary~\ref{cor:lipschitz-cont-x} that $x_1(\alpha)$ is continuous in the open interval $(0, 1)$. Thus, we just need to prove the continuity of $x_1(\alpha)$ for $\alpha = 0, 1$. To do so, we introduce a problem instance specific constant $C(G, D, \{t_e\}_{e = 1}^2)$ that denotes the system optimum objective value when $x_1 = x_2 = D$. We note that $C(G, D, \{t_e\}_{e = 1}^2)$ is an upper bound on the system optimal objective for any feasible traffic assignment and it is a constant that depends on the graph $G$, the demand $D$ and the travel time functions. We prove the continuity for $\alpha = 0$ in what follows. First, observe that
\begin{align*}
    \objitap(x_1(0)) - \objitap(x_1(\alpha)) 
    &= \alpha \left( \objso(x_1(0)) - \objso(x_1(\alpha)) \right) \\&+ (1-\alpha) \left( \objue(x_1(0)) - \objue(x_1(\alpha)) \right), \\
    &< \alpha \left( \objso(x_1(0)) - \objso(x_1(\alpha)) \right), \\
    &\leq \alpha C(G, D, \{t_e\}_{e = 1}^2),
\end{align*}
where the strict inequality follows since $\objue(x_1(0)) < \objue(x_1(\alpha))$, and the last inequality follows by the boundedness of $\objso$ due to the finite demand $D$. Finally, by the strict convexity of the objective $\objitap(\cdot)$ it must follow that $x_1(\alpha) \rightarrow x_1(0)$ as $\alpha \rightarrow 0$ since $\objitap(x_1(0)) \rightarrow \objitap(x_1(\alpha))$ as $\alpha \rightarrow 0$ by the above analysis. Thus, we have continuity of $x_1(\alpha)$ for $\alpha = 0$. Through an analogous analysis we can also obtain continuity for $\alpha = 1$.

Finally, we note that both the above claims jointly imply by the intermediate value theorem that for any $\beta \in [1, \infty)$ there exists some $\alpha^* \in [0, 1]$ such that $x_1(\alpha^*) = x_1^{\beta}$, thereby proving our result.

\section{Additional Unfairness Measures} \label{apdx:general-unfairness}

While we presented an unfairness notion in Section~\ref{sec:fair-eff-metric} that compares the ratio of the travel times on positive paths, there have been many other unfairness notions that have been investigated in the fair traffic routing literature as well as in other economic applications to evaluate the distribution of a resource among users. In this section, we present a few other commonly-used unfairness measures used in the traffic routing and economics literature (Section~\ref{sec:apdx-new-measures}) and a numerical study of these unfairness measures (Section~\ref{sec:apdx-numerical}) on the six instances in Table~\ref{tab:problem-instances}.

\subsection{Unfairness Measures for Path Flows} \label{sec:apdx-new-measures}

In this section, we present several measures that can be used to evaluate the level of unfairness of a traffic assignment. While many unfairness measures for traffic routing have been proposed, most notions rely on a quantification of the discrepancy between user travel times. For instance,~\citet{correa-fair} measure unfairness through the maximum ratio of the travel times between two users for a given set of path flows $\f$ of the edge flow $\x$. This definition of unfairness is termed as ``envy-free'' by~\citet{basu2017reconciling} since one user can envy another's path by a factor no more than $\beta$. An analogous notion to the envy-free definition is that of a ``Used Nash Equlibrium'', wherein the level of unfairness is calculated through the maximum ratio between the experienced travel time for any given user to the travel time on any other positive path between the same O-D pair. Note that the ``Used Nash Equlibrium'' serves as an intermediary to the unfairness notion presented Section~\ref{sec:fair-eff-metric} and the envy-free notion. There are also several other definitions of max-min fairness proposed in the traffic routing literature, and we defer the interested reader to a discussion of some of these notions to~\citet{so-routing-seminal}.

Beyond the fairness notions considered in the traffic routing literature, there have also been other measures proposed to evaluate fairness or equity in the distribution of a given resource amongst users. One such popular measure is the Gini coefficient~\cite{gini-index}, which is used to evaluate the level of dispersion of wealth in society. Applying this idea in context of traffic routing, we can use the discrete Gini coefficient measure~\cite{WU20121273,jalota-cprr} to evaluate the spread in the travel times of users travelling between the same O-D pair. We summarize the above mentioned unfairness measures in Table~\ref{tab:unfair-measures}. 

Note that some path decompositions may be more fair according to one metric as compared to another. For instance, consider the situation when a small group of users incur an exceedingly large travel time while the remaining users travelling between the same O-D pair incur a small travel time. In this case, the level of unfairness for the Envy Free and Used Nash measures will typically be high but the level of unfairness as measured according to the Gini measure will be quite low. As a result, each of these metrics has their own merit as well as limitations when applied to measuring the level of unfairness of specific path flow decompositions.

\begin{table}[!] 
\centering
\caption{{\small \sf Path-based unfairness Measures that depend on the realization of a specific path flow decomposition $\f$ of an edge flow $\x$. }  }
\small
\begin{tabular}{cc}
\toprule
Unfairness Measure  & Formula ($\Tilde{U}(\f)$) \\
\midrule 
Envy Free & $\max_{k \in K} \max_{Q, R \in \P_k: \x_Q, \x_R>0} \frac{t_Q(\x)}{t_R(\x)} $ \\
Used Nash & $\max_{k \in K}  \max_{R \in \P_k^+, Q\in \P_k: \x_Q>0} \frac{t_Q(\x)}{t_R(\x)} $ \\
Discrete Gini Coefficient & $\max_{k \in K} \frac{\sum_{P, Q \in \P_k} \x_{Q} \x_{P} |t_Q(\x) - t_P(\x)|}{2 d_k \sum_{P \in \P_k} \x_{P} t_{P}(\x)} $ \\
\bottomrule
\end{tabular} \label{tab:unfair-measures}
\end{table}
Finally, observe that each of the above defined notions of unfairness involve taking a maximum of the unfairness measure for each O-D pair. We note that different methods of aggregating across O-D pairs can also be used, e.g., averaging the unfairness measure. Furthermore, there is a key aspect of the above defined unfairness measures that makes the I-TAP a particularly natural approach for balancing efficiency and fairness. In particular, for each of the above defined unfairness notions, the user equilibrium achieves the highest possible level of fairness, while the system optimum achieves the lowest total travel times. Thus, interpolating between the UE and SO objectives can potentially achieve the best of both worlds, i.e., a low total travel time with a high level of fairness.

\subsection{Numerical Results} \label{sec:apdx-numerical}
We now present some numerical results to demonstrate the performance of I-TAP with respect to the above defined path-based unfairness measures. In particular, we use the Frank Wolfe method to compute a specific path decomposition for I-TAP and use this to evaluate the level of unfairness based on the Envy-Free and discrete Gini coefficient metrics. Figures~\ref{fig:vals_updated_envy_free} and~\ref{fig:vals_updated_gini} depict the trade-off between system efficiency and unfairness for the two metrics, Envy-Freeness and the discrete Gini coefficient, respectively.

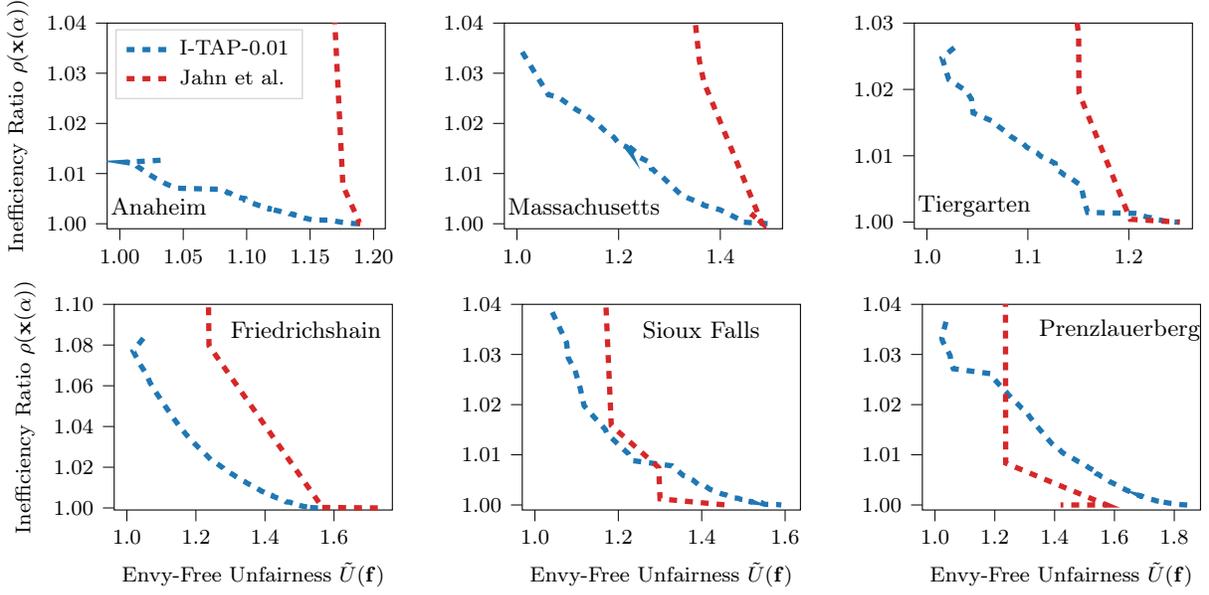
\begin{figure*}
    \centering
    \begin{subfigure}[b]{0.323\columnwidth}
        \centering
\begin{tikzpicture}

\definecolor{color0}{rgb}{0.12156862745098,0.466666666666667,0.705882352941177}
\definecolor{color1}{rgb}{1,0.498039215686275,0.0549019607843137}
\definecolor{color2}{rgb}{0.172549019607843,0.627450980392157,0.172549019607843}
\definecolor{color3}{rgb}{0.83921568627451,0.152941176470588,0.156862745098039}

\begin{axis}[
width = \siwidth,
height = 1.7in,
legend cell align={left},
legend style={
  fill opacity=0.8,
  draw opacity=1,
  text opacity=1,
  at={(0.03,0.97)},
  anchor=north west,
  draw=white!80!black
},
tick align=outside,
tick pos=left,
x grid style={white!69.0196078431373!black},
xmin=0.990035819317997, xmax=1.20928673089475,
xticklabels={0.00, 0.00, 1.00, 1.05, 1.10, 1.15, 1.20},
xtick style={color=black},
y grid style={white!69.0196078431373!black},
ylabel={Inefficiency Ratio \(\displaystyle \rho(\mathbf{x}(\alpha))\)},
ymin=0.999, ymax=1.04,
yticklabels={0.00, 0.00, 1.00, 1.01, 1.02, 1.03, 1.04},
ytick style={color=black},
style={font=\footnotesize}
]
\addplot [dashed, line width=2pt, color=color0]
table {%
1.03216119669047 1.01265315471268
1.0044897566147 1.01225288492082
1.01260734852479 1.01191904507463
1.01379583603029 1.01142381375606
1.01638579106409 1.0110124492958
1.0195643650976 1.01039067771015
1.02276312171845 1.00987786255013
1.02840240378469 1.00902442598697
1.03393768664059 1.00830465071282
1.03607425022471 1.00799248081522
1.04413805909871 1.00708905081794
1.0811721459356 1.00684693861472
1.08340075892904 1.00650449003206
1.08486532206039 1.00617137760173
1.08668665674585 1.00613020702589
1.09042465771605 1.00550156998356
1.09443358247832 1.00519999074778
1.09856213216164 1.00503238003864
1.09868873788211 1.00468393056898
1.1010729555084 1.00446461325952
1.10310868961225 1.00425031887833
1.10488335657065 1.00407368072452
1.10702671395972 1.00387315146655
1.10909450596478 1.0037038361224
1.11088022820249 1.00355003916991
1.11312043762464 1.00337889475555
1.11832150560789 1.00298301190617
1.11957239718466 1.00288441426009
1.11576299843076 1.00305009196332
1.12450214819407 1.00258902194077
1.12612854940984 1.00248056257713
1.12744798632572 1.00236745095086
1.12893258511788 1.002300547913
1.1311955842223 1.0022465208458
1.13216155957712 1.00215474985936
1.13302907757777 1.00208193191694
1.13345357755932 1.00200762116026
1.13627781829183 1.00191806209722
1.13711055209856 1.00185019411029
1.13724681942364 1.00176410377697
1.13932546097928 1.00165964368335
1.14171019107959 1.00148624054797
1.14171292312897 1.00140838904069
1.14442385650706 1.0012506913988
1.14522803298108 1.00117394594353
1.14574102750633 1.00111580041303
1.14686538133259 1.00106149080132
1.1482093785137 1.00099741876134
1.1488659400729 1.00094440473113
1.14965327876036 1.00088451135893
1.15062395280812 1.00083530277679
1.15142091365338 1.00078839081307
1.15226540279908 1.00075312935743
1.16678580372163 1.00071182649804
1.16762903190981 1.00066199777364
1.16828372103866 1.00062028381514
1.1704556182483 1.00049481587456
1.17124394966693 1.00046400788947
1.17270926611697 1.00039228193738
1.17336047256507 1.00036159346821
1.17417311790774 1.00033008696083
1.17479082160598 1.00030876127438
1.17534906643653 1.00028218163202
1.17746072261263 1.00021628392479
1.1781671985861 1.00019634652979
1.17863973074529 1.00018160479142
1.17937033822046 1.00015854319168
1.18097984313844 1.00011095481322
1.18171879975772 1.0000966828893
1.18336391765158 1.00005878547196
1.18387715543773 1.00005023892565
1.18546981957506 1.00002745289667
1.18780607586213 1.00000430615126
1.1885084951836 1.00000205872337
1.18875075894308 1
};
\addlegendentry{I-TAP-0.01}
\addplot [dashed, line width=2pt, color=color3]
table {%
1.00000176984421 1.90599615279674
1.17585465341427 1.00759134147248
1.18766847879678 1.00095753440826
1.18914090901968 1.00001776131243
1.18948116322059 0.999998787384622
1.18917921401213 1.00000246852559
1.1892481556778 1
};
\addlegendentry{Jahn et al.}
\end{axis}
\node[text width=1.5cm] at (0.8,0.3) {\small Anaheim};
\end{tikzpicture}
    \end{subfigure}
    \begin{subfigure}[b]{0.323\columnwidth}
        \centering
\begin{tikzpicture}

\definecolor{color0}{rgb}{0.12156862745098,0.466666666666667,0.705882352941177}
\definecolor{color1}{rgb}{1,0.498039215686275,0.0549019607843137}
\definecolor{color2}{rgb}{0.172549019607843,0.627450980392157,0.172549019607843}
\definecolor{color3}{rgb}{0.83921568627451,0.152941176470588,0.156862745098039}

\begin{axis}[
width = \siwidth,
height = 1.7in,
legend cell align={left},
legend style={fill opacity=0.8, draw opacity=1, text opacity=1, draw=white!80!black},
tick align=outside,
tick pos=left,
x grid style={white!69.0196078431373!black},
xmin=0.97512811431037, xmax=1.52230959948223,
xtick style={color=black},
xticklabels={0.00, 0.00, 1.0, 1.2, 1.4},
y grid style={white!69.0196078431373!black},
ymin=0.999, ymax=1.04,
restrict y to domain=0.999:1.06,
yticklabels={0.00, 0.00, 1.00, 1.01, 1.02, 1.03, 1.04},
ytick style={color=black},
style={font=\footnotesize}
]
\addplot [dashed, line width=2pt, color=color0]
table {%
1.0102572793435 1.03423045286804
1.06158675604029 1.02571243757666
1.08185840707965 1.02522150240824
1.10655186163974 1.0234053265193
1.11898706998626 1.02260234421599
1.13490304002042 1.02175922017397
1.13965646004481 1.02136926088818
1.14535785953177 1.02092627646149
1.15900691522328 1.01985855439678
1.16381814767336 1.01939585750602
1.1753132627482 1.018307262107
1.18474864835685 1.01777432193977
1.19094137876347 1.01717943306762
1.19992179533637 1.0160216536372
1.21614204695219 1.01527929304595
1.22063352222043 1.0146632044624
1.21614053255694 1.01573468782437
1.23713989762536 1.01398380473779
1.24247064651539 1.01316806472034
1.26155499212072 1.01220940170584
1.26446452107007 1.01093862529464
1.26731561266407 1.01085180694545
1.27600417506479 1.00995504949882
1.28249493982 1.00967637699464
1.32414057605043 1.00563350166797
1.36425652279932 1.00397657779287
1.3658437817358 1.00378626098321
1.36777972164666 1.00366953398749
1.39943436400079 1.00278477876433
1.4377449472494 1.00085309683641
1.44026706225733 1.00033734297481
1.49129560747643 1.00002869037902
1.4974377137926 1
};
\addplot [dashed, line width=2pt, color=color3]
table {%
1 11.9228264121261
1.34361081401399 1.04903719183988
1.34890817058833 1.04184483668255
1.35876340064822 1.03251666135641
1.36779255775214 1.02819445998195
1.48164207227087 1.00018223266581
1.46066713877938 1.0019872348476
1.48324585096453 1
};
\end{axis}
\node[text width=1.5cm] at (0.8,0.3) {\small Massachusetts};
\end{tikzpicture}
    \end{subfigure}
    \begin{subfigure}[b]{0.323\columnwidth}
        \centering
\begin{tikzpicture}

\definecolor{color0}{rgb}{0.12156862745098,0.466666666666667,0.705882352941177}
\definecolor{color1}{rgb}{1,0.498039215686275,0.0549019607843137}
\definecolor{color2}{rgb}{0.172549019607843,0.627450980392157,0.172549019607843}
\definecolor{color3}{rgb}{0.83921568627451,0.152941176470588,0.156862745098039}

\begin{axis}[
width = \siwidth,
height = 1.7in,
legend cell align={left},
legend style={fill opacity=0.8, draw opacity=1, text opacity=1, draw=white!80!black},
tick align=outside,
tick pos=left,
x grid style={white!69.0196078431373!black},
xmin=0.987473997719711, xmax=1.26304604788608,
xticklabels={0.00, 0.00, 1.0, 1.1, 1.2},
xtick style={color=black},
y grid style={white!69.0196078431373!black},
ymin=0.999, ymax=1.03,
yticklabels={0.00, 0.00, 1.00, 1.01, 1.02, 1.03},
ytick style={color=black},
style={font=\footnotesize}
]
\addplot [dashed, line width=2pt, color=color0]
table {%
1.02709771362984 1.02629964809451
1.01501092579337 1.02476463815947
1.02163661668061 1.02149320788889
1.03920897531093 1.01974910394532
1.04417999183975 1.01851943548177
1.04575892223503 1.01646103437012
1.06472184800697 1.0152377929048
1.0714251538072 1.01435601081908
1.07701418239315 1.01372989751561
1.08280579453186 1.01294362273387
1.08932864167602 1.01230335610018
1.09514445590111 1.01184637913861
1.09869978534192 1.01130984303842
1.10371348536129 1.01085796949852
1.10888525286956 1.01046757068132
1.11296364942916 1.00991163212425
1.11849496246563 1.00954035689957
1.12698322827414 1.00885665604232
1.12712504246939 1.00832659298915
1.1272385195253 1.00796476103077
1.13297986718668 1.00751132107223
1.1358245167617 1.00713048998565
1.1424960512246 1.00671667777246
1.14401849403861 1.00646575147914
1.14900317487466 1.0059153831851
1.15258362964901 1.00555316266743
1.15647819077166 1.00283248031482
1.15927304127252 1.00145522669928
1.20729323916325 1.00131054262375
1.21028393071082 1.00125872990354
1.21113489881462 1.00109541998315
1.21194606747654 1.00103995637995
1.21494065785012 1.00090545852646
1.21626319112678 1.00084393735486
1.2171942626905 1.0007932085936
1.21845816905592 1.00073766823008
1.22212342894047 1.00066178583271
1.22300226427117 1.00059646879999
1.2241248929463 1.00051526056854
1.22486422178514 1.00046559594146
1.22672905729286 1.00044251076002
1.22864732885857 1.00038090498326
1.22916156831712 1.00033789528396
1.23074298329181 1.00030660245646
1.23201259908945 1.00028711389123
1.23297643164075 1.00022828163916
1.23461373237738 1.00020113918006
1.23731958715417 1.00013713530775
1.23890492305139 1.00009958653318
1.24031109726405 1.00008929574864
1.24036834877736 1.00007492934102
1.2422835903351 1.00005442225699
1.24299614449695 1.00004473523393
1.24353628795276 1.00003546167301
1.24448941880417 1.00002885212432
1.24512949778629 1.00002383046016
1.24569516008279 1.00001788281691
1.247634435561 1.00001135450937
1.2475504476628 1.00001145674198
1.24981628419672 1.00000339303689
1.25031639679881 1
};
\addplot [dashed, line width=2pt, color=color3]
table {%
1 1.24625684755394
1.11785597262373 1.07113824546842
1.15032749803855 1.02809181860033
1.15063965446512 1.01944791904243
1.20148930554737 1.00042647951916
1.25022564856515 1.0000040274766
1.25052004560579 1
1.25052004560579 1
};
\draw (axis cs:1.5,0.9) node[
  scale=0.5,
  text=black,
  rotate=0.0
]{Tiergarten};
\end{axis}
\node[text width=1.5cm] at (0.8,0.3) {\small Tiergarten};
\end{tikzpicture}
    \end{subfigure}\par\medskip
    \vspace{-20pt}
    \hspace{0pt}
    \begin{subfigure}[b]{0.323\columnwidth}
        \centering
\begin{tikzpicture}

\definecolor{color0}{rgb}{0.12156862745098,0.466666666666667,0.705882352941177}
\definecolor{color1}{rgb}{1,0.498039215686275,0.0549019607843137}
\definecolor{color2}{rgb}{0.172549019607843,0.627450980392157,0.172549019607843}
\definecolor{color3}{rgb}{0.83921568627451,0.152941176470588,0.156862745098039}

\begin{axis}[
width = \siwidth,
height = 1.7in,
legend cell align={left},
legend style={fill opacity=0.8, draw opacity=1, text opacity=1, draw=white!80!black},
tick align=outside,
tick pos=left,
x grid style={white!69.0196078431373!black},
xlabel={Envy-Free Unfairness $\Tilde{U}(\f)$},
xmin=0.96338951825704, xmax=1.76882011660215,
xticklabels={0.00, 0.00, 1.0, 1.2, 1.4, 1.6},
xtick style={color=black},
y grid style={white!69.0196078431373!black},
ylabel={Inefficiency Ratio \(\displaystyle \rho(\mathbf{x}(\alpha))\)},
ymin=0.999, ymax=1.1,
yticklabels={0.00, 0.00, 1.00, 1.02, 1.04, 1.06, 1.08, 1.10},
ytick style={color=black},
style={font=\footnotesize}
]
\addplot [dashed, line width=2pt, color=color0]
table {%
1.04777619741354 1.08333681641639
1.01805202659178 1.07710131262555
1.03856957576225 1.07046933966396
1.05621547295758 1.06580191568301
1.06739136023125 1.06145508118298
1.08261267704885 1.05723137287348
1.09749662468401 1.05345783131864
1.111158570345 1.0499776990804
1.1257605120788 1.04637820942525
1.13765767631152 1.04348819399355
1.15073367543052 1.04044064006127
1.17409443530752 1.03561785030523
1.18540518922701 1.03351829670561
1.1957438278156 1.03170437811372
1.20654009047451 1.02985388900329
1.21649517029881 1.02821424641006
1.22525243967528 1.02688051215448
1.24408863119237 1.02353305734928
1.25378198558326 1.02280015379761
1.26177654679862 1.02171169256114
1.27061851020754 1.0206113870703
1.27878549059673 1.0195973868961
1.28222964995237 1.01889486334355
1.29422500292896 1.01778353516114
1.30161631721178 1.01695653933158
1.30936485066497 1.01618174040173
1.3157398834171 1.01541563528172
1.32244638216431 1.01470379877342
1.32903083379994 1.01400226329925
1.33548356433071 1.01339094478428
1.34159886710222 1.01289871201323
1.34776369223884 1.01222672934129
1.35256289514385 1.01178669562625
1.35940012648589 1.01119051825643
1.3647525061162 1.01068835281413
1.3723945023135 1.01032792955752
1.37572277012039 1.0096877769556
1.38064604179561 1.00920442853723
1.38592842834554 1.00875717194621
1.38863075654231 1.00831839941908
1.39607844759622 1.00785560538963
1.40056897402795 1.00745992640823
1.40464279861951 1.00696927852201
1.40897113167956 1.00671355064748
1.4081215725295 1.0068390950926
1.42119791005106 1.00610423029478
1.42231346414045 1.00578764161838
1.42639394291261 1.0055067284872
1.43052192470923 1.00524900506612
1.43428496571041 1.00498803927669
1.43813751057205 1.00477510445006
1.44183173030671 1.00457549064643
1.44400058400344 1.00434688224529
1.44901423088906 1.00409613521423
1.45252238397262 1.00390571396136
1.45576110565408 1.0037241267993
1.45718289429052 1.00357014227175
1.46121782868752 1.00297916712246
1.47209345623178 1.00284059630784
1.47490128944075 1.00263125951364
1.47811924918983 1.0024477975019
1.48101690976591 1.00228771143831
1.48432926192216 1.00207101284246
1.48664075358742 1.00196770131731
1.48957064740802 1.00179804116805
1.49198569127231 1.00165366483195
1.49483594230124 1.00153126284474
1.49346541729261 1.00159613416909
1.49789277508697 1.00141491115098
1.50365105099353 1.00109218719726
1.50655370564991 1.00100424987227
1.5094254104892 1.00086276299035
1.51028534459979 1.00078538061107
1.51739787535417 1.00061393612214
1.51896747794997 1.00056171675806
1.52663405852932 1.00036200856258
1.51957659664831 1.0003695032153
1.53230159318027 1.00020363255038
1.53513701275818 1.00019133888759
1.53694358197204 1.00017015667451
1.53878073347556 1.00014885590579
1.54052131136049 1.00013050425699
1.54262431739639 1.00009965487693
1.54432292523624 1.0000844240366
1.54634023237672 1.00005568724939
1.55397178591997 1.00001047563676
1.55799302060645 1
};
\addplot [dashed, line width=2pt, color=color3]
table {%
1 1.51355675945567
1.02176255267465 1.18592437282101
1.23181669748365 1.16590656115161
1.23786559835767 1.07971011751014
1.55772539482382 1.00231949993955
1.55780963667575 1.00025161441859
1.73220963485919 1
};
\end{axis}
\node[text width=1.5cm] at (2.3,2.4) {\small Friedrichshain};
\end{tikzpicture}
    \end{subfigure}\hspace{8pt}
    \begin{subfigure}[b]{0.323\columnwidth}
        \centering
\begin{tikzpicture}

\definecolor{color0}{rgb}{0.12156862745098,0.466666666666667,0.705882352941177}
\definecolor{color1}{rgb}{1,0.498039215686275,0.0549019607843137}
\definecolor{color2}{rgb}{0.172549019607843,0.627450980392157,0.172549019607843}
\definecolor{color3}{rgb}{0.83921568627451,0.152941176470588,0.156862745098039}

\begin{axis}[
width = \siwidth,
height = 1.7in,
legend cell align={left},
legend style={fill opacity=0.8, draw opacity=1, text opacity=1, draw=white!80!black},
tick align=outside,
tick pos=left,
x grid style={white!69.0196078431373!black},
xlabel={Envy-Free Unfairness $\Tilde{U}(\f)$},
xmin=0.96967947835315, xmax=1.63674378150016,
xticklabels={0.00, 0.00, 1.0, 1.2, 1.4, 1.6},
xtick style={color=black},
y grid style={white!69.0196078431373!black},
ymin=0.999, ymax=1.04,
yticklabels={0.00, 0.00, 1.00, 1.01, 1.02, 1.03, 1.04},
ytick style={color=black},
style={font=\footnotesize}
]
\addplot [dashed, line width=2pt, color=color0]
table {%
1.0415798961255 1.03842126315122
1.07576637026551 1.032385835623
1.08007789074929 1.02930208526587
1.09631464870632 1.02683345004132
1.10865263706428 1.02327474195923
1.11795290580365 1.01978487884171
1.13125068929408 1.01849543729173
1.16881357208051 1.01504965589419
1.17239742617107 1.01427346251606
1.19896661950175 1.01214513408235
1.22466422763239 1.0101790973289
1.23276469068346 1.00886189012525
1.32949096403207 1.00778616681752
1.35524879895757 1.00590006741294
1.36737208496817 1.00556887791138
1.37303459904537 1.00517062262808
1.37945605002522 1.00493660427735
1.39242057305898 1.0040357445566
1.40894807070717 1.00376197904945
1.41947797920481 1.00313676400093
1.42898305752734 1.00283657839387
1.43503204874575 1.00264481939525
1.43745778233635 1.00255411075781
1.4482867333244 1.0022000421102
1.4583722094363 1.00201006210967
1.46383799994945 1.0018166843505
1.46582095770959 1.00181070943076
1.47225033209451 1.00165644994654
1.47720310046831 1.00139529387857
1.48741095489382 1.00133933067813
1.49315318439882 1.00113356540961
1.50311876096558 1.00102698289222
1.51065704251881 1.00095245104715
1.51914889718233 1.00065026097793
1.5131317066045 1.00070788620207
1.52638666160683 1.00059863455539
1.53441255139926 1.00056497508452
1.54386108444038 1.0004404439087
1.54868286665003 1.00019581288586
1.55170182456946 1.00014197231009
1.56227253811229 1.00010832409044
1.57942622925095 1.00009872043471
1.59088103480664 1
};
\addplot [dashed, line width=2pt, color=color3]
table {%
1.00000058304165 8.52174813028758
1.00965980530861 8.46919433999255
1.01624640160335 5.34444789301628
1.02995503944954 3.0176163816457
1.1027181808789 1.87069600386083
1.11597088236877 1.14125963093868
1.18312237580067 1.01585642119357
1.29681347394932 1.00746786362138
1.30006096789984 1.00121828594962
1.45769176085875 1
};
\end{axis}
\node[text width=2cm] at (2.6,2.4) {\small Sioux Falls};
\end{tikzpicture}
    \end{subfigure}\hspace{-5pt}
    \begin{subfigure}[b]{0.323\columnwidth}
        \centering
\begin{tikzpicture}

\definecolor{color0}{rgb}{0.12156862745098,0.466666666666667,0.705882352941177}
\definecolor{color1}{rgb}{1,0.498039215686275,0.0549019607843137}
\definecolor{color2}{rgb}{0.172549019607843,0.627450980392157,0.172549019607843}
\definecolor{color3}{rgb}{0.83921568627451,0.152941176470588,0.156862745098039}

\begin{axis}[
width = \siwidth,
height = 1.7in,
legend cell align={left},
legend style={fill opacity=0.8, draw opacity=1, text opacity=1, draw=white!80!black},
tick align=outside,
tick pos=left,
x grid style={white!69.0196078431373!black},
xlabel={Envy-Free Unfairness $\Tilde{U}(\f)$},
xmin=0.957727873437494, xmax=1.88771465781263,
xticklabels={0.0, 0.0, 1.0, 1.2, 1.4, 1.6, 1.8},
xtick style={color=black},
y grid style={white!69.0196078431373!black},
ymin=0.999, ymax=1.04,
yticklabels={0.00, 0.00, 1.00, 1.01, 1.02, 1.03, 1.04},
ytick style={color=black},
style={font=\footnotesize}
]
\addplot [dashed, line width=2pt, color=color0]
table {%
1.03672245291724 1.03652424176196
1.0226344847168 1.03324752465105
1.02550179776974 1.03262933652619
1.04287080855659 1.03071210147687
1.05076227038471 1.02975565441929
1.05388727071992 1.02815996213305
1.06119700607259 1.02715623344039
1.18628935642275 1.02616883558966
1.1967912039846 1.02519872420061
1.21395587870278 1.02403422424142
1.23796790885949 1.02254556066436
1.24537947260722 1.02198488741122
1.26033373518824 1.02106832114016
1.2767270448256 1.02012271531014
1.29039888398771 1.01926684665749
1.30500702749838 1.01858634455729
1.31737457556133 1.01751915577269
1.33057673564651 1.01651136355009
1.34368011810848 1.01551871509716
1.35663634143437 1.01460529098159
1.3681592320483 1.01373857236214
1.38070454904429 1.01291032884821
1.39230905047826 1.01213488144674
1.40438803836272 1.01156160047472
1.41533613168528 1.01110466467822
1.425975934206 1.01044354105206
1.43663419499817 1.01006084821827
1.44682964565156 1.00971249855662
1.45762548231361 1.00939516103938
1.46741657602615 1.00906467673361
1.47694577100812 1.008715617981
1.49076795080293 1.0083475629617
1.4956413097021 1.00815330957918
1.50461904869886 1.00786678565653
1.51386749588823 1.00744967756313
1.53805306958992 1.00666200204941
1.53111300536628 1.00674783010233
1.53922884253374 1.00624802896109
1.54757725236952 1.00606394697898
1.55435289610883 1.00573750127925
1.56605753078536 1.00540825624147
1.57157095627726 1.00517277741572
1.57916776911221 1.00490673359551
1.58572771075022 1.0046056303769
1.59375831823923 1.00438481945924
1.6065131196664 1.00400763174575
1.6081559088087 1.00390491354557
1.61516803979347 1.00368551477299
1.62210347703546 1.00347661546136
1.62870805659832 1.00329154763578
1.63543666623415 1.00309926952468
1.64245871468799 1.00288760884354
1.64757384713466 1.00276240287848
1.6585006939462 1.00252751203415
1.66111949639228 1.00242746407034
1.67076106669155 1.00224130578431
1.67191069770475 1.00214305176096
1.67745164444213 1.00200457846405
1.69040935023902 1.00171105516087
1.68872682281297 1.00176827595284
1.69522925412295 1.0016807315389
1.69990499073641 1.0015489011978
1.70503389001784 1.00144411066528
1.71044075716925 1.00135360369856
1.71547390222687 1.0012577712703
1.72127914579171 1.00114711776424
1.72532145913196 1.00109146383828
1.73053553677932 1.00101237816411
1.73459477204077 1.00095002859474
1.73931791357628 1.00082213108646
1.74943615581081 1.00074290982962
1.75189420260262 1.00073619637768
1.75815161279011 1.00063403376162
1.76270878449534 1.00058844853654
1.76711023325426 1.00054922396839
1.76210366503677 1.00060075444421
1.77589082131155 1.00044357054071
1.77959337551325 1.00040420027332
1.7840882476894 1.00037176451859
1.78647076475064 1.00034056705524
1.79396614573188 1.00029462989122
1.79564229911913 1.00025012698988
1.80052617131768 1.00022481105851
1.80340139370984 1.00019236244154
1.80511885674786 1.00016591744845
1.81108400933928 1.00013170352962
1.8137377847152 1.00012123266896
1.81775971141073 1.0000992220329
1.82209282795821 1.00007007770955
1.82835601887225 1.00004915349003
1.83525901128958 1.00002021991049
1.8391186241206 1.00001311574237
1.84214251019417 1.00001264492894
1.84544253125013 1
};
\addplot [dashed, line width=2pt, color=color3]
table {%
1 5.62867155729629
1.2280164862816 2.82784066350701
1.22834009750098 2.39480344356529
1.23084411732849 1.42967488182058
1.23488796641289 1.07833119177006
1.23525832830698 1.05908882652385
1.23618046730401 1.01461690354873
1.23636106761171 1.00834160341115
1.57447257397045 0.99999859355778
1.41993571208141 1
};
\end{axis}
\node[text width=1.5cm] at (2.3,2.4) {\small Prenzlauerberg};
\end{tikzpicture}
    \end{subfigure}
    \vspace{-15pt}
    \caption{{\small \sf Pareto frontier depicting the trade-off between efficiency and ``envy-freeness'' for the 
    (i) I-TAP method with a step size $s = 0.01$, and (ii) Jahn et al.'s method with $s=0.05$.}} 
    \label{fig:vals_updated_envy_free}
\end{figure*}

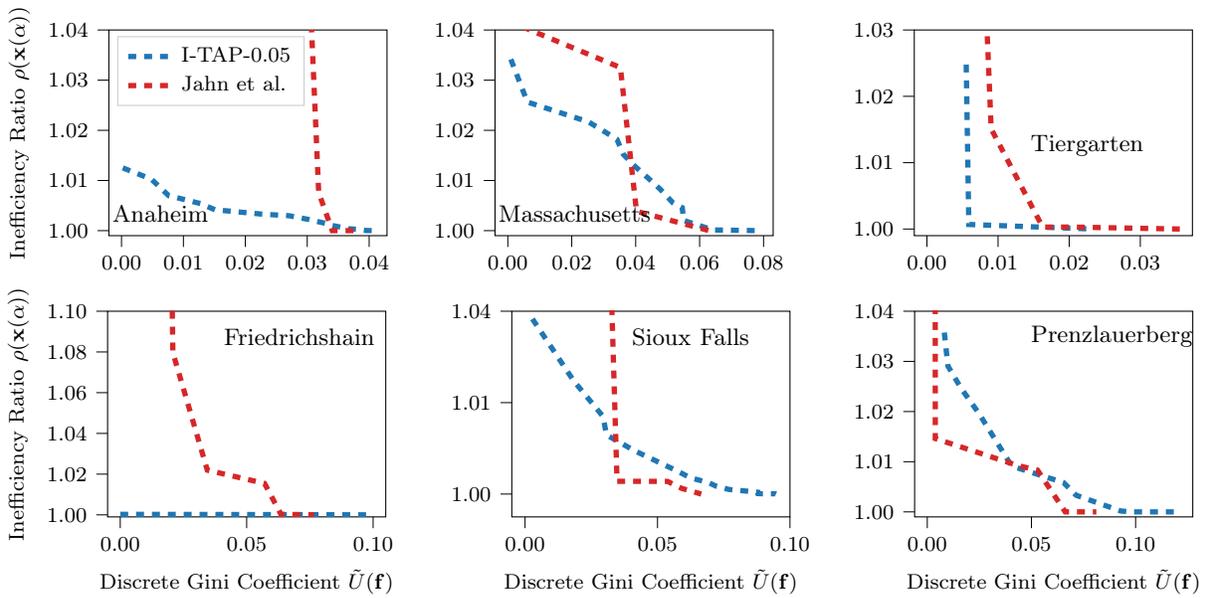
\begin{figure*}
    \centering
    \begin{subfigure}[b]{0.323\columnwidth}
        \centering
\begin{tikzpicture}

\definecolor{color0}{rgb}{0.12156862745098,0.466666666666667,0.705882352941177}
\definecolor{color1}{rgb}{1,0.498039215686275,0.0549019607843137}
\definecolor{color2}{rgb}{0.172549019607843,0.627450980392157,0.172549019607843}
\definecolor{color3}{rgb}{0.83921568627451,0.152941176470588,0.156862745098039}

\begin{axis}[
width = \siwidth,
height = 1.7in,
legend cell align={left},
legend style={
  fill opacity=0.8,
  draw opacity=1,
  text opacity=1,
  at={(0.03,0.97)},
  anchor=north west,
  draw=white!80!black
},
tick align=outside,
tick pos=left,
x grid style={white!69.0196078431373!black},
xmin=-0.00204175132817059, xmax=0.0428864969209217,
xticklabels={0.00, 0.00, 0.01, 0.02, 0.03, 0.04},
xtick style={color=black},
y grid style={white!69.0196078431373!black},
ylabel={Inefficiency Ratio \(\displaystyle \rho(\mathbf{x}(\alpha))\)},
ymin=0.999, ymax=1.04,
yticklabels={0.00, 0.00, 1.00, 1.01, 1.02, 1.03, 1.04},
ytick style={color=black},
style={font=\footnotesize}
]
\addplot [dashed, line width=2pt, color=color0]
table {%
0.000207525266860167 1.01250767697042
0.00482947714956234 1.01025854375383
0.00761605852792852 1.00698212287757
0.0127389063186292 1.00540571425179
0.0150921393158084 1.00416221425264
0.0226782831429586 1.00329546868477
0.0270502247018398 1.00297627709161
0.0302621973257367 1.00222153417545
0.0315602246154931 1.00184170279742
0.033050710830598 1.00142302371727
0.0338808170584262 1.00106017744469
0.0350398504464269 1.00078470973894
0.0359127463802999 1.00057302840763
0.0366630645526584 1.00038687243162
0.0380682160276109 1.00024992967375
0.0394727748470348 1.00002793831495
0.0408443038186902 1
};
\addlegendentry{I-TAP-0.05}
\addplot [dashed, line width=2pt, color=color3]
table {%
0.0000004177 1.90599846403979
0.0318514499759228 1.00759256329472
0.0340404365188011 1.00000260875856
0.0367622335113123 1.00000121261685
0.0387028060983835 1
};
\addlegendentry{Jahn et al.}
\end{axis}
\node[text width=1.5cm] at (0.8,0.3) {\small Anaheim};
\end{tikzpicture}
    \end{subfigure}
    \begin{subfigure}[b]{0.323\columnwidth}
        \centering
\begin{tikzpicture}

\definecolor{color0}{rgb}{0.12156862745098,0.466666666666667,0.705882352941177}
\definecolor{color1}{rgb}{1,0.498039215686275,0.0549019607843137}
\definecolor{color2}{rgb}{0.172549019607843,0.627450980392157,0.172549019607843}
\definecolor{color3}{rgb}{0.83921568627451,0.152941176470588,0.156862745098039}

\begin{axis}[
width = \siwidth,
height = 1.7in,
legend cell align={left},
legend style={fill opacity=0.8, draw opacity=1, text opacity=1, draw=white!80!black},
tick align=outside,
tick pos=left,
x grid style={white!69.0196078431373!black},
xmin=-0.004, xmax=0.083,
xtick style={color=black},
xticklabels={0.00, 0.00, 0.02, 0.04, 0.06, 0.08},
y grid style={white!69.0196078431373!black},
ymin=0.999, ymax=1.04,
restrict y to domain=0.999:1.06,
yticklabels={0.00, 0.00, 1.00, 1.01, 1.02, 1.03, 1.04},
ytick style={color=black},
style={font=\footnotesize}
]
\addplot [dashed, line width=2pt, color=color0]
table {%
0.000947727825637959 1.03415104739086
0.00611902781110721 1.02563761914634
0.0252205275367306 1.02168908337573
0.0340980988095859 1.01824079246336
0.0361400701825563 1.01521216442623
0.0392813019705228 1.01310173961343
0.0449191817360341 1.00988763164429
0.0477167501608073 1.0082282068902
0.051614579775064 1.00556769695321
0.0545725428006892 1.00457167967715
0.0550098802360843 1.00202674847829
0.064198478842727 1.00013837971253
0.0788118553846344 1
};
\addplot [dashed, line width=2pt, color=color3]
table {%
0 1.04184483668255
0.0352122108921305 1.03251666135641
0.040255723610816 1.00384905065396
0.0630614711565026 1
};
\end{axis}
\node[text width=1.5cm] at (0.8,0.3) {\small Massachusetts};
\end{tikzpicture}
    \end{subfigure}
    \begin{subfigure}[b]{0.323\columnwidth}
        \centering
\begin{tikzpicture}

\definecolor{color0}{rgb}{0.12156862745098,0.466666666666667,0.705882352941177}
\definecolor{color1}{rgb}{1,0.498039215686275,0.0549019607843137}
\definecolor{color2}{rgb}{0.172549019607843,0.627450980392157,0.172549019607843}
\definecolor{color3}{rgb}{0.83921568627451,0.152941176470588,0.156862745098039}

\begin{axis}[
width = \siwidth,
height = 1.7in,
legend cell align={left},
legend style={fill opacity=0.8, draw opacity=1, text opacity=1, draw=white!80!black},
tick align=outside,
tick pos=left,
x grid style={white!69.0196078431373!black},
xmin=-0.00177749753399089, xmax=0.0373274482138088,
xticklabels={0.00, 0.00, 0.01, 0.02, 0.03},
xtick style={color=black},
y grid style={white!69.0196078431373!black},
ymin=0.999, ymax=1.03,
yticklabels={0.00, 0.00, 1.00, 1.01, 1.02, 1.03},
ytick style={color=black},
yticklabel style = {scaled y ticks=false},
style={font=\footnotesize}
]
\addplot [dashed, line width=2pt, color=color0]
table {%
0.00552565570748857 1.02478812979544
0.00589162225633511 1.00068861695954
0.0233433575048311 1
};
\addplot [dashed, line width=2pt, color=color3]
table {%
0 1.24625684755394
0.00903332410902836 1.01510091070449
0.016267715705469 1.00029198043607
0.0207801228407467 1.0002919015365
0.0355499506798179 1
0.0355499506798179 1
};
\end{axis}
\node[text width=1.5cm] at (2.3,1.2) {\small Tiergarten};
\end{tikzpicture}
    \end{subfigure}\par\medskip
    \vspace{-20pt}
    \hspace{0pt}
    \begin{subfigure}[b]{0.323\columnwidth}
        \centering
\begin{tikzpicture}

\definecolor{color0}{rgb}{0.12156862745098,0.466666666666667,0.705882352941177}
\definecolor{color1}{rgb}{1,0.498039215686275,0.0549019607843137}
\definecolor{color2}{rgb}{0.172549019607843,0.627450980392157,0.172549019607843}
\definecolor{color3}{rgb}{0.83921568627451,0.152941176470588,0.156862745098039}

\begin{axis}[
width = \siwidth,
height = 1.7in,
legend cell align={left},
legend style={fill opacity=0.8, draw opacity=1, text opacity=1, draw=white!80!black},
tick align=outside,
tick pos=left,
x grid style={white!69.0196078431373!black},
xlabel={Discrete Gini Coefficient $\Tilde{U}(\f)$},
xmin=-0.005, xmax=0.105,
xticklabels={0.00, 0.00, 0.05, 0.10},
xtick style={color=black},
y grid style={white!69.0196078431373!black},
ylabel={Inefficiency Ratio \(\displaystyle \rho(\mathbf{x}(\alpha))\)},
ymin=0.999, ymax=1.1,
yticklabels={0.00, 0.00, 1.00, 1.02, 1.04, 1.06, 1.08, 1.10},
ytick style={color=black},
style={font=\footnotesize}
]
\addplot [dashed, line width=2pt, color=color0]
table {%
0 1.00025723899804
0.0699258085322619 1.00011010832599
0.0999203271887767 1
};
\addplot [dashed, line width=2pt, color=color3]
table {%
0 1.51355675945567
0 1.48408912474145
0.00537615874751643 1.18592437282101
0.0198290940600971 1.16590656115161
0.0208096798411033 1.07971011751014
0.0344308522431447 1.0219458553382
0.0570613267376704 1.01545066428188
0.0635542846906807 1.00025161441859
0.0766730922612522 1
};
\end{axis}
\node[text width=1.5cm] at (2.3,2.4) {\small Friedrichshain};
\end{tikzpicture}
    \end{subfigure}\hspace{8pt}
    \begin{subfigure}[b]{0.323\columnwidth}
        \centering
\begin{tikzpicture}

\definecolor{color0}{rgb}{0.12156862745098,0.466666666666667,0.705882352941177}
\definecolor{color1}{rgb}{1,0.498039215686275,0.0549019607843137}
\definecolor{color2}{rgb}{0.172549019607843,0.627450980392157,0.172549019607843}
\definecolor{color3}{rgb}{0.83921568627451,0.152941176470588,0.156862745098039}

\begin{axis}[
width = \siwidth,
height = 1.7in,
legend cell align={left},
legend style={fill opacity=0.8, draw opacity=1, text opacity=1, draw=white!80!black},
tick align=outside,
tick pos=left,
x grid style={white!69.0196078431373!black},
xlabel={Discrete Gini Coefficient $\Tilde{U}(\f)$},
xmin=-0.005, xmax=0.1,
xticklabels={0.00, 0.00, 0.05, 0.10},
xtick style={color=black},
y grid style={white!69.0196078431373!black},
ymin=0.995, ymax=1.04,
yticklabels={0.00, 0.00, 1.00, 1.01, 1.04},
ytick style={color=black},
style={font=\footnotesize}
]
\addplot [dashed, line width=2pt, color=color0]
table {%
0.00278719740524245 1.03831401210849
0.0183595975215492 1.02477755913583
0.0294289880263073 1.01697634891823
0.0307402488357417 1.01293488690625
0.0413892223390185 1.00945953291084
0.058006314026618 1.00488206066092
0.062308451972296 1.00356367887171
0.0692607892888724 1.0025063135943
0.0719309983587558 1.00177166746141
0.0769595979083775 1.00100443030013
0.0830958673771238 1.00066912628689
0.0878300956156301 1.00040723602053
0.0879515839413494 1.00010519259326
0.0896608916200845 1.00006249729303
0.0944901665693773 1.00006170142505
0.0947216530062414 1
};
\addplot [dashed, line width=2pt, color=color3]
table {%
0 8.52174813028758
0.00035172502641255 8.46919433999255
0.000736495667136006 5.34444789301628
0.0064020994358792 3.0176163816457
0.0141005405407945 1.61127405051206
0.0181163375835599 1.18592479698335
0.0231432509266176 1.14125963093868
0.0305771714490453 1.08269887610134
0.0347279335017648 1.00279000458574
0.0537655237420207 1.00277023970412
0.0592372006990233 1.00121828594962
0.0665173781595905 1
};
\end{axis}
\node[text width=2cm] at (2.6,2.4) {\small Sioux Falls};
\end{tikzpicture}
    \end{subfigure}\hspace{-5pt}
    \begin{subfigure}[b]{0.323\columnwidth}
        \centering
\begin{tikzpicture}

\definecolor{color0}{rgb}{0.12156862745098,0.466666666666667,0.705882352941177}
\definecolor{color1}{rgb}{1,0.498039215686275,0.0549019607843137}
\definecolor{color2}{rgb}{0.172549019607843,0.627450980392157,0.172549019607843}
\definecolor{color3}{rgb}{0.83921568627451,0.152941176470588,0.156862745098039}

\begin{axis}[
width = \siwidth,
height = 1.7in,
legend cell align={left},
legend style={fill opacity=0.8, draw opacity=1, text opacity=1, draw=white!80!black},
tick align=outside,
tick pos=left,
x grid style={white!69.0196078431373!black},
xlabel={Discrete Gini Coefficient $\Tilde{U}(\f)$},
xmin=-0.00606850963430577, xmax=0.127438702320421,
xticklabels={0.00, 0.00, 0.05, 0.10},
xtick style={color=black},
y grid style={white!69.0196078431373!black},
ymin=0.999, ymax=1.04,
yticklabels={0.00, 0.00, 1.00, 1.01, 1.02, 1.03, 1.04},
ytick style={color=black},
style={font=\footnotesize}
]
\addplot [dashed, line width=2pt, color=color0]
table {%
0.00817827285838573 1.03574965919486
0.00991543641846255 1.02908732260394
0.0161805178279814 1.02467431411308
0.0247126961236807 1.01968164146827
0.0405111963913554 1.00913005264855
0.0653932644235734 1.00584561786713
0.0714497502840285 1.00331316685426
0.0928845221645388 1.00017515026314
0.098198672112231 1.000001644158
0.121370192686115 1
};
\addplot [dashed, line width=2pt, color=color3]
table {%
0 5.62867947370875
0.0038495162353235 1.01461833055078
0.043926838313763 1.00917607647704
0.0528808842237193 1.00834302158734
0.0662180475095875 1.0000014064442
0.0811587130996032 1
};
\end{axis}
\node[text width=1.5cm] at (2.3,2.4) {\small Prenzlauerberg};
\end{tikzpicture}
    \end{subfigure}
    \vspace{-15pt}
    \caption{{\small \sf Pareto frontier depicting the trade-off between efficiency and discrete Gini coefficient for the 
    (i) I-TAP method with a step size $s = 0.05$, and (ii) Jahn et al.'s method with $s=0.05$.}} 
    \label{fig:vals_updated_gini}
\end{figure*}

We note that the Pareto-frontier depicting the Envy-Free metric in Figure~\ref{fig:vals_updated_envy_free} is almost identical to the Pareto frontier of the edge-based unfairness metric presented in Section~\ref{sec:fair-eff-metric}. This observation suggests that for many practical instances these two notions of unfairness are quite close to each other. We note this despite the fact that the Envy-Free notion involves the maximum ratio between experienced travel times for a given flow decomposition while the edge-based unfairness metric involves the maximum possible ratio between travel times on positive paths and is thus independent of the realized path decomposition. Since the Used Nash measure is an intermediary between the Envy-Free and the edge-based unfairness measure in Section~\ref{sec:fair-eff-metric}, we note that it would also correspond to near identical Pareto frontiers for the six tested instances. Thus, we do not depict the Used Nash unfairness measure.


For the discrete Gini coefficient metric, we first note that at the user equilibrium solution the discrete Gini coefficient is zero, rather than having unfairness of one for the earlier explored unfairness measures. This is because all users travelling between the same O-D pair have equal travel times at the user equilibrium and thus the numerator of the discrete Gini coefficient is 0. Note that the further the discrete Gini coefficient is away from zero the more unfair the solution.

We observe from Figure~\ref{fig:vals_updated_gini} that, modulo numerical errors, the discrete Gini coefficient of the user equilibrium for all the scenarios is very close to zero. As with the earlier explored unfairness measures, we also observe that the I-TAP method outperforms the approach used in~\citet{so-routing-seminal} even on the discrete Gini coefficient metric for low levels of desired Gini coefficients for all data-sets but Prenzlauerberg.

\end{document}